\RequirePackage[l2tabu, orthodox]{nag}		
\documentclass[10pt]{article}

\usepackage[T1]{fontenc}
\usepackage{lmodern}
\usepackage[frenchmath]{mathastext}							
\usepackage{natbib}       									
\usepackage{amsmath}                        
\numberwithin{equation}{section}
\usepackage{graphicx} 											
\usepackage{enumitem}												
\usepackage{mdwlist}												
\usepackage[dvipsnames]{xcolor}							
\usepackage[plainpages=false, pdfpagelabels]{hyperref} 
	\hypersetup{
		colorlinks   = true,
		citecolor    = RoyalBlue, 
		linkcolor    = RubineRed, 
		urlcolor     = Turquoise
	}
\usepackage{amssymb}                        
\usepackage[mathscr]{eucal}                 
\usepackage{dsfont}													
\usepackage[paperwidth=8.5in,paperheight=11in,top=1.25in, bottom=1.25in, left=1.0in, right=1.0in]{geometry} 
\usepackage{mathtools}                      
\mathtoolsset{showonlyrefs=true}            
\usepackage{amsthm}                         
	\allowdisplaybreaks                       
	\theoremstyle{plain}
	\newtheorem{theorem}{Theorem}
	\numberwithin{theorem}{section}
	\newtheorem{proposition}[theorem]{Proposition}
	
	\theoremstyle{definition}

\newcommand\Eb{\mathds{E}}

\newcommand\Rb{\mathds{R}}

\newcommand\Ac{\mathscr{A}}

\newcommand\Nc{\mathscr{N}}

\renewcommand\d{\partial}

\newcommand\dd{\mathrm{d}}
\newcommand\ee{\mathrm{e}}

\begin{document}

\title{Optimal Trading with Differing Trade Signals}

\author{
Ryan Donnelly
\thanks{Department of Mathematics, King's College London.  \textbf{e-mail}: \url{ryan.f.donnelly@kcl.ac.uk}}
\and
Matthew Lorig
\thanks{Department of Applied Mathematics, University of Washington.  \textbf{e-mail}: \url{mlorig@uw.edu}}
}

\date{This version: \today}

\maketitle


\begin{abstract}
	We consider the problem of maximizing portfolio value when an agent has a subjective view on asset value which differs from the traded market price.
	The agent's trades will have a price impact which affect the price at which the asset is traded.
	In addition to the agent's trades affecting the market price, the agent may change his view on the asset's value if its difference from the market price persists.
	We also consider a situation of several agents interacting and trading simultaneously when they have a subjective view on the asset value.
	Two cases of the subjective views of agents are considered, one in which they all share the same information, and one in which they all have an individual signal correlated with price innovations.
	To study the large agent problem we take a mean-field game approach which remains tractable. After classifying the mean-field equilibrium we compute the cross-sectional distribution of agents' inventories and the dependence of price distribution on the amount of shared information among the agents.
\end{abstract}

%
%

\section{Introduction}
\label{sec:intro}

A significant proportion of trading performed in modern markets is done by computer algorithms with some reports giving figures of up to 80\% of trades in some markets (see \cite{kaya2016high} and \cite{bigiotti2018optimizing}). Many of these algorithms are used to execute strategies that manage inventory, for example to rebalance a portfolio or achieve a desired hedge ratio. Others may also be speculative, executing trades based on predictions of market behaviour. When a trading strategy is designed based on speculation, trade executions are typically based on a trade signal, which indicates that the value of the asset at a future time will be predictably different from its present value. Exploiting this predicted difference offers the possibility of attaining a profit.

The trade executions which are conducted in a market will also have impact on the dynamics of the market itself, and when several market participants implement trading strategies simultaneously they will inevitably influence the behaviour of each other. Thus, in order for a strategy to be designed to perform executions in an optimal manner, the trade signals of other market participants should also be taken into account. The analysis of such a system quickly leads to a high dimensional problem, but framing the system in terms of a mean-field game allows for further tractability.

In this work, we consider how an agent will base his trades through time on an observed trade signal. In general, a trade signal can be an abstract quantity which dictates tendencies of market dynamics as in \cite{donnelly2018optimal}, or it can be treated as a direct valuation adjustment of the asset compared to the prevailing midprice as in \cite{lehalle2017limit}. We take the latter approach in that the agent's trade signal is directly transformed into a monetary quantity to be added to the asset price to give a subjective valuation. The agent also controls his trading to manage the risk of his position at the end of the trading horizon as an acknowledgement that his assessment of value may not be accurate. This single agent framework is structured similar to \cite{almgren2001optimal} with the addition of the observed trade signal. Beyond the single agent problem we investigate a market in which several agents are trading, each of whom is observing a trade signal that dictates their subjective valuation of the asset. In this setting, in order to fully optimize the profits they seek to extract from their trade signal, the agents must take into account the aggregate behaviour of other market participants. In order to maintain tractability of the model when there are a large number of agents, we use a mean-field game approach. Similar approaches with respect to optimal execution and algorithmic trading are conducted in \cite{huang2015mean} and \cite{casgrain2018mean}.

Additionally, the work \cite{casgrain2018beliefs} also considers a mean-field approach where subsets of agents have different views of the asset price. While in that work the differing beliefs are modeled as agents behaving according to different probability measures, we work with a fixed probability where each agent observes a different trade signal process which forms their subjective valuation. We consider two specifications of how the trade signals of different agents relate to each other. First, we suppose that they all share the same trade signal which is correlated with the asset's midprice. Then any differences in their trading behaviour will come from different initial inventories. Second, we suppose that they all have different trade signals, each correlated with the asset's midprice with a structure that also dictates the nature of the correlation between any pair of signals. In the second case, if all correlations are equal to $1$ and the initial signal states are identical across all agents, then the model reverts to the first case of the shared signal.

In order for each agent to optimize their trades they must take into account the order flow submitted by other agents. This is similar to models proposed in \cite{cartea2016closed} and \cite{cartea2016incorporating} in which net order flow is given by an exogenous process. In this work, agents make an assumption about the net order flow process before conducting their individual optimization. A mean-field equilibrium is reached by finding a fixed point of the net order flow process. If all agents trade according to the mean-field equilibrium, then we can quantify the relationship between the correlation of their trade signals and the cross sectional distribution of their inventories through time. This also allows us to study how the overall price impact on the asset depends on how much information is shared between agents as dictated by the correlation between their signals.

In Section \ref{sec:one-agent} we propose our model of a single agent optimizing trades with the observation of a trade signal and analyze the agent's optimal trading strategy. In Section \ref{sec:multiple-agents-model} we propose our model when there are several agents trading simultaneously. The section is broken down into subsections depending on whether the agents share the same signal (Section \ref{sec:shared}) or have separate but correlated signals (Section \ref{sec:separate}). In Section \ref{sec:distribution} we compute the cross-sectional joint distribution of the inventory and trade signal across all agents and show how this depends on the correlation of the trade signals. We also compute how this correlation affects the variance of the asset price. We conclude in Section \ref{sec:conclusion}.

%
%

\section{One Agent}
\label{sec:one-agent}


\subsection{Model without interaction}
\label{sec:model-1}
In this section, we consider a single agent that wishes to maximize the value of a portfolio at a future time $T<\infty$ through trading a single risky asset with temporary and permanent price impact. The agent has his own subjective valuation of the asset which may be different from the market price. At each point in time the agent chooses a rate at which he trades shares of the asset via a process $\nu = (\nu_t)_{0 \leq t \leq T}$. Thus, denoting the agent's inventory holdings by $Q^\nu = (Q^\nu_t)_{0 \leq t \leq T}$, it changes according to
\begin{align}
	\dd Q^\nu_t &= \nu_t \dd t\,, & Q^\nu_0 &= Q_0\,.
\end{align}
The sign of $\nu_t$ indicates the types of trades to be submitted with positive values representing buy orders and negative values representing sell orders. The market view of asset value is denoted $S^\nu = (S^\nu_t)_{0 \leq t \leq T}$, which will be subject to a permanent price impact due to the agent's trades. We model permanent price impact through a linear relation and thus the market view of the asset satisfies
\begin{align}
	\dd S^\nu_t &= (\mu + b\nu_t) \dd t + \sigma \dd W_t\,, & S^\nu_0 &= S_0\,,
\end{align}
where $W = (W_t)_{0\leq t \leq T}$ is a Brownian motion. Temporary price impact is also accounted for by modeling the price of trades as being dependent on the speed of trading. Given that the speed of trading at time $t$ is $\nu_t$, the price at which the transaction occurs is
\begin{align}
	\widehat{S}^\nu_t = S^\nu_t + k \nu_t\,.
\end{align}
Thus, the agent's cash process, denoted $X^\nu = (X^\nu_t)_{0 \leq t \leq T}$ satisfies
\begin{align}
	\dd X^\nu_t &= -\widehat{S}^\nu_t\nu_t \dd t\,, & X^\nu_0 &= X_0\,.
\end{align}
We suppose that the agent observes a trade signal which means his own subjective view of the asset's value differs from the market view of the asset's value $S^\nu$. We will denote the difference between the subjective value and the market view of value by $V^\nu = (V^\nu_t)_{0 \leq t \leq T}$, which satisfies
\begin{align}
	\dd V^\nu_t &= -(\beta V^\nu_t + \gamma \nu_t) \dd t + \eta \dd Z_t\,, & V^\nu_0 &= V_0\,,
\end{align}
where $Z = (Z_t)_{0\leq t \leq T}$ is a Brownian motion correlated with $W$ with constant correlation parameter $\rho$. At time $t$ when the market value of the asset is $S_t^\nu$, the agent's subjective valuation of the asset due to the trade signal is $S_t^\nu + V_t^\nu$. Even though $V_t^\nu$ represents a valuation adjustment due to the trade signal, we will refer to the process $V^\nu$ as the trade signal. The dynamics of $V^\nu$ imply that the trade signal is influenced by the trading of the agent, this coming from the term $-\gamma \nu_t$. This is to capture the effect of diminishing the trade signal's strength when the agent acts upon the information that it provides. We also include the mean reverting term $-\beta\,V_t^\nu$ because the value imparted by the agent's trade signal will not necessarily last for long periods of time on average, even if not acted upon explicitly.


\subsection{Agent's Objective Functional and HJB Equation}
\label{sec:objective-one}
Throughout the remainder of Section \ref{sec:one-agent} we work with a complete and filtered probability space $(\Omega, (\mathcal{F}_t)_{0\leq t \leq T},\mathbb{P})$ where $(\mathcal{F}_t)_{0\leq t \leq T}$ is the standard augmentation of the natural filtration generated by $(W_t,Z_t)_{0\leq t\leq T}$ and $(S_0,Q_0,X_0,V_0)$. We suppose the agent wishes to maximize the following functional of $\nu$:
\begin{align}
J(\nu)
	&:=	\Eb \Big( X_T^\nu + Q_T^\nu (S_T^\nu + V_T^\nu) - \alpha (Q_T^\nu)^2\Big) , 
\end{align}
where the control process $\nu$ must be taken from the admissible set $\mathcal{N}$ which consists of $\mathcal{F}$-predictable processes such that $\mathbb{E}[\int_0^T \nu_t^2\,dt]<\infty$.
The first term in the expectation $X_T^\nu$ is the amount of cash on hand at time $T$.
The second term in the expectation $Q_T^\nu (S_T^\nu + V_T^\nu)$ is the agent's assessment of the value of his inventory holdings at time $T$. The third term $-\alpha (Q_T^\nu)^2$ behaves as a risk control term and is present because the agent acknowledges that his valuation due to the trade signal may not be completely accurate. This term helps to ensure that he does not acquire very large inventory positions due to the the risk of being incorrect. Let us define the agent's \textit{value function} $H$ as follows
\begin{align}
H(t,x,q,S,V)
	&:=	\sup_{\nu \in \Nc} \Eb_{t,x,q,S,V} \Big( X_T^\nu + Q_T^\nu (S_T^\nu + V_T^\nu) - \alpha (Q_T^\nu)^2 \Big) . \label{def:H}
\end{align}
The Hamilton-Jacobi-Bellman (HJB) partial differential equation (PDE) associated with $H$ is
\begin{align}
\d_t H + \sup_{\nu \in \Rb} ( \Ac^\nu H  )
	&=	0 , &
H(T,x,q,S,V)
	&=	x + q(S+V) - \alpha q^2 , \label{eq:hjb-pde}
\end{align}
where  the operator $\Ac^\nu$ is given by
\begin{align}
\Ac^\nu
	&=	-(S+k\nu)\nu\d_x + \nu\d_q + (\mu + b\nu)\d_S - (\beta V + \gamma \nu) \d_V + \frac{1}{2}\sigma^2\d_{SS} + \frac{1}{2}\eta^2\d_{VV} + \rho\sigma\eta\d_{SV}\,.
\end{align}


\subsection{Solving the HJB PDE}
\label{sec:hjb-pde}
In this section, we express the solution to the HJB PDE \eqref{eq:hjb-pde} in terms of a system of coupled ODEs.
\begin{proposition}
	Suppose $c_1,\dots,c_6:[0,T]\rightarrow\mathbb{R}$ satisfy the following system of ODEs with terminal conditions:
	\begin{align}
		c'_1 + \eta^2c_5 + \frac{(c_2-\gamma c_3)^2}{4k} &= 0\,,						& c_1(T) &= 0\,,\\
		c'_2 + \mu + \frac{(c_2-\gamma c_3)(b + 2 c_4 - \gamma c_6)}{2k} &= 0\,, 		& c_2(T) &= 0\,, \label{eqn:prop-c2}\\
		c'_3 - \beta c_3 + \frac{(c_2 - \gamma c_3)(c_6 - 2\gamma c_5)}{2k} &= 0\,,		& c_3(T) &= 0\,, \label{eqn:prop-c3}\\
		c'_4 + \frac{(b + 2c_4 - \gamma c_6)^2}{4k} &= 0\,,								& c_4(T) &= -\alpha\,,\\
		c'_5 - 2\beta c_5 + \frac{(c_6 - 2\gamma c_5)^2}{4k} &= 0\,,					& c_5(T) &= 0\,,\\
		c'_6 - \beta c_6 + \frac{(c_6 - 2\gamma c_5)(b+2c_4-\gamma c_6)}{2k} &= 0\,, 	& c_6(T) &= 1\,. 
	\end{align}
	Then the value function $H$ is given by
	\begin{align}
		H(t,x,q,S,V) &= x + qS + h(t,q,V)\,,\\
		h(t,q,V) &= c_1(t) + c_2(t)q + c_3(t)V + c_4(t)q^2 + c_5(t)V^2 + c_6(t)qV\,,
	\end{align}
	and the optimal trading strategy in feedback form is
	\begin{align}
		\nu^*(t,q,V) = \frac{c_2(t)-\gamma c_3(t) + (b + 2c_4(t) - \gamma c_6(t))q + (c_6(t) - 2\gamma c_5(t))V}{2k}\,.
	\end{align}
\end{proposition}
\begin{proof}
	This is shown by direct substitution into \eqref{eq:hjb-pde}.
\end{proof}
\begin{proposition}\label{prop:symmetry1}
	If $\mu = 0$ then $c_2 = c_3 \equiv 0$.
\end{proposition}
\begin{proof}
	This is immediate from equations \eqref{eqn:prop-c2} and \eqref{eqn:prop-c3}.
\end{proof}

Proposition \ref{prop:symmetry1} is a result of symmetry in the model when the unaffected market value of the asset is a martingale. In this circumstance, the agent's value function remains unchanged if the underlying state variables are transformed according to $(q,S,V)\mapsto(-q,-S,-V)$. In addition, we also see $\nu^*(t,-q,-V) = -\nu^*(t,q,V)$. This is an expected result because if the dynamics possess enough symmetry, then the agent should place equal value on a long position in the asset as they would on a short position of equal magnitude, as long as his future projection of the total value of his current holdings is the same.

\subsection{Numerical Experiments}\label{sec:numerical-experiments-individual}

Of particular interest is how the trading strategy depends on the values of $Q_t$ and $V_t$. The effect of these underlying processes on the trading strategy can be directly quantified by the corresponding loadings. 
To this end, observe that $\nu^*$ can be written as
\begin{align}
		\nu^*(t,q,V) = \frac{c_2(t)-\gamma c_3(t)}{2k} + \nu^*_q(t) q + \nu^*_V(t)  V \, ,
\end{align}
where $\nu^*_q$ and $\nu^*_V$ are defined by
\begin{align}
	\nu^*_q(t) &= \frac{b + 2 c_4(t) - \gamma c_6(t)}{2k}\,, & \nu^*_V(t) &= \frac{c_6(t) - 2bc_5(t)}{2k}\,.
\end{align}

\begin{figure}
	\begin{center}
		{\includegraphics[trim=140 240 140 240, scale=0.48]{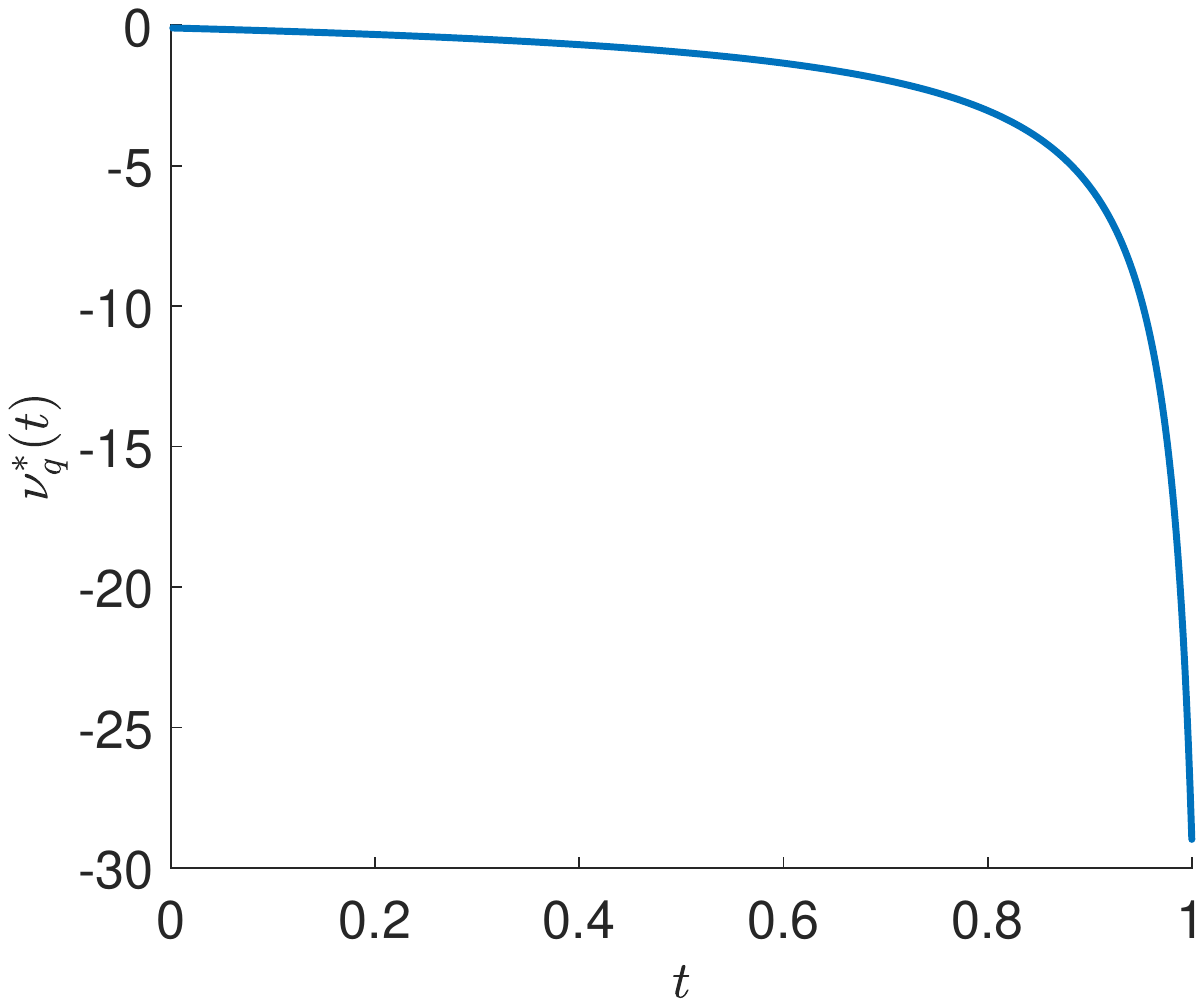}}\hspace{10mm}
		{\includegraphics[trim=140 240 140 240, scale=0.48]{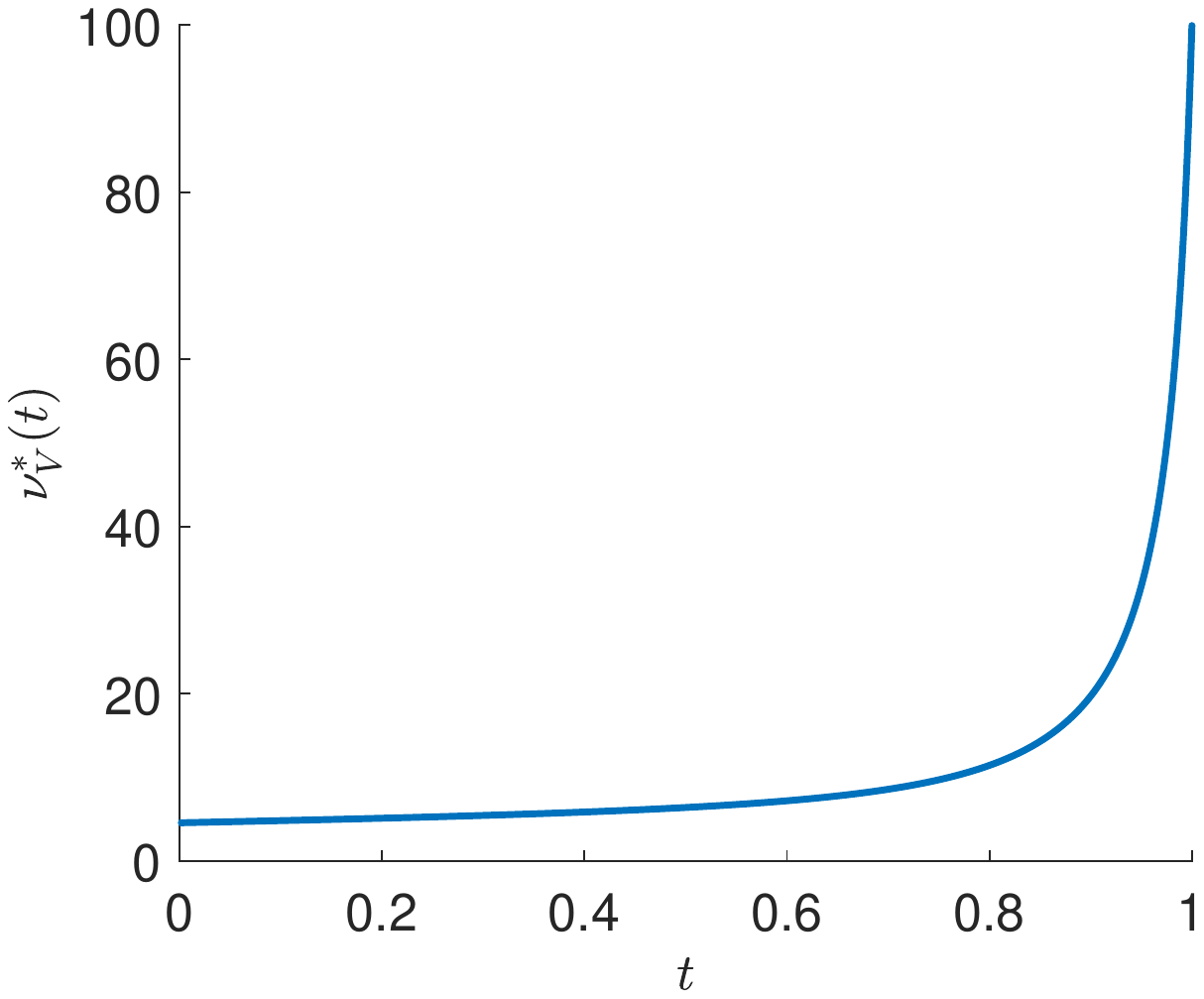}}
	\end{center}
	\vspace{-1em}
	\caption{Optimal loadings on $Q_t$ and $V_t$. Parameters used are $\mu = 0$, $\sigma = 1$, $\eta = 0.5$, $\beta = 1$, $\gamma = 0.1$, $\rho = 0.3$, $b = 10^{-2}$, $k = 5\cdot10^{-3}$, $\alpha = 0.1$, and $T = 1$. \label{fig:trade-loadings-individual}}
\end{figure}

In Figure \ref{fig:trade-loadings-individual} we plot the loadings on $Q_t$ and $V_t$ for the optimal trading strategy $\nu^*$. The behaviour we see in the left panel is typical of this type of model for optimal execution. In particular, the loading on $Q_t$ is seen to grow to a large negative value as $t\rightarrow T$. This is because of the agent's terminal risk control represented by $\alpha Q_T^2$ and the relevance of this term becomes stronger as time approaches the horizon of the trading period. 
Although the loading shown above is typical, we will see below that the contribution of inventory towards trading speed, and indeed the total trading speed, exhibits behaviour which is not typically seen in this style of optimal execution.

There is also an intuitive explanation for why the loading on $V_t$ is most significant close to time $T$, and it is because $V_t$ is more likely to experience a change in sign if there is a longer time to the horizon of the trading period. If the agent is overeager in his attempts to extract profits early in the trading period due to exploiting the trade signal $V_t$, then he risks this quantity changing sign in which case his prior trades are in fact working against his goals. If this occurs then the agent would wish to reverse his trades, but the round trip involved in this task accumulates needless costs due to temporary price impact. By waiting until a time closer to $T$, the probability and magnitude of this type of sign change is significantly lowered, thus the agent prefers to wait before extracting profits.

In Figure \ref{fig:MC-individual} we show simulated paths of the optimal trading strategy broken down into the two main contributing components, as well as the total trading speed. The top left panel shows the graph of $\nu_q^*(t)Q^{\nu^*}_t$, the top center panel shows $\nu_V^*(t)V^{\nu^*}_t$, and the top left shows $\nu^*(t,Q^{\nu^*}_t,V^{\nu^*}_t)$. Here we see what may be considered atypical behaviour in an optimal execution program. Namely that all of the individual contributions to the trading speed, as well as the total trading speed, are concentrated towards the end of the trading period. In other optimal execution models, for example \cite{almgren2001optimal} and \cite{cartea2018hedging}, the loading on $Q_t$ becomes large as $t$ approaches $T$, but the magnitude of the contribution to the trading speed from inventory does not change significantly over the course of the trading period. The main contributing factor to this difference is related to the discussion of the previous paragraph. The agent's main source of profits is due to the value of the trade signal $V_t$, but this is only taken advantage of through trades that are submitted when the sign of $V_t$ is the same as the sign of $V_T$. This gives incentive for the agent to delay trades and perform most of his action at times close to $T$.

\begin{figure}
	\begin{center}
		{\includegraphics[trim=140 240 140 240, scale=0.4]{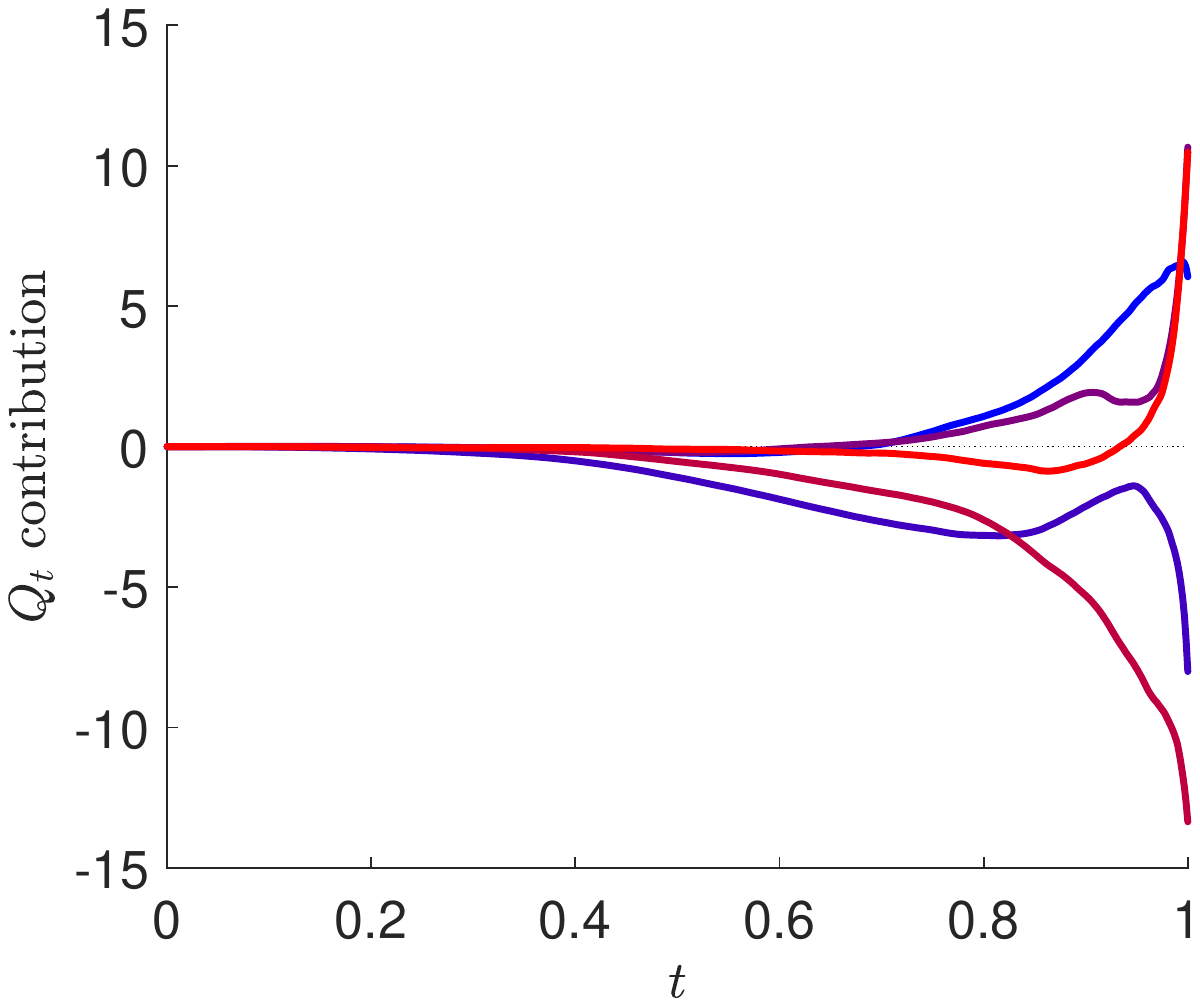}}\hspace{10mm}
		{\includegraphics[trim=140 240 140 240, scale=0.4]{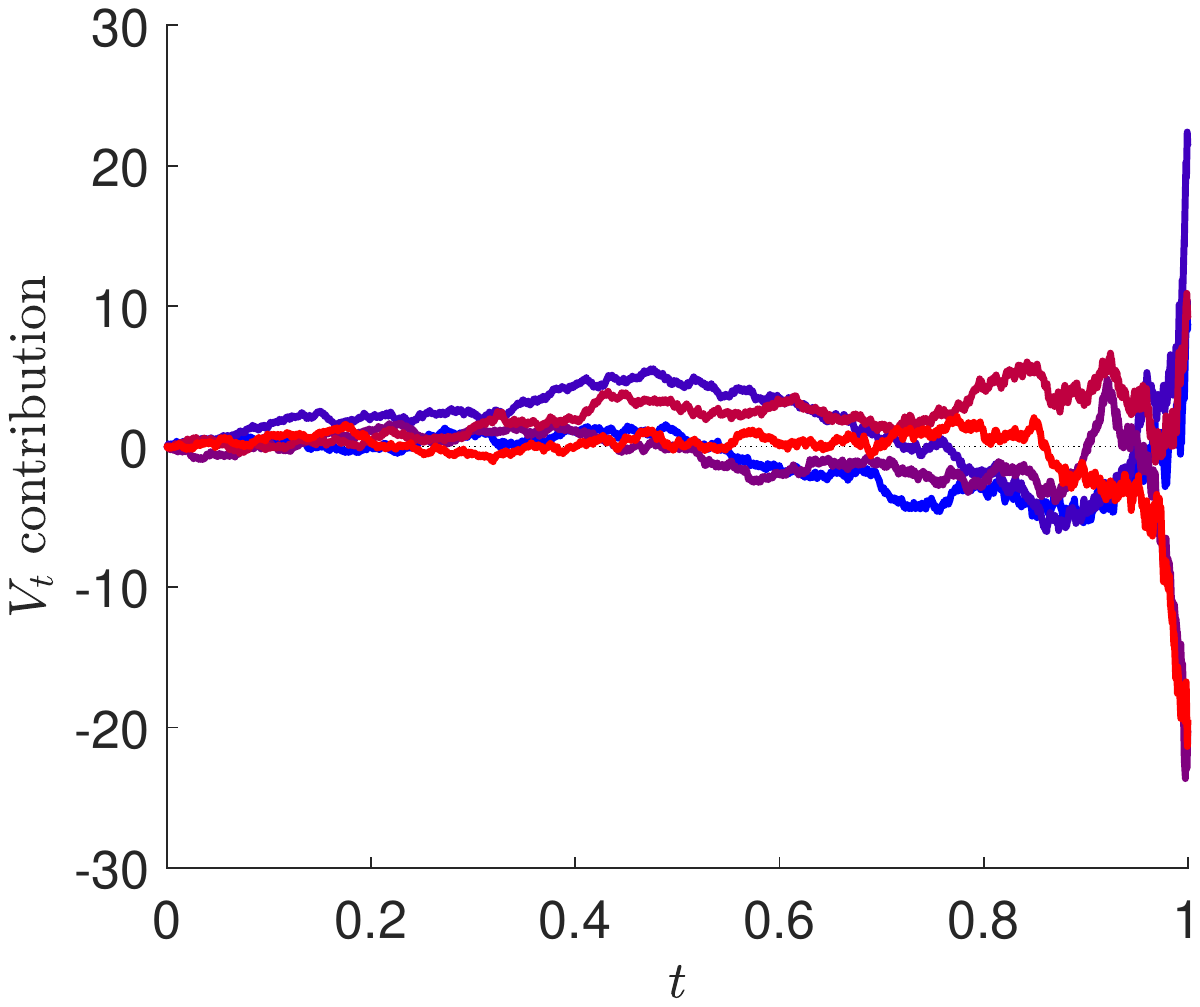}}\hspace{10mm}
		{\includegraphics[trim=140 240 140 240, scale=0.4]{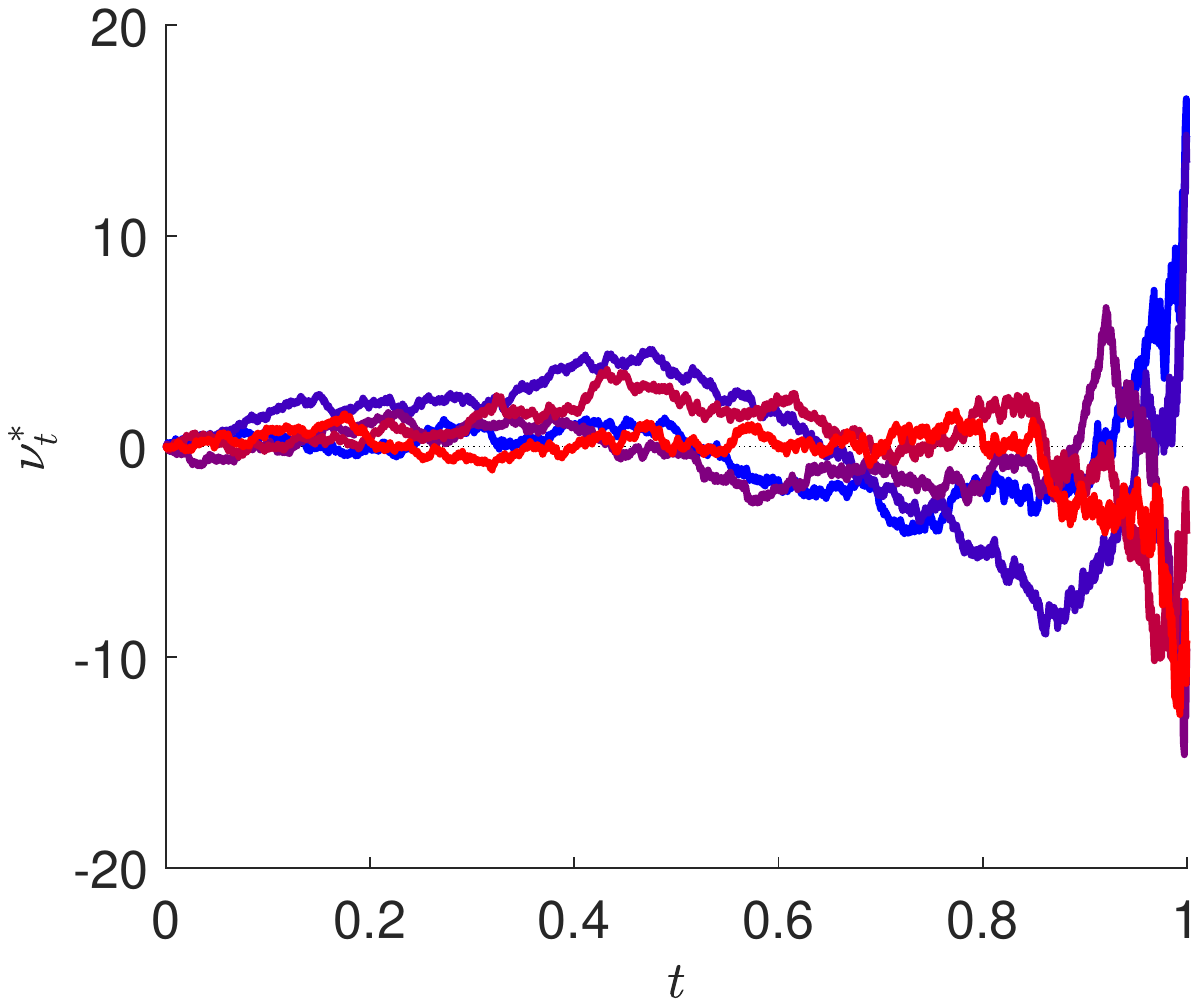}}\\
		{\includegraphics[trim=140 240 140 240, scale=0.4]{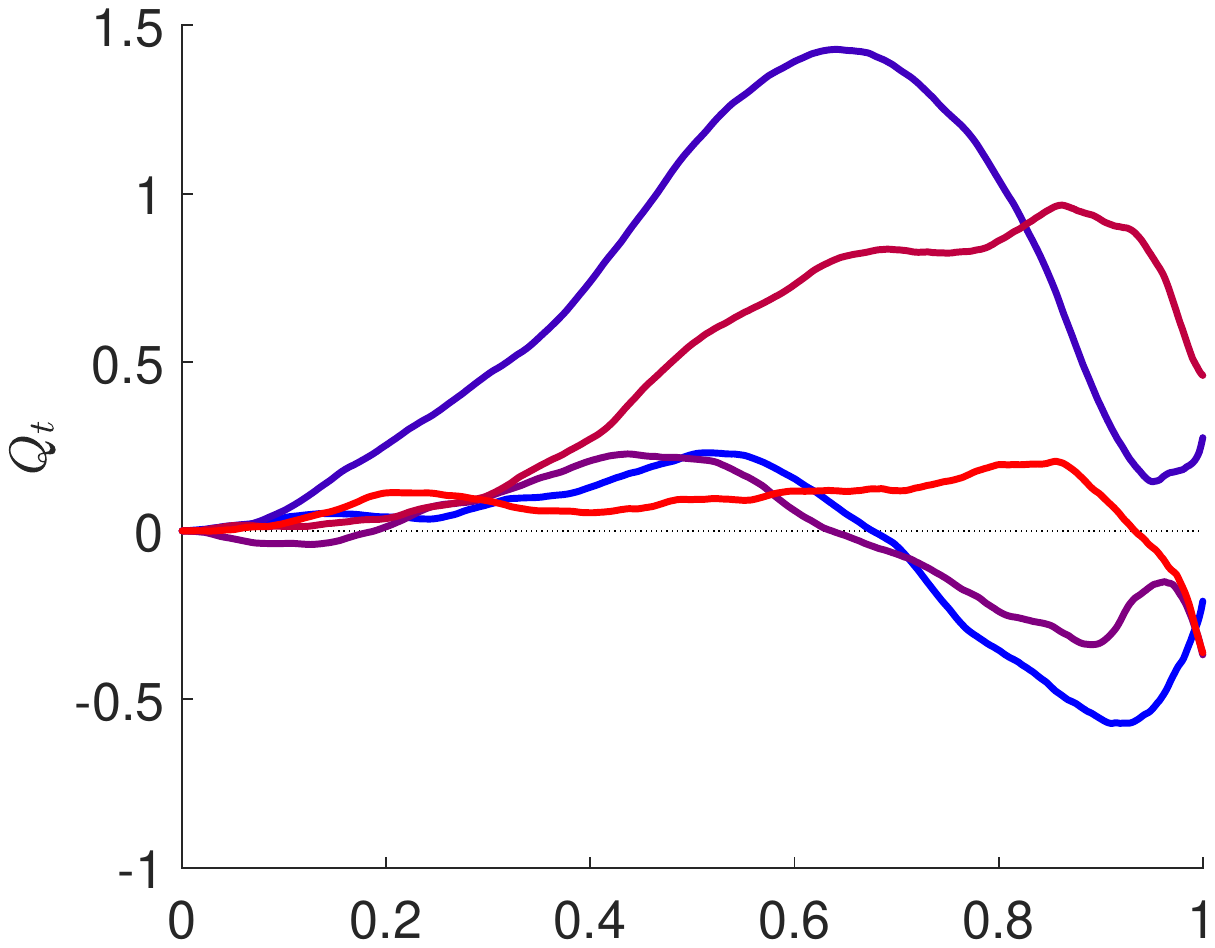}}\hspace{10mm}
		{\includegraphics[trim=140 240 140 240, scale=0.4]{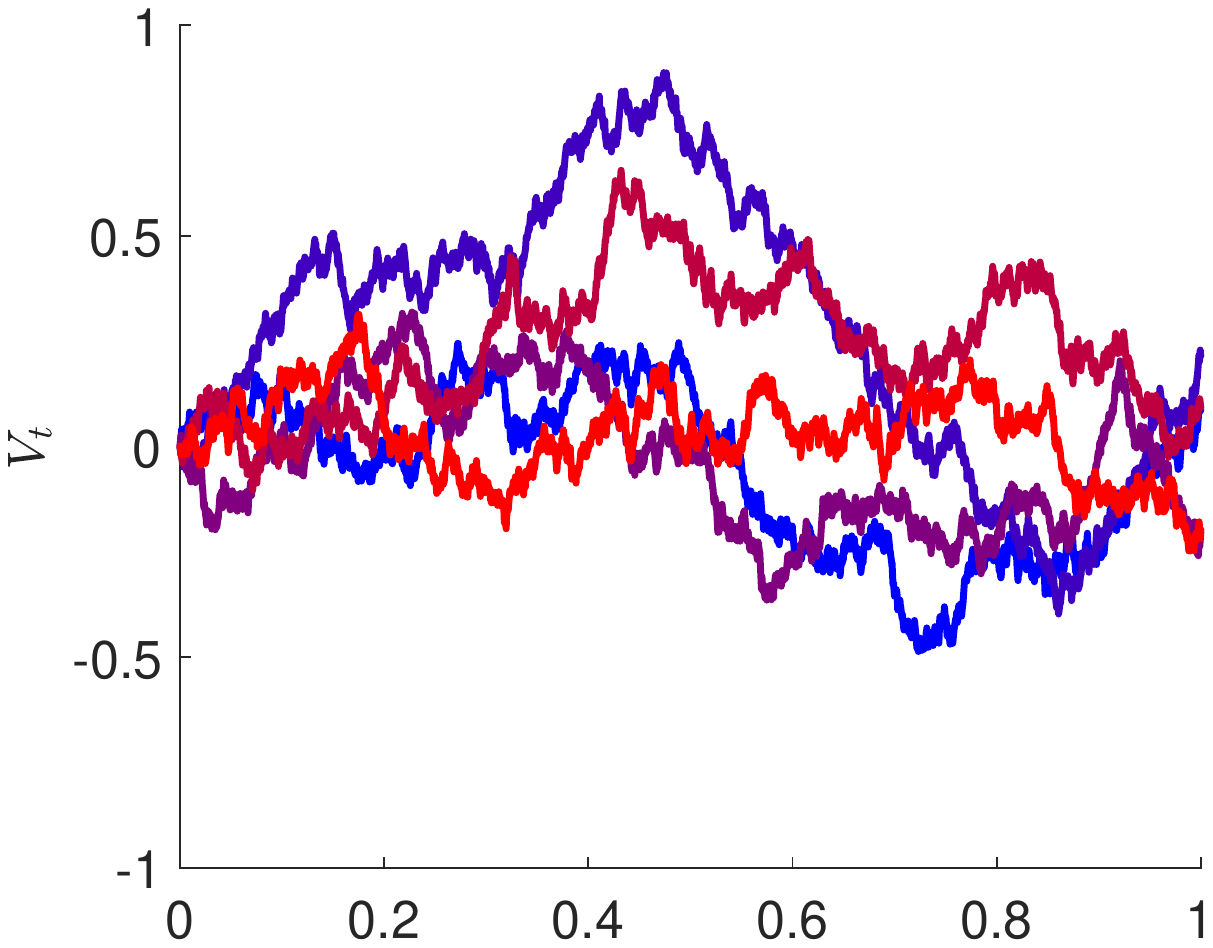}}\hspace{10mm}
	\end{center}
	\vspace{-1em}
	\caption{Contribution to trading speed of $Q_t$ and $V_t$, and total optimal trading speed $\nu_t$. Parameters used are $\mu = 0$, $\sigma = 1$, $\eta = 0.5$, $\beta = 1$, $\gamma = 0.1$, $\rho = 0.3$, $b = 10^{-2}$, $k = 5\cdot10^{-3}$, $\alpha = 0.1$, $T = 1$, $S_0 = 100$, $V_0 = 0$, and $Q_0 = 0$. \label{fig:MC-individual}}
\end{figure}

\section{Multiple Agents}

\subsection{Model with Interaction}\label{sec:multiple-agents-model}

Here we consider a model in which multiple agents trade with interaction. The interaction stems from the fact that price impact will account for the trading of all agents, not just an individual. Agents are indexed by $n\in\{1, \dots, N\}$ and each agent has his own control process $\nu^n = (\nu^n_t)_{0 \leq t \leq T}$. As before the control process represents the rate of trading for agent $n$, thus the inventory holdings of agent $n$ denoted by $Q^{n,\nu^n} = (Q^{n,\nu^n}_t)_{0 \leq t \leq T}$
changes according to
\begin{align}
	\dd Q^{n,\nu^n}_t &= \nu^n_t \dd t\,, & Q^{n,\nu^n}_0 &= Q^n_0\,,
\end{align}
where we assume each $Q_0^n$ is independent from all other variables with finite expectation and variance. 
The market view of the asset value is denoted $S^{\bar{\nu}} = (S^{\bar{\nu}}_t)_{0 \leq t \leq T}$ and changes according to
\begin{align}
	\dd S^{\bar{\nu}}_t &= (\mu + b\bar{\nu}_t) \dd t + \sigma \dd W_t\,, & S^{\bar{\nu}}_0 &= S_0\,,\label{eqn:S_shared}
\end{align}
where $\bar{\nu}_t$ is the average trading rate of all agents at time $t$
\begin{align}
	\bar{\nu}_t &:= \frac{1}{N}\sum_{n=1}^N \nu_t^n \, .
\end{align}
Temporary price impact incurred by agent $n$ depends on his own rate of trading as well as that of the other agents such that the transaction price for agent $n$ is
\begin{align}
	\widehat{S}^{\nu^n,\bar{\nu}}_t &= S^{\bar{\nu}}_t + k\nu^n_t + \bar{k}\bar{\nu}_t\,.
\end{align}
Thus, the cash process of agent $n$, denoted $X^{n,\nu^n,\bar{\nu}} = (X^{n,\nu^n,\bar{\nu}}_t)_{0 \leq t \leq T}$ changes according to
\begin{align}
		\dd X^{n,\nu^n,\bar{\nu}}_t &= -\widehat{S}^{\nu^n,\bar{\nu}}_t\nu^n_t \dd t\,, & X^{n,\nu^n,\bar{\nu}}_0 &= X^n_0\,.
\end{align}
Lastly, it will be useful to define the average inventory holdings of all agents $\bar{Q}^{\bar{\nu}} = (\bar{Q}^{\bar{\nu}}_t)_{0 \leq t \leq T}$ which is given by
\begin{align}
	\bar{Q}^{\bar{\nu}}_t &:= \frac{1}{N}\sum_{n=1}^N Q^{n,\nu^n}_t\,.\label{eqn:Q_bar-definition}
\end{align}
It appears from the definition that $\bar{Q}^{\bar{\nu}}$ depends on each individual $\nu^n$, but given $\bar{Q}_0$ we have
\begin{align}
	\dd \bar{Q}^{\bar{\nu}}_t 
	& = \frac{1}{N}\sum_{n=1}^N \dd Q_t^{n,\nu^n} \\
	& = \frac{1}{N}\sum_{n=1}^N \nu^n_t \dd t\\
	& =\bar{\nu}_t \dd t\,,
\end{align}
which justifies the dependence on only $\bar{\nu}$. We remark here that based on the definition of $\bar{Q}^{\bar{\nu}}_t$ in \eqref{eqn:Q_bar-definition}, if we assume in addition that all $Q_0^n$ are independent and identically distributed then when we directly handle the limiting case $N\rightarrow\infty$ we have $\bar{Q}^{\bar{\nu}}_0 = \Eb[Q_0^{n,\nu^n}]$. From now on we will make this assumption on the collection $(Q_0^n)_{n\in\mathbb{N}}$.

\subsection{Shared Subjective View of Asset Value}\label{sec:shared}
First we consider a model in which each agent receives the same trade signal. When trading takes place in an order book market, some examples of possible shared signals are the order book imbalance or micro-price, quantities which are known to convey information about the future dynamics of the asset. The common trade signal is denoted by $\bar{V}^{\bar{\nu}} = (\bar{V}^{\bar{\nu}}_t)_{0 \leq t \leq T}$ and changes according to
\begin{align}
	\dd \bar{V}^{\bar{\nu}}_t &= -(\beta \bar{V}_t^{\bar{\nu}} + \bar{\gamma}\bar{\nu}_t) \dd t + \eta \dd Z_t\,, & \bar{V}^{\bar{\nu}}_0 &= \bar{V}_0\,, \label{eqn:V_shared}
\end{align}
where $Z = (Z_t)_{0\leq t \leq T}$ is a Brownian motion correlated with $W$ with constant correlation parameter $\rho$. The $N$ agents under consideration do not represent all participants in the market, only the ones which are acting based on the trade signal $\bar{V}^{\bar{\nu}}$. This is why there is a prevailing market view of the value $S^{\bar{\nu}}$ which is different from the subjective valuation due to the trade signal of $S^{\bar{\nu}}+\bar{V}^{\bar{\nu}}$. We assume that the agents have a method of distinguishing between order flow which is based on the information content contained in the trade signal versus other order flow. High-frequency traders often have methods of distinguishing between informed and uninformed order flow, or trades that are part of a large meta-order. The ability to make this identification can be interpreted as a filtering problem (see for example \cite{kyle1985continuous} and its numerous extensions), but incorporating that directly into this model is beyond the scope the paper.

Each agent attempts to maximize his own expected future wealth given that the trading strategies of all other agents are fixed. We let $\nu^{-n}$ denote the collection of trading strategies for all agents except agent $n$. Then for a fixed $\nu^{-n}$, agent $n$ wishes to maximize the functional
\begin{align}
	J(\nu^n;\nu^{-n}) &= \Eb \Big( X_T^{n,\nu^n,\bar{\nu}} + Q_T^{n,\nu^n} (S_T^{\bar{\nu}} + \bar{V}_T^{\bar{\nu}})  - \alpha (Q_T^{n,\nu^n})^2\Big) \,.
\end{align}
Throughout Section \ref{sec:shared} we work with a complete and filtered probability space $(\Omega, (\mathcal{F}_t)_{0\leq t \leq T},\mathbb{P})$ where $(\mathcal{F}_t)_{0\leq t \leq T}$ is the standard augmentation of the natural filtration generated by $(W_t,Z_t)_{0\leq t\leq T}$ and the initial state $(S_0,(Q^n_0)_{n\in\mathbb{N}},(X^n_0)_{n\in\mathbb{N}},\bar{V}_0)$.

\subsubsection{HJB Equation and Consistency Condition}\label{sec:HJB-shared}

We do not attempt to solve the finite player game, rather we consider the limiting case $N\rightarrow\infty$ directly. Under this condition the average trading speed $\bar{\nu}$ is not affected by any one individual control $\nu^n$. Thus, fixing $\nu^{-n}$ is equivalent to fixing $\bar{\nu}$. In addition, we assume $\bar{\nu}_t = \bar{\nu}(t,\bar{Q}^{\bar{\nu}}_t,\bar{V}^{\bar{\nu}}_t)$ so that we remain within a Markovian framework. With a fixed function $\bar{\nu}$ we may define the value function for agent $n$ as
\begin{align}
	H^n(t,x,q,\bar{q},S,\bar{V};\bar{\nu}) &:= \sup_{\nu^n \in \Nc} \Eb_{t,x,q,\bar{q},S,\bar{V}} \Big( X_T^{n,\nu^n,\bar{\nu}} + Q_T^{n,\nu^n} (S_T^{\bar{\nu}} + \bar{V}_T^{\bar{\nu}}) - \alpha (Q_T^{n,\nu^n})^2\Big) , \label{def:Hn}
\end{align}
where the collection of admissible strategies $\mathcal{N}$ consists of $\mathcal{F}$-predicable processes such that $\mathbb{E}[\int_0^T (\nu^n_t)^2\,dt]<\infty$.

The value function in \eqref{def:Hn} has an associated HJB equation of the form
\begin{align}
	\d_t H^n + \sup_{\nu^n \in \Rb} ( \Ac^{\nu^n,\bar{\nu}} H^n  )
		&=	0 , &
	H^n(T,x,q,\bar{q},S,\bar{V};\bar{\nu})
	&=	x + q\bar{V} - \alpha q^2 , \label{eq:hjb-pde-shared}
\end{align}
where the operator $\Ac^{\nu^n,\bar{\nu}}$ is given by
\begin{align}
\Ac^{\nu^n,\bar{\nu}}
	&=	-(S+k\nu^n + \bar{k}\bar{\nu})\nu^n\d_x + \nu^n\d_q + \bar{\nu}\d_{\bar{q}} + (\mu + b\bar{\nu})\d_S - (\beta \bar{V} + \bar{\gamma} \bar{\nu})\d_{\bar{V}} + \frac{1}{2}\sigma^2\d_{SS} + \frac{1}{2}\eta^2\d_{\bar{V}\bar{V}} + \rho\sigma\eta\d_{S\bar{V}}\,.
\end{align}
Based on the form of the feedback control in the previous section, we make the ansatz
\begin{align}
	\bar{\nu}(t,\bar{q},S,\bar{V}) &= f_1(t) + f_2(t)\bar{q} + f_3(t)\bar{V}\,. \label{eq:nubar-form}
\end{align}
This ansatz will allow the problem to retain a linear-quadratic structure, but as a consequence the trading strategies in equilibrium will be linear. We do not claim that the equilibrium in general is unique, and there may exist non-linear trading strategies which form a mean-field Nash equilibrium, but we do not consider them in this work. With this ansatz the solution to the HJB equation \eqref{eq:hjb-pde-shared} along with the optimal control in feedback form can be characterized by a solution to a system of ODE's.

\begin{proposition}\label{prop:ODE_shared}
	Given $\bar{\nu}$ in \eqref{eq:nubar-form}, suppose $c_1,\dots,c_{10}:[0,T]\rightarrow\mathbb{R}$ satisfy the following system of ODEs with terminal conditions:
	\begin{align}
		c'_1 + f_1(c_3 - \bar{\gamma} c_4) + \eta^2 c_7 + \frac{(c_2 - \bar{k}f_1)^2}{4k} &= 0\,,										& c_1(T) &= 0\,,\label{eqn:shared-c1}\\
		c'_2 + \mu + f_1(b + c_8-\bar{\gamma} c_9) + \frac{(c_2-\bar{k}f_1)c_5}{k} &= 0\,,												& c_2(T) &= 0\,,\\
		c'_3 + f_1(2c_6 - \bar{\gamma} c_{10}) + f_2(c_3 - \bar{\gamma} c_4) + \frac{(c_2-\bar{k}f_1)(c_8-\bar{k}f_2)}{2k} &= 0\,,						& c_3(T) &= 0\,,\\
		c'_4 + f_1(c_{10} - 2\bar{\gamma} c_7) + f_3(c_3 - \bar{\gamma} c_4) - \beta c_4 + \frac{(c_2-\bar{k}f_1)(c_9-\bar{k}f_3)}{2k} &= 0\,,			& c_4(T) &= 0\,,\\
		c'_5 + \frac{c_5^2}{k} &= 0\,,																				& c_5(T) &= -\alpha\,,\\
		c'_6 + f_2(2c_6 - \bar{\gamma} c_{10}) + \frac{(c_8-\bar{k}f_2)^2}{4k} &= 0\,,													& c_6(T) &= 0\,,\\
		c'_7 + f_3(c_{10} - 2\bar{\gamma} c_7) - 2\beta c_7 + \frac{(c_9-\bar{k}f_3)^2}{4k} &= 0\,,									& c_7(T) &= 0\,,\\
		c'_8 + f_2(b+c_8-\bar{\gamma} c_9) + \frac{c_5(c_8-\bar{k}f_2)}{k} &= 0\,,														& c_8(T) &= 0\,,\\
		c'_9 + f_3(b+c_8-\bar{\gamma} c_9) - \beta c_9 + \frac{c_5(c_9-\bar{k}f_3)}{k} &= 0\,,											& c_9(T) &= 1\,,\\
		c'_{10} + f_2(c_{10}-2\bar{\gamma} c_7) + f_3(2c_6-\bar{\gamma} c_{10}) - \beta c_{10} + \frac{(c_8-\bar{k}f_2)(c_9-\bar{k}f_3)}{2k} &= 0\,,		& c_{10}(T) &= 0\,,\label{eqn:shared-c10}
	\end{align}
	Then the value function $H^n$ is given by
	\begin{align}
		H^n(t,x,q,\bar{q},S,\bar{V}) &= x + qS + h^n(t,q,\bar{q},\bar{V})\,,\\
		h^n(t,q,\bar{q},\bar{V}) &= c_1(t) + c_2(t)q + c_3(t)\bar{q} + c_4(t)\bar{V} \\ &\quad
			+ c_5(t)q^2 + c_6(t)\bar{q}^2 + c_7(t)\bar{V}^2 + c_8(t)q\bar{q} + c_9(t)q\bar{V} + c_{10}(t)\bar{q}\bar{V}\,,
	\end{align}
	and the optimal trading strategy in feedback form is
	\begin{align}
		\nu^{n*}(t,q,\bar{q},V) = \frac{c_2(t)-\bar{k}f_1(t)}{2k} + \frac{c_5(t)}{k}q + \frac{c_8(t) - \bar{k}f_2(t)}{2k}\bar{q} + \frac{c_9(t) - \bar{k}f_3(t)}{2k}\bar{V}\,. \label{eq:nustar-shared}
	\end{align}
\end{proposition}
\begin{proof}
	This is shown by direct substitution into \eqref{eq:hjb-pde-shared}.
\end{proof}
In order for the trading strategy in \eqref{eq:nustar-shared} to yield a mean-field Nash equilibrium it is necessary that a consistency condition is satisfied. Because \eqref{eq:nustar-shared} is based on the ansatz that the average trading speed is given by \eqref{eq:nubar-form}, we must impose that when each agent uses the strategy \eqref{eq:nustar-shared} the resulting average trading speed is \eqref{eq:nubar-form}. Thus, we require
\begin{align}
	\lim_{N\rightarrow\infty}\frac{1}{N}\sum_{n=1}^N \nu^{n*}(t,q^n,\bar{q},\bar{V}) = \bar{\nu}(t,\bar{q},\bar{V})\,.
\end{align}
Substituting \eqref{eq:nubar-form} and \eqref{eq:nustar-shared} into this equation yields
\begin{align}
	f_1 &= \frac{c_2}{2k + \bar{k}}\,, &	f_2 &= \frac{2c_5+c_8}{2k + \bar{k}}\,, & f_3 &= \frac{c_9}{2k + \bar{k}}\,.\label{eqn:shared-f}
\end{align}
When solving equations \eqref{eqn:shared-c1} to \eqref{eqn:shared-c10} we shall always substitute \eqref{eqn:shared-f} first to guarantee that the optimal strategy in \eqref{eq:nustar-shared} represents an equilibrium. In a mean-field Nash equilibrium, the optimal strategy of agent $n$ can be written in a particular form demonstrated in the next proposition.
\begin{proposition}
	In equilibrium, the trading strategy of agent $n$ and the average trading rate of all agents are related by
	\begin{align}
		\nu^{n*}(t,q^n,\bar{q},\bar{V}) &= \frac{c_5(t)}{k}(q^n-\bar{q}) + \bar{\nu}(t,\bar{q},\bar{V})\,.
	\end{align}
\end{proposition}
\begin{proof}
	This is a consequence of combining equations \eqref{eq:nubar-form}, \eqref{eq:nustar-shared}, and \eqref{eqn:shared-f}.
\end{proof}
\begin{proposition}
	If $\alpha = 0$ then $c_3 = c_5 = c_6 = c_8 = c_{10} \equiv 0$ in equilibrium. If $\mu = 0$ then $c_2 = c_3 = c_4 \equiv 0$ in equilibrium. If $\alpha = \mu = 0$ then the non-zero $c_i$ in equilibrium are given by
	\begin{align}
		c_1(t) &= \frac{k\eta^2}{2b\kappa} \int_t^T (1 - e^{-\frac{2b}{\kappa}(T-s)}) \dd s\,,\\
		c_7(t) &= \frac{k}{2b\kappa}c_9^2(t)(1 - e^{-\frac{2b}{\kappa}(T-t)})\,,\\
		c_9(t) &= \frac{2z}{(1 + 2z)e^{\omega(T-t)} - 1}\,,
	\end{align}
	and the optimal strategy in feedback form is
	\begin{align*}
		\nu^{n*}(t,q^n,\bar{q},\bar{V}) &= \frac{c_9(t)}{\kappa}\bar{V} = \frac{1}{\kappa}\frac{2z}{(1 + 2z)e^{\omega(T-t)} - 1}\bar{V}\,,
	\end{align*}
	where
	\begin{align*}
		\kappa &= 2k + \bar{k}\,, & z &= \frac{\kappa\beta - b}{2\bar{\gamma}}\,, & \omega &= \frac{\kappa\beta - b}{\kappa}\,.
	\end{align*}
\end{proposition}
\begin{proof}
	In equations \eqref{eqn:shared-c1} to \eqref{eqn:shared-c10} we substitute \eqref{eqn:shared-f}. The result can then be seen by direct substitution.
\end{proof}

\subsubsection{Numerical Experiments}

In a similar fashion to Section \ref{sec:numerical-experiments-individual} we are interested in the loadings of the optimal strategy on the underlying processes and the resulting pathwise behaviour. By substituting equations \eqref{eqn:shared-f} into equations \eqref{eqn:shared-c1} to \eqref{eqn:shared-c10} we arrive at a system of ODE's which define a mean-field Nash equilibrium, and this system easily lends itself to numerical methods. From \eqref{eq:nustar-shared} we have
	\begin{align}
		\nu^{n*}(t,q,\bar{q},\bar{V}) = \frac{c_2(t)}{2k+\bar{k}} + \nu^*_q(t) q + \nu^*_{\bar{q}}  \bar{q} + \nu^*_{\bar{V}}(t) \bar{V}\, ,
	\end{align}
where we have defined
\begin{align}
	\nu^*_q(t) &= \frac{c_5(t)}{k}\,, & \nu^*_{\bar{q}} &= \frac{2kc_8(t)-2\bar{k}c_5(t)}{2k(2k+\bar{k})} & \nu^*_{\bar{V}}(t) &= \frac{c_9(t)}{2k+\bar{k}}\, .
\end{align}
We plot the functions $\nu^*_q$, $\nu^*_{\bar{q}}$ and $\nu^*_{\bar{V}}$ in Figure \ref{fig:trade-loadings-shared}.  We see that $\nu^*_q$ and $\nu^*_{\bar{V}}$ have qualitatively similar behaviour in the multi-agent setting as when there is an individual agent, and the reasoning is the same.

The additional loading $\nu^*_{\bar{q}}$ is seen to change sign about half way through the trading period, and shortly after reaching its maximum positive value it quickly drops to zero. The behaviour of this loading can be explained by considering the actions of the entire population of agents combined with the resulting dynamics of the midprice and trade signal, and the results are more easily understood by seeing the effect that the price impact and trade signal impact parameters, $b$ and $\bar{\gamma}$, have on this loading.

In Figure \ref{fig:nu_shared_qbar_gammab} we show the loading $\nu^*_{\bar{q}}$ as a function of time for several values of the parameters $b$ and $\bar{\gamma}$. In the left panel, the parameter $\bar{\gamma}$ is fixed and each curve represents a different value of $b$. As $b$ is increased, the loading decreases, and the sensitivity is greatest at earlier times. The direction and magnitude of these changes as well as the sign of the loading are understood by considering what the population will do on average based on the average inventory of the population and based on the time remaining in the trading period. If the average inventory is high, then the agent expects the order flow to be negative. This will have two effects on quantities relevant to the agent. First, it will decrease the price of the asset in the future, and second, it will increase the value of the trade signal to the agent. When the remaining time is short, this increase in the value of the trade signal gives incentive to buy shares of the asset to benefit from this perceived increase in subjective value. However, the impact on the trade signal is short lived due to the mean-reverting dynamics, so when the remaining time is long this incentive to buy shares does not arise because it will have disappeared before the advantage can be gained. For longer remaining times the effect of permanent price impact dominates, and the price decrease caused by negative order flow incentivizes the agent to sell shares. A larger value of price impact $b$ makes the decision to trade based on price impact dominate the decision to trade on a change in the trade signal, hence the loading $\nu_{\bar q}^*(t)$ decreases with $b$. The fact that the sensitivity to $b$ is greatest at time $t = 0$ is also explained due to the permanent nature of order flow on price impact along with the transient nature of the effect on the trade signal.

In the right panel, the parameter $b$ is fixed and each curve represents a different value of $\bar{\gamma}$. As $\bar{\gamma}$ is increased, the loading increases, and this effect is most pronounced close to the end of the trading period. The reasoning for the general shape of each curve is the same as in the discussion of the left panel in the previous paragraph. Also based on the same discussion is the reason that the sensitivity towards $\bar{\gamma}$ is greatest close to the end of the trading period. This comes from the transient nature of the impact effects on the trade signal, so when the remaining time is large the agent knows these effects will gradually disappear before they can offer their advantage.

\begin{figure}
	\begin{center}
		{\includegraphics[trim=140 240 140 240, scale=0.4]{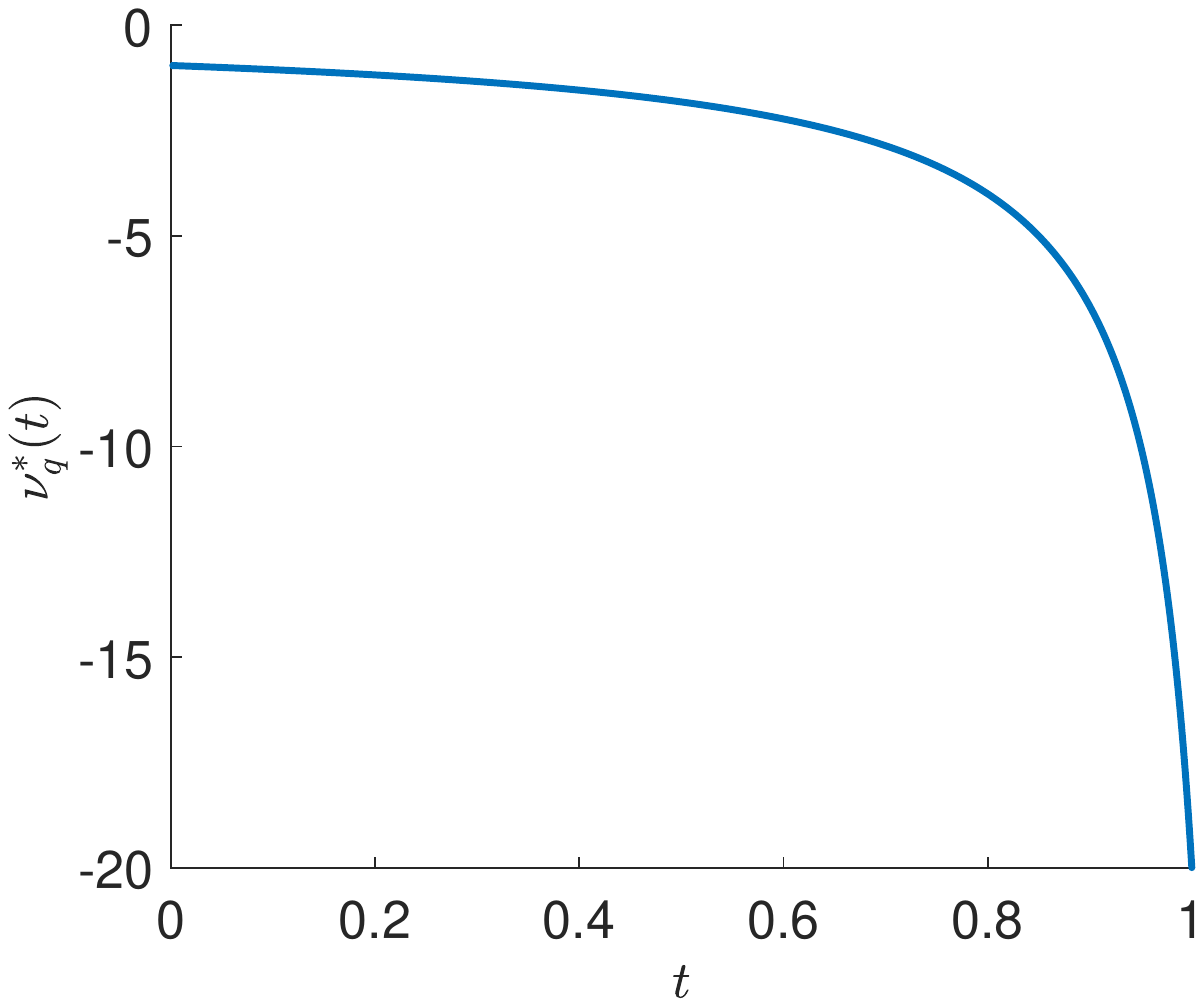}}\hspace{10mm}
		{\includegraphics[trim=140 240 140 240, scale=0.4]{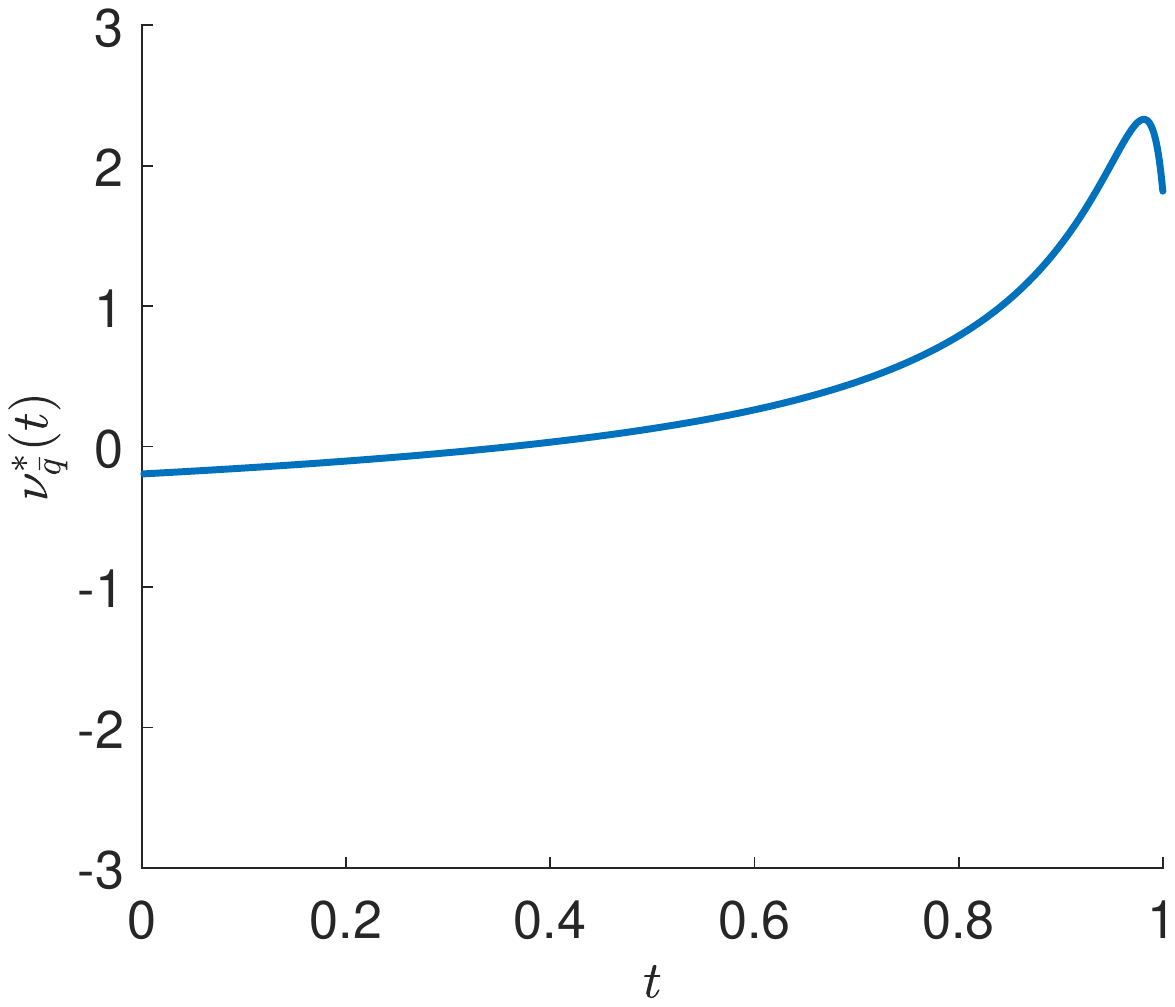}}\hspace{10mm}
		{\includegraphics[trim=140 240 140 240, scale=0.4]{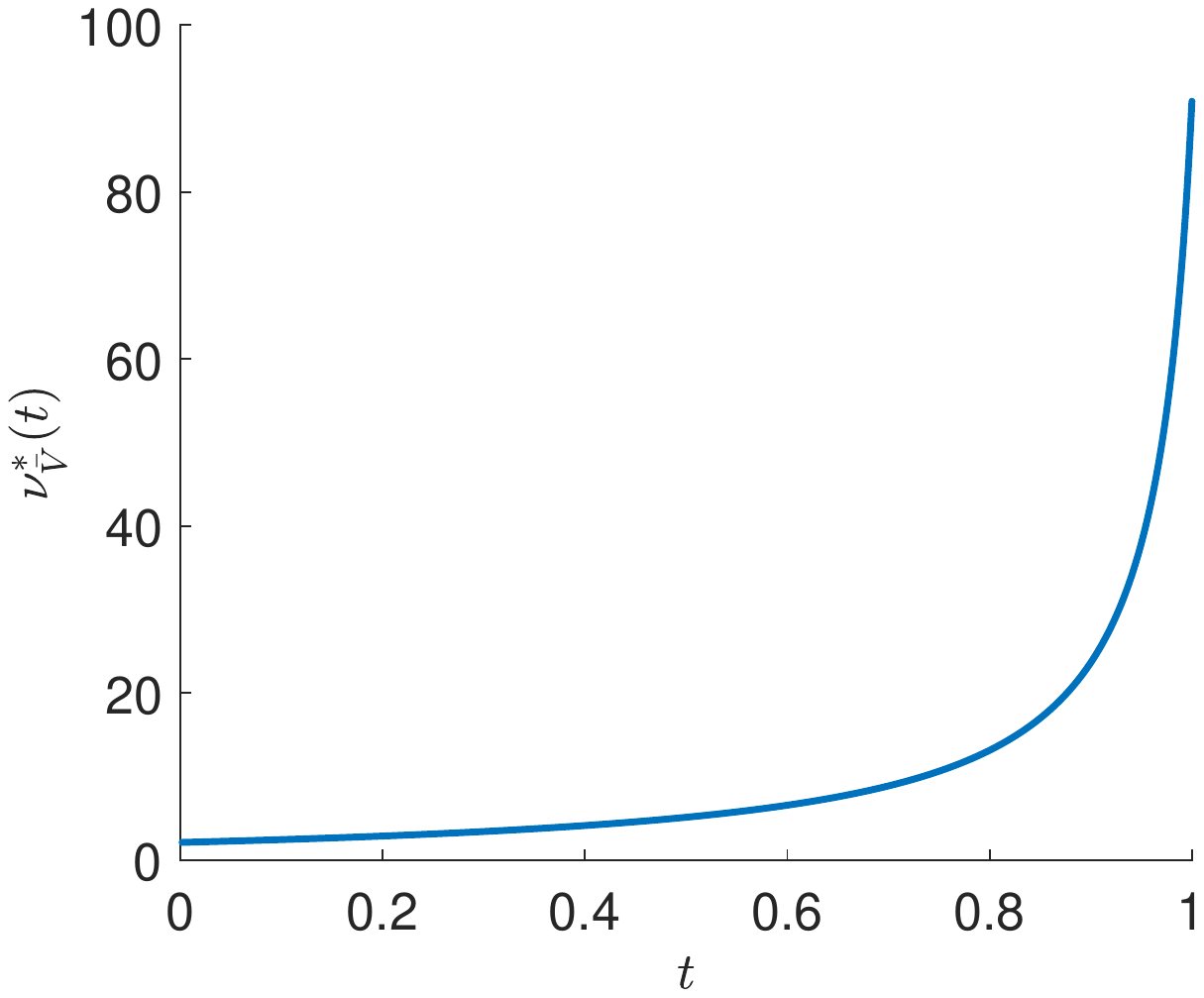}}
	\end{center}
	\vspace{-1em}
	\caption{Optimal loadings on $Q_t$, $\bar{Q}_t$, and $V_t$. Parameters used are $\mu = 0$, $\sigma = 1$, $\eta = 0.5$, $\beta = 1$, $\bar{\gamma} = 0.1$, $\rho = 0.3$, $b = 10^{-2}$, $k = 5\cdot10^{-3}$, $\bar{k} = 10^{-3}$, $\alpha = 0.1$, and $T = 1$. \label{fig:trade-loadings-shared}}
\end{figure}

\begin{figure}
	\begin{center}
		{\includegraphics[trim=140 240 140 240, scale=0.4]{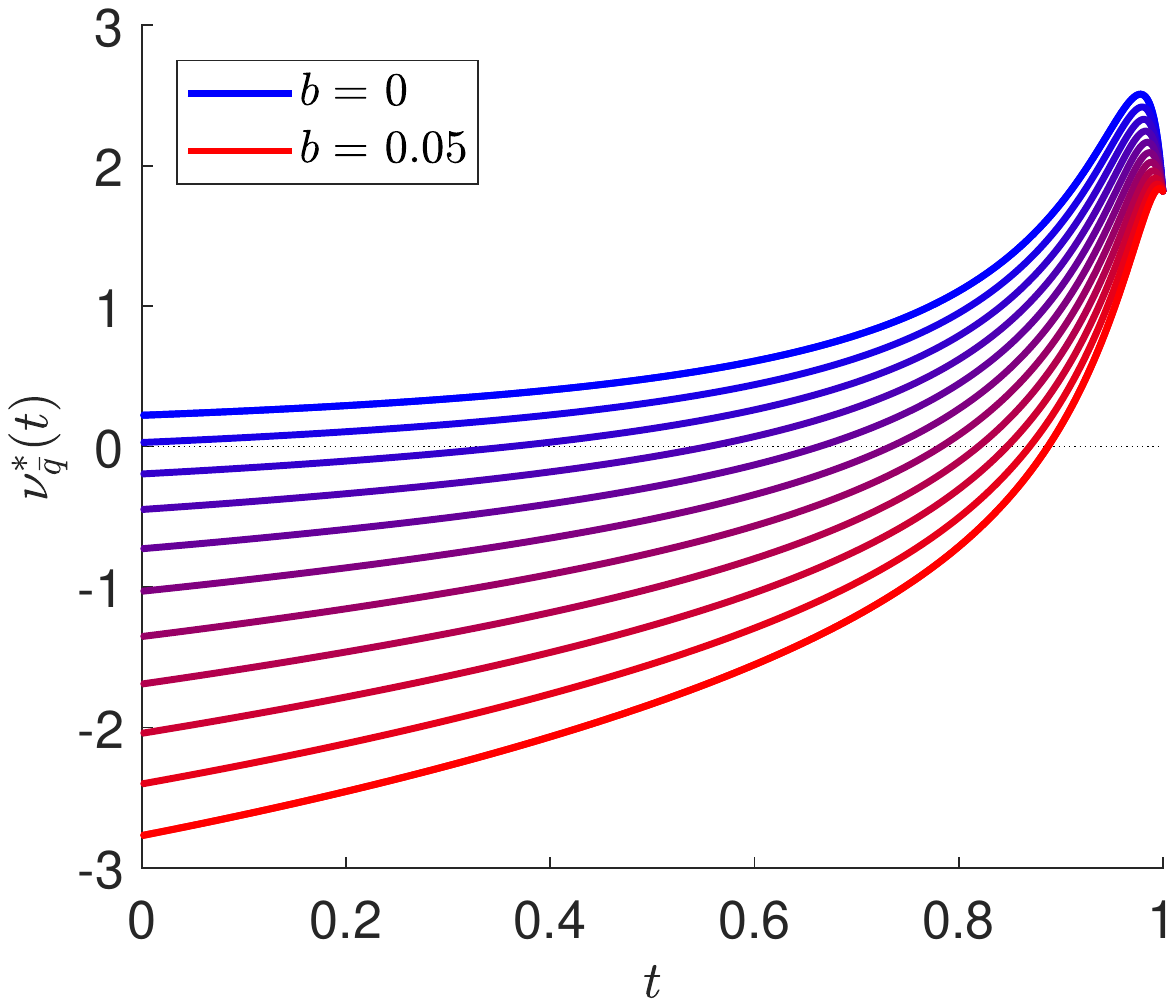}}\hspace{10mm}
		{\includegraphics[trim=140 240 140 240, scale=0.4]{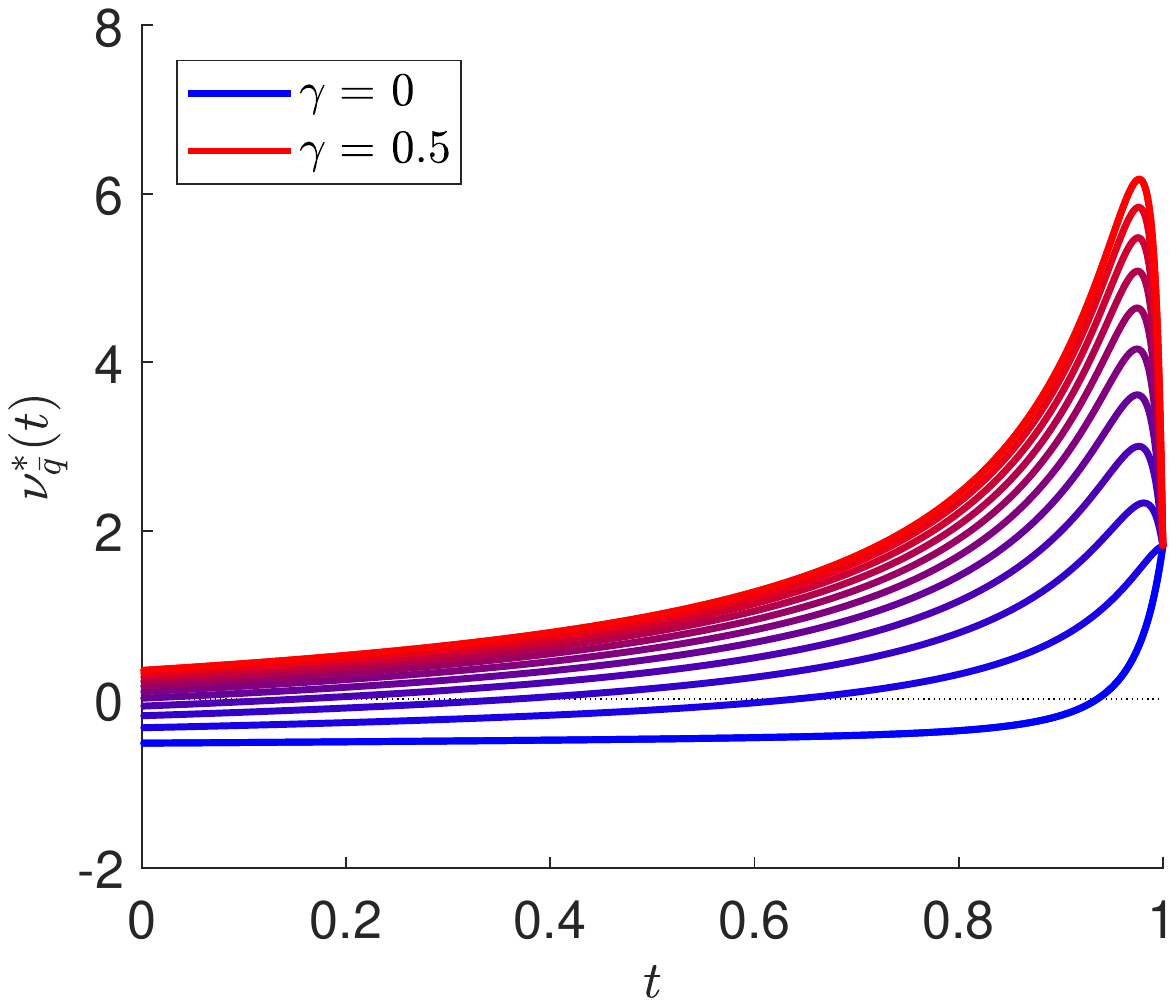}}
	\end{center}
	\vspace{-1em}
	\caption{Optimal loading on $\bar{Q}_t$. Left panel: $\bar{\gamma} = 0.1$ and $b$ ranging from $0$ (blue curve) to $5\cdot10^{-2}$ (red curve). Right panel: $b = 10^{-2}$ and $\bar{\gamma}$ ranging from $0$ (blue curve) to $0.5$ (red curve). Other parameters are $\mu = 0$, $\sigma = 1$, $\eta = 0.5$, $\beta = 1$, $\rho = 0.3$, $k = 5\cdot10^{-3}$, $\bar{k} = 10^{-3}$, $\alpha = 0.1$, and $T = 1$.  \label{fig:nu_shared_qbar_gammab}}
\end{figure}

In Figure \ref{fig:MC-shared} we show the result of a simulation when each agent acts according to the mean-field optimal strategy described by the loadings plotted in Figure \ref{fig:trade-loadings-shared}. The most striking feature of this simulation is that all agents appear to approach very similar terminal inventory holdings even though the initial positions are wide spread. This stems from the fact that they each share the same view of the asset's value. At time $T$, there will be a tradeoff to holding non-zero inventory between the terminal liquidation penalty and the additional value imparted by the signal $\bar{V}_T$. Since each agent assigns the same value to this signal, they are willing to accept the same magnitude of terminal penalty in order to benefit from the signal. In Section \ref{sec:distribution} we show that the minimal cross-sectional variance of inventories is indeed achieved when the agents share the same trade signal.

%
%
%

\begin{figure}
	\begin{center}
		{\includegraphics[trim=140 240 140 240, scale=0.4]{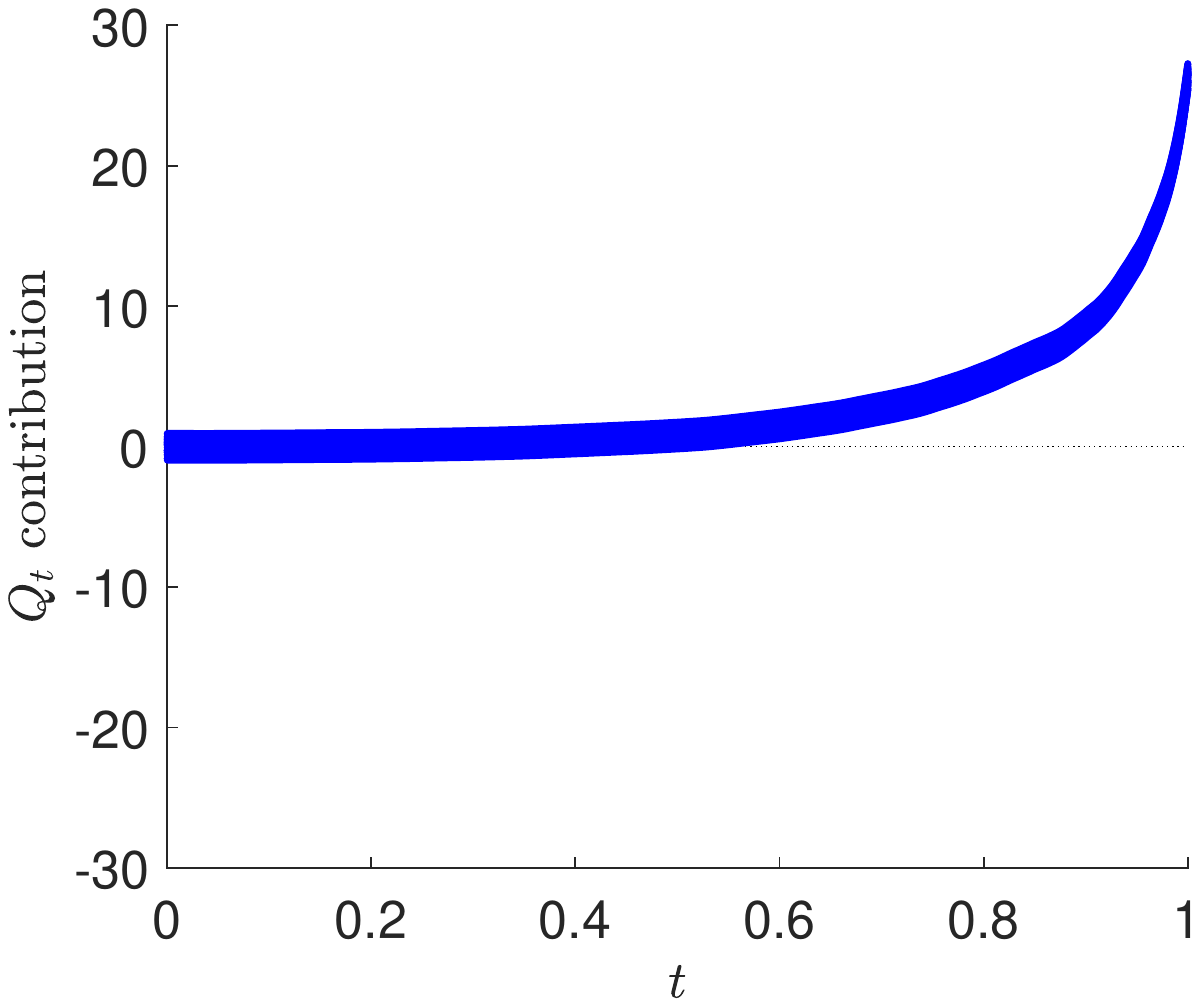}}\hspace{10mm}
		{\includegraphics[trim=140 240 140 240, scale=0.4]{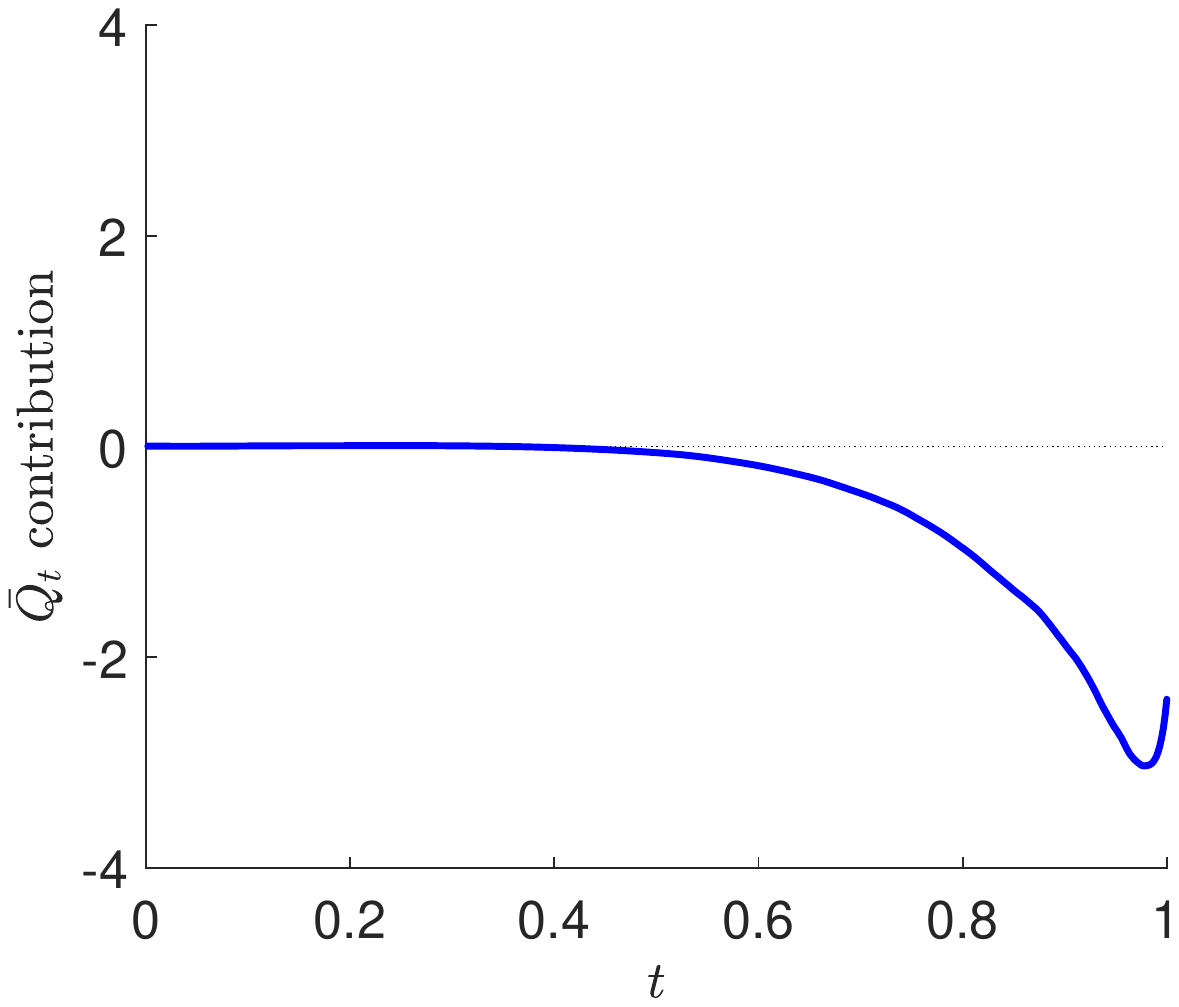}}\hspace{10mm}
		{\includegraphics[trim=140 240 140 240, scale=0.4]{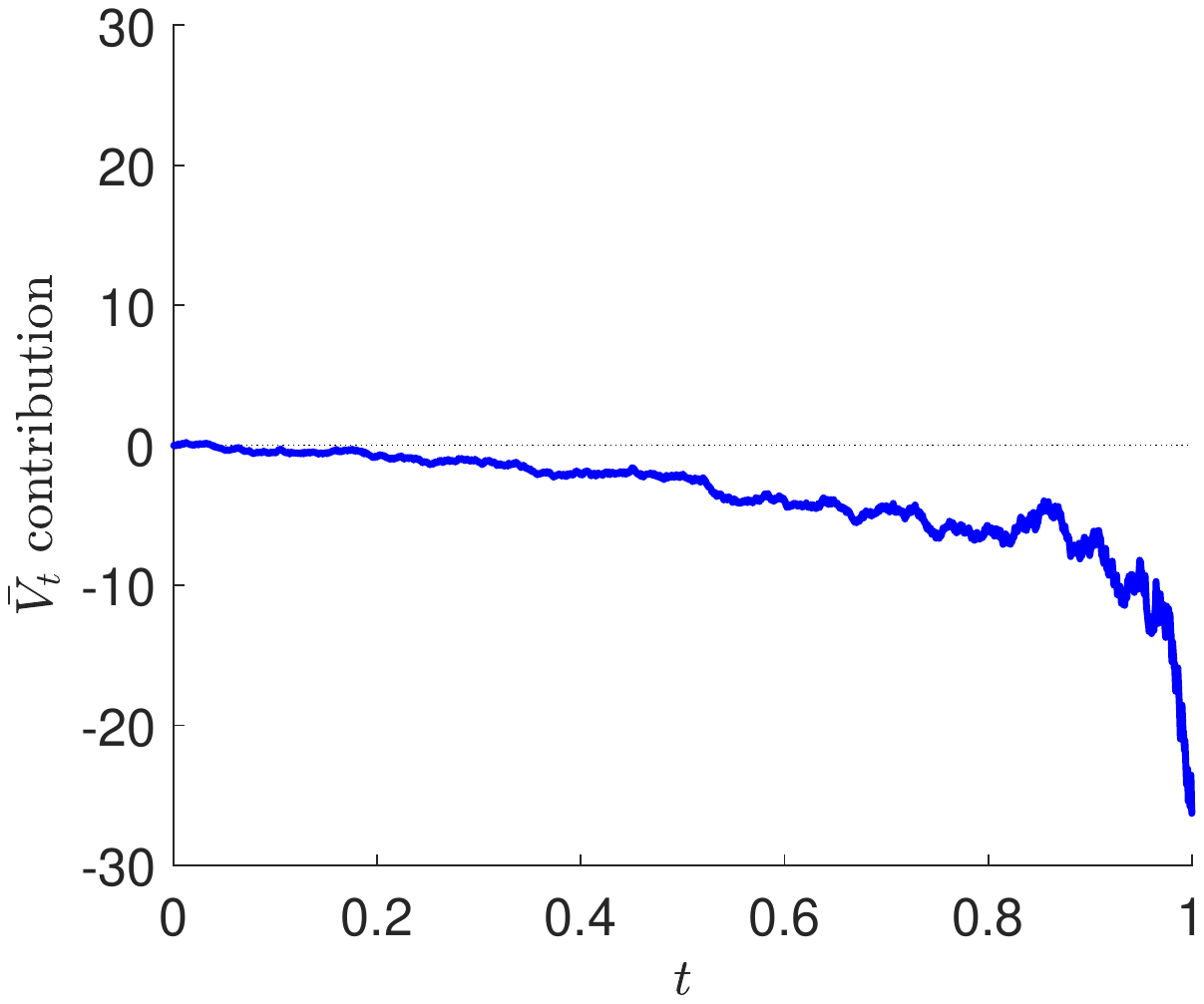}}\hspace{10mm}
		{\includegraphics[trim=140 240 140 240, scale=0.4]{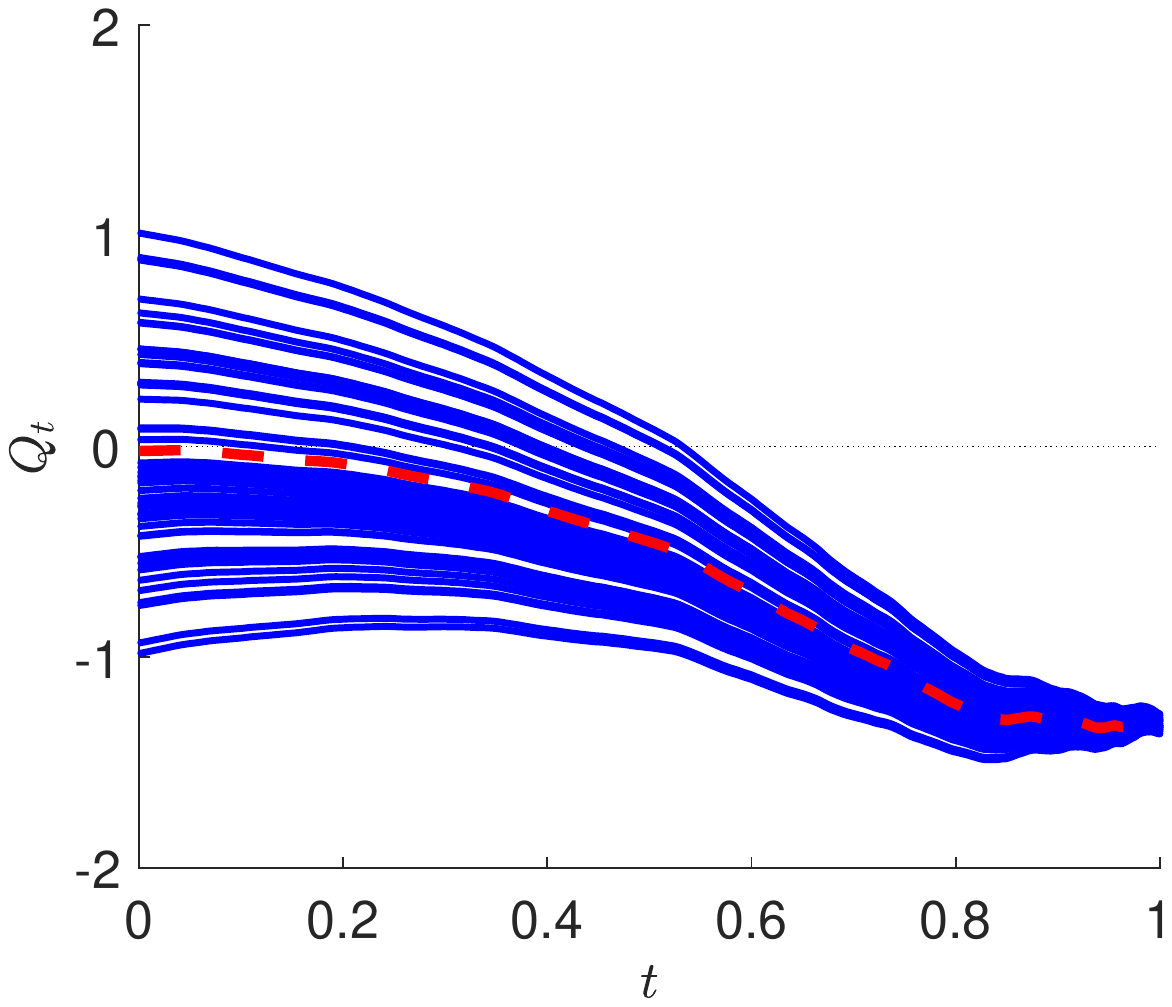}}\hspace{10mm}
		{\includegraphics[trim=140 240 140 240, scale=0.4]{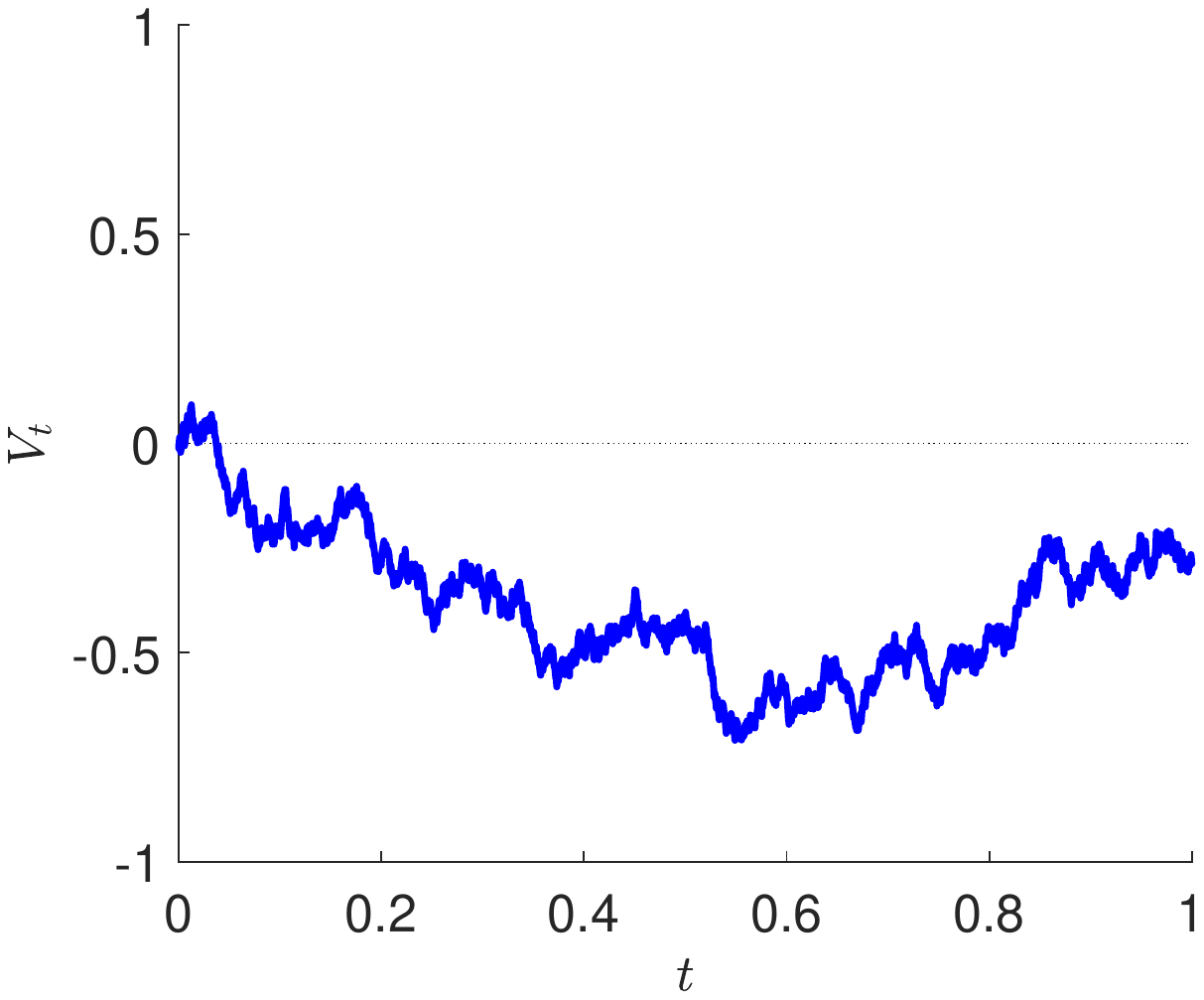}}\hspace{10mm}
		{\includegraphics[trim=140 240 140 240, scale=0.4]{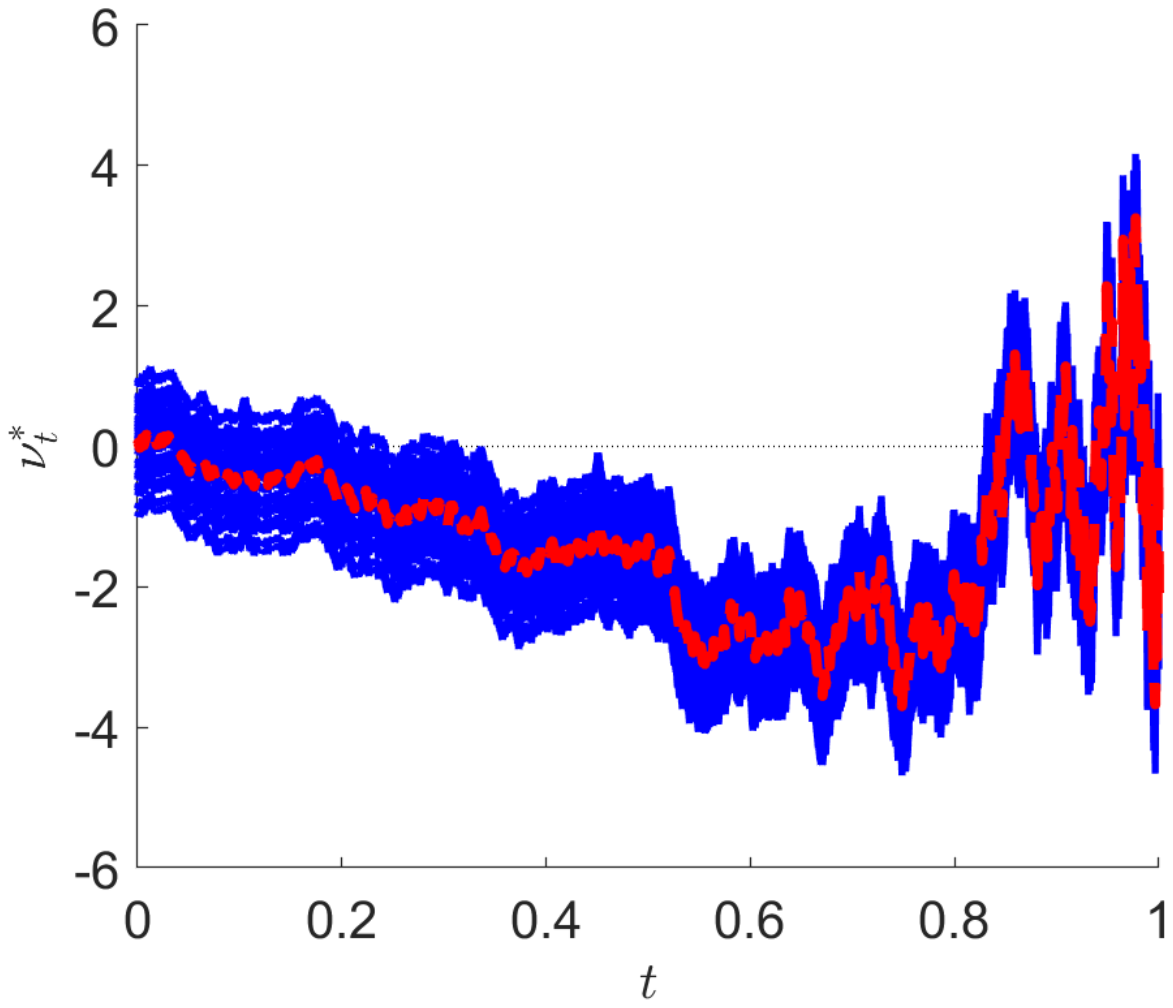}}
	\end{center}
	\vspace{-1em}
	\caption{The top row shows the contributions to trading speed from $Q_t$, $\bar{Q}_t$, and $\bar{V}_t$. The left panel of the second row shows each agent's inventory path $Q_t^{n,\nu^n}$ (blue curves) as well as the average inventory of all agents $\bar{Q}^{\bar{\nu}}_t$ (red dotted curve). The right panel of the second row shows the optimal trading speed $\nu^n_t$ (blue curves) and the average trading speed $\bar{\nu}_t$ (red dotted curve). Parameters used are $\mu = 0$, $\sigma = 1$, $\eta = 0.5$, $\beta = 1$, $\bar{\gamma} = 0.1$, $\rho = 0.3$, $b = 10^{-2}$, $k = 5\cdot10^{-3}$, $\bar{k} = 10^{-3}$, $\alpha = 0.1$, $T = 1$, $S_0 = 100$, $\bar{V}_0 = 0$, $Q_0^n\sim\mathcal{N}(0,0.5^2)$, and $N = 50$. 
\label{fig:MC-shared}}
\end{figure}

\newpage

\subsection{Separate Subjective View of Asset Value}\label{sec:separate}

Here we consider a model in which each agent has his own individual trading signal, each of which changes according to a different stochastic process. For agent $n$ we denote his trading signal by $V^{n,\nu^n,\bar{\nu}} = (V^{n,\nu^n,\bar{\nu}}_t)_{0 \leq t \leq T}$ which changes according to
\begin{align}
	\dd V^{n,\nu^n,\bar{\nu}}_t &= -(\beta V_t^{n,\nu^n,\bar{\nu}} + \gamma \nu_t^n + \bar{\gamma}\bar{\nu}_t) \dd t + \eta \dd Z^n_t\,, & V_0^{n,\nu^n,\bar{\nu}} &= V_0^n\,,\label{eqn:V_separate}\\
	Z_t^n &= \rho W_t + \sqrt{1-\rho^2}W_t^{n,\perp}\,,
\end{align}
where each $W^{n,\perp} = (W^{n,\perp}_t)_{0\leq t \leq T}$ is a Brownian motion, independent of one another for different $n$, and independent of $W$. In addition we assume that all $V_0^{n}$ are i.i.d. with finite expectation and variance, and independent from all other variables. 
Inventory and price dynamics are equivalent to those of Section \ref{sec:multiple-agents-model}. In this section it will be useful to consider the average value of the trading signals over all agents which will be denoted $\bar{V}^{\bar{\nu}} = (\bar{V}^{\bar{\nu}}_t)_{0 \leq t \leq T}$ and is defined by
\begin{align}
	\bar{V}^{\bar{\nu}}_t &:= \frac{1}{N}\sum_{n=1}^N V^{n,\nu^n,\bar{\nu}}_t\,.\label{eqn:V_bar-definition}
\end{align}
Similar to the section in which the agents share the same trade signal, the average signal $\bar{V}^{\bar{\nu}}_t$ may be thought of as a commonly identified quantity which conveys information about order flow or price dynamics, such as order book imbalance or micro-price. Then the individual signals $V^{n,\nu^n,\bar{\nu}}_t$ may be interpreted as quantities used by each individual agent believing they have a modification to the common trade signal which represents an improvement.

Based on the definition \eqref{eqn:V_bar-definition} we may also compute the dynamics to be
\begin{align}
	\dd \bar{V}^{\bar{\nu}}_t &= \frac{1}{N}\sum_{n=1}^N \dd V^{n,\nu^n,\bar{\nu}}_t\\
	&= -(\beta \bar{V}_t^{\bar{\nu}} + (\gamma + \bar{\gamma})\bar{\nu}_t)\dd t + \frac{\eta}{N}\sum_{n=1}^N \dd Z_t^n\\
	&= -(\beta \bar{V}_t^{\bar{\nu}} + (\gamma + \bar{\gamma})\bar{\nu}_t)\dd t + \eta\rho \dd W_t + \frac{\eta\sqrt{1-\rho^2}}{N}\sum_{n=1}^N \dd W^{n,\perp}_t\, .
\end{align}
Due to the independence of each $W^{n,\perp}$, when we consider the limit $N\rightarrow\infty$ the last term above becomes zero due to the law of large numbers. It is worth making the brief remark that this model of separate trade signals can be reduced to the shared signal of Section \ref{sec:shared} by choosing some parameter values in a particular way. Specifically, if each $V_0^n$ in \ref{eqn:V_separate} and $V_0$ of \ref{eqn:V_shared} are equal to a constant (the same constant for each $n$), and if $\gamma=0$ and $\rho=\pm1$ (in both models), then we have $V^{n,\nu^n,\bar{\nu}}_t = \bar{V}^{\bar{\nu}}_t$ and every agent observes the same trade signal, which is the setting considered in Section \ref{sec:shared}.

Each agent attempts to maximize his own expected future wealth given that the trading strategies of all other agents are fixed. That is, if $\nu^{-n}$ is fixed, agent $n$ wishes to maximize the functional
\begin{align}
	J(\nu^n;\nu^{-n}) &= \Eb \Big( X_T^{n,\nu^n,\bar{\nu}} + Q_T^{n,\nu^n} (S_T^{\bar{\nu}} + V_T^{n,\nu^n,\bar{\nu}})  - \alpha (Q_T^{n,\nu^n})^2\Big) \,.
\end{align}

For the remainder of Section \ref{sec:separate} we work with a complete and filtered probability space $(\Omega, (\mathcal{F}_t)_{0\leq t \leq T},\mathbb{P})$ where $(\mathcal{F}_t)_{0\leq t \leq T}$ is the standard augmentation of the natural filtration generated by $(W_t,Z^n_t)_{0\leq t\leq T,n\in\mathbb{N}}$ and the initial state $(S_0,(Q^n_0)_{n\in\mathbb{N}},(X^n_0)_{n\in\mathbb{N}},(V^n_0)_{n\in\mathbb{N}})$.

\subsubsection{HJB Equation and Consistency Condition with Separate Subjective Views}

With a similar approach to Section \ref{sec:HJB-shared} we consider a solution in the limiting case $N\rightarrow\infty$. We assume that the average trading speed is of the form $\bar{\nu}_t = \bar{\nu}(t,\bar{Q}^{\bar{\nu}}_t,\bar{V}^{\bar{\nu}}_t)$ to remain within a Markovian framework, similar to the assumption made in Section \ref{sec:HJB-shared}. With this function fixed we define the value function for agent $n$ as

\begin{align}
	H^n(t,x,q,\bar{q},S,V,\bar{V};\bar{\nu}) &:= \sup_{\nu^n \in \Nc} \Eb_{t,x,q,\bar{q},S,V,\bar{V}} \Big( X_T^{n,\nu^n,\bar{\nu}} + Q_T^{n,\nu^n} (S_T^{\bar{\nu}} + V_T^{n,\nu^n,\bar{\nu}}) - \alpha (Q_T^{n,\nu^n})^2\Big) , \label{def:Hn-separate}
\end{align}
where the set of admissible strategies $\mathcal{N}$ consists of all $\mathcal{F}$-predictable processes such that $\mathbb{E}[\int_0^T (\nu^n_t)^2\,dt]<\infty$.
The value function in \eqref{def:Hn-separate} has an associated HJB equation of the form
\begin{align}
	\d_t H^n + \sup_{\nu^n \in \Rb} ( \Ac^{\nu^n,\bar{\nu}} H^n  )
		&=	0 , &
	H^n(T,x,q,\bar{q},S,V,\bar{V};\bar{\nu})
	&=	x + q(S+V) - \alpha q^2, \label{eq:hjb-pde-separate}
\end{align}
where  the operator $\Ac^{\nu^n,\bar{\nu}}$ is given by
\begin{align}
\Ac^{\nu^n,\bar{\nu}}
	&=	-(S+k\nu^n+\bar{k}\bar{\nu})\nu^n\d_x + \nu^n\d_q + \bar{\nu}\d_{\bar{q}} + (\mu + b\bar{\nu})\d_S - (\beta V + \gamma \nu^n + \bar{\gamma}\bar{\nu})\d_V - (\beta \bar{V} + (\gamma+\bar{\gamma})\bar{\nu})\d_{\bar{V}}\\
	& \hspace{10mm} +  \frac{1}{2}\sigma^2\d_{SS} + \frac{1}{2}\eta^2\d_{VV} + \frac{1}{2}\rho^2\eta^2 \d_{\bar{V}\bar{V}} + \rho\sigma\eta\d_{SV} + \rho\sigma\eta\d_{S\bar{V}} + \rho^2\eta^2\d_{V\bar{V}}\,.
\end{align}
Based on the form of the feedback control in the previous sections, we make the ansatz
\begin{align}
	\bar{\nu}(t,\bar{q},\bar{V}) &= f_1(t) + f_2(t)\bar{q} + f_3(t)\bar{V}\,. \label{eq:nubar-form-separate}
\end{align}
With this ansatz the solution to the HJB equation \eqref{eq:hjb-pde-separate} along with the optimal control in feedback form can be characterized by a solution to a system of ODE's.
\begin{proposition}
	Given $\bar{\nu}$ in \eqref{eq:nubar-form-separate}, suppose $c_1,\dots,c_{15}:[0,T]\rightarrow\mathbb{R}$ satisfy the following system of ODEs with terminal conditions:
	\begin{align}
				c'_1+f_1(c_3-\bar{\gamma}c_4-(\gamma+\bar{\gamma})c_5)+\eta^2 c_8+\rho^2\eta^2(c_9+c_{15})+\frac{(c_2-\gamma c_4 - \bar{k}f_1)^2}{4k}&=0\,,																	& c_1(T) &= 0\,, \label{eqn:separate-c1}\\
				c'_2+\mu+f_1(b+c_{10}-\bar{\gamma}c_{11}-(\gamma+\bar{\gamma})c_{12})+\frac{(2c_6-\gamma c_{11})(c_2-\gamma c_4 - \bar{k}f_1)}{2k}&=0\,,																	& c_2(T) &= 0\,,\\
				c'_3+f_1(2 c_7-\bar{\gamma}c_{13}-(\gamma+\bar{\gamma})c_{14})+f_2(c_3-\bar{\gamma}c_4-(\gamma+\bar{\gamma})c_5)+\frac{(c_2-\gamma c_4 - \bar{k}f_1)(c_{10}-\gamma c_{13} - \bar{k}f_2)}{2k}&=0\,,			& c_3(T) &= 0\,,\\
				c'_4-\beta c_4+f_1(c_{13}-2\bar{\gamma}c_8-(\gamma+\bar{\gamma})c_{15})+\frac{(c_2-\gamma c_4 - \bar{k}f_1)(c_{11}-2\gamma c_8)}{2k}&=0\,,																	& c_4(T) &= 0\,,\\
				c'_5-\beta c_5+f_1(c_{14}-\bar{\gamma}c_{15}-2(\gamma+\bar{\gamma})c_9)+f_3(c_3-\bar{\gamma}c_4-(\gamma+\bar{\gamma})c_5)+\frac{(c_2-\gamma c_4 - \bar{k}f_1)(c_{12}-\gamma c_{15}-\bar{k}f_3)}{2k}&=0\,,	& c_5(T) &= 0\,,\\
				c'_6+\frac{(2c_6-\gamma c_{11})^2}{4k}&=0\,,																																								& c_6(T) &= -\alpha\,,\\
				c'_7+f_2(2c_7-\bar{\gamma}c_{13}-(\gamma+\bar{\gamma})c_{14})+\frac{(c_{10}-\gamma c_{13} - \bar{k}f_2)^2}{4k}&=0\,,																						& c_7(T) &= 0\,,\\
				c'_8-2\beta c_8+\frac{(c_{11}-2\gamma c_8)^2}{4k}&=0\,,																																						& c_8(T) &= 0\,,\\
				c'_9-2\beta c_9+f_3(c_{14}-\bar{\gamma}c_{15}-2(\gamma+\bar{\gamma})c_9)+\frac{(c_{12}-\gamma c_{15} - \bar{k}f_3)^2}{4k}&=0\,,																				& c_9(T) &= 0\,,\\
				c'_{10}+f_2(b+c_{10}-\bar{\gamma}c_{11}-(\gamma+\bar{\gamma})c_{12})+\frac{(2c_6-\gamma c_{11})(c_{10}-\gamma c_{13} - \bar{k}f_2)}{2k}&=0\,,																& c_{10}(T) &= 0\,,\\
				c'_{11}-\beta c_{11}+\frac{(2c_6-\gamma c_{11})(c_{11}-2\gamma c_8)}{2k}&=0\,,																																& c_{11}(T) &= 1\,,\\
				c'_{12}-\beta c_{12}+f_3(b+c_{10}-\bar{\gamma}c_{11}-(\gamma+\bar{\gamma})c_{12})+\frac{(2c_6-\gamma c_{11})(c_{12}-\gamma c_{15} - \bar{k}f_3)}{2k}&=0\,,													& c_{12}(T) &= 0\,,\\
				c'_{13}-\beta c_{13}+f_2(c_{13}-2\bar{\gamma}c_8-(\gamma+\bar{\gamma})c_{15}) + \frac{(c_{11}-2\gamma c_8)(c_{10}-\gamma c_{13} - \bar{k}f_2)}{2k}&=0\,,													& c_{13}(T) &= 0\,,\\
\hspace{-10mm}	c'_{14}-\beta c_{14}+f_2(c_{14}-\bar{\gamma}c_{15}-2(\gamma+\bar{\gamma})c_9)+f_3(2c_7-\bar{\gamma}c_{13}-(\gamma+\bar{\gamma})c_{14})+\frac{(c_{10}-\gamma c_{13} - \bar{k}f_2)(c_{12}-\gamma c_{15} - \bar{k}f_3)}{2k} &= 0\,,	& c_{14}(T) &= 0\,,\\
				c'_{15}-2\beta c_{15}+f_3(c_{13}-2\bar{\gamma}c_8-(\gamma+\bar{\gamma})c_{15})+\frac{(c_{11}-2\gamma c_8)(c_{12}-\gamma c_{15} - \bar{k}f_3)}{2k}&= 0\,,													& c_{15}(T) &= 0\,.\label{eqn:separate-c15}
	\end{align}
	Then the value function $H^n$ is given by
	\begin{align}
		H^n(t,x,q,\bar{q},S,V,\bar{V}) &= x + qS + h^n(t,q,\bar{q},V,\bar{V})\,,\\
		h^n(t,q,\bar{q},V,\bar{V}) &= c_1(t) + c_2(t)q + c_3(t)\bar{q} + c_4(t)V + c_5(t)\bar{V}  + c_6(t)q^2 + c_7(t)\bar{q}^2 + c_8(t)V^2 \\
		&\quad + c_9(t)\bar{V}^2 + c_{10}(t)q\bar{q} + c_{11}(t)qV + c_{12}(t)q\bar{V} + c_{13}(t)\bar{q}V + c_{14}(t)\bar{q}\bar{V} + c_{15}(t)V\bar{V}\,,
	\end{align}
	and the optimal trading strategy in feedback form is
	\begin{align}
	\nu^{n*}(t,q,\bar{q},V,\bar{V}) = \frac{c_2(t)-\gamma c_4(t) - \bar{k}f_1(t)}{2k} + \frac{2c_6(t)-\gamma c_{11}(t)}{2k}q + \frac{c_{10}(t)-\gamma c_{13}(t) - \bar{k}f_2(t)}{2k}\bar{q}\nonumber\\
	\hspace{10mm} + \frac{c_{11}(t)-2\gamma c_8(t)}{2k}V + \frac{c_{12}-\gamma c_{15}(t) - \bar{k}f_3(t)}{2k}\bar{V}\,. \label{eq:nustar-separate}
	\end{align}
\end{proposition}
\begin{proof}
	This is shown by direct substitution into equation \eqref{eq:hjb-pde-separate}.
\end{proof}
In a similar fashion to the previous section, we require a consistency condition to be satisfied in order for the trading strategy in \eqref{eq:nustar-separate} to yield a mean-field Nash equilibrium. The strategy \eqref{eq:nustar-separate} is based on the ansatz that the average trading speed is given by \eqref{eq:nubar-form-separate}, therefore we must impose that when each agent uses the strategy \eqref{eq:nustar-separate} the resulting average trading speed is \eqref{eq:nubar-form-separate}.  Thus, we require
\begin{align}
	\lim_{N\rightarrow\infty}\frac{1}{N}\sum_{n=1}^N \nu^{n*}(t,q^n,\bar{q},V^n,\bar{V}) = \bar{\nu}(t,\bar{q},\bar{V})\,.
\end{align}
Substituting \eqref{eq:nubar-form} and \eqref{eq:nustar-shared} into this equation yields
\begin{align}
	f_1 &= \frac{c_2-\gamma c_4}{2k+\bar{k}}\,, &	f_2 &= \frac{2c_6+c_{10} - \gamma(c_{11}+c_{13})}{2k+\bar{k}}\,, & f_3 &= \frac{c_{11}+c_{12}-\gamma(2c_8+c_{15})}{2k+\bar{k}}\,.\label{eqn:separate-f}
\end{align}
As we did in the previous section, we will only consider solutions of \eqref{eqn:separate-c1} to \eqref{eqn:separate-c15} in which \eqref{eqn:separate-f} has been enforced. This means we only consider optimal trading strategies that result in equilibrium. Also as in the previous section, the trading strategies in a mean-field Nash equilibrium can be written in a particular form
\begin{proposition}
	In equilibrium, the trading strategy of agent $n$ and the average trading rate of all agents are related by
	\begin{align}
		\nu^{n*}(t,q^n,\bar{q},V^n,\bar{V}) &= \frac{2c_6(t)-\gamma c_{11}(t)}{k}(q^n-\bar{q}) + \frac{c_{11}(t)-2\gamma c_8(t)}{2k}(V^n-\bar{V}) + \bar{\nu}(t,\bar{q},\bar{V})
	\end{align}
\end{proposition}
\begin{proposition}\label{prop:separate_closed}
	If $\alpha = \gamma = 0$ then $c_3 = c_6 = c_7 = c_{10} = c_{13} = c_{14} \equiv 0$ in equilibrium. If $\mu = 0$ then $c_2 = c_3 = c_4 = c_5 \equiv 0$ in equilibrium. If $\alpha=\gamma=\mu=0$ then the non-zero $c_i$ in equilibrium are given by
	\begin{align}
		c_1(t) &= \int_t^T \eta^2 c_8(s) + \rho^2\eta^2 (c_9(s) + c_{15}(s)) \dd s\,,\\
		c_8(t) &= \frac{T-t}{4k}e^{-2\beta(T-t)}\,,\\
		c_9(t) &= \frac{-4z^2e^{-\frac{2b}{\kappa}(T-t)}}{((1+2z)e^{\omega(T-t)}-1)^2}D_9(t)\,,\\
		c_{11}(t) &= e^{-\beta(T-t)}\,,\\
		c_{12}(t) &= \frac{2z}{(1+2z)e^{\omega(T-t)}-1}-e^{-\beta(T-t)}\,,\\
		c_{15}(t) &= -\frac{(T-t)}{2k}e^{-2\beta(T-t)} + \biggl(\frac{2z}{(1+2z)e^{\omega(T-t)}-1}\biggr)\biggl(\frac{1-e^{-\frac{b}{\kappa}(T-t)}}{b}\biggr)e^{-\beta(T-t)}\,,
	\end{align}
	where
	\begin{align}
		D_9(t) &= \frac{1}{16kz^2}\biggl(\frac{1-e^{-2\omega\tau}}{2\omega}-\tau e^{-2\omega\tau}\biggr) - \frac{1+2z}{8kz^2}\biggl(\frac{1-e^{-\omega\tau}}{\omega}-\tau e^{-\omega\tau}\biggr) + \frac{1}{2bz}\biggl(1-e^{(\frac{2b}{\kappa}-\beta)\tau}\biggr) + \frac{k}{2b\kappa}\biggl(1-e^{\frac{2b}{\kappa}\tau}\biggr)\\
		& \hspace{10mm} - \frac{1+2z}{2bz}\biggl(1-e^{\frac{b}{\kappa}\tau}\biggr)-\frac{1}{32kz^2\omega}\biggl(1-e^{-2\omega\tau}\biggr) + \frac{(1+2z)b-4kz\omega}{8bkz^2\omega}\biggl(1-e^{-\omega\tau}\biggr)-\frac{(1+2z)^2}{16kz^2}\tau\,,
	\end{align}
	\begin{align*}
		\kappa &= 2k + \bar{k}\,, & z &= \frac{\kappa\beta - b}{2\bar{\gamma}}\,, & \omega &= \frac{\kappa\beta-b}{\kappa}\,, & \tau &= T-t\,.
	\end{align*}
\end{proposition}
\begin{proof}
	In equations \eqref{eqn:separate-c1} to \eqref{eqn:separate-c15} we substitute \eqref{eqn:separate-f}. The first two conclusions can be seen by inspection. The expressions for the non-zero $c_i$ come from a tedious computation, but can be checked by direct substitution.
\end{proof}

\subsubsection{Numerical Experiments}

We consider again the loadings of the optimal strategy on the underlying processes. Note that from \eqref{eq:nustar-separate} we have
\begin{align}
		\nu^{n*}(t,q,\bar{q},V,\bar{V}) 
			&= \frac{c_2(t)-\gamma c_4(t)}{2k+\bar{k}} + \nu^*_q(t) q + \nu^*_{\bar{q}}(t)\bar{q} + \nu^*_V(t)V + \nu^*_{\bar{V}}(t) \bar{V}\, ,
\end{align}
where the loadings on $q$, $\bar{q}$, $V$ and $\bar{V}$ are given by
\begin{align}
	\nu^*_q(t) &= \frac{2c_6(t)-\gamma c_{11}(t)}{2k}\,, & \nu^*_{\bar{q}}(t) &= \frac{2k(c_{10}(t)-\gamma c_{13}(t)) - \bar{k}(2c_6(t) - \gamma c_{11}(t))}{2k(2k+\bar{k})}\,, \\
	\nu^*_V(t) &= \frac{c_{11}(t)-2\gamma c_8(t)}{2k}\,, & \nu^*_{\bar{V}}(t) &= \frac{2k(c_{12}(t)-\gamma c_{15}(t)) - \bar{k}(c_{11}(t) - 2\gamma c_8(t))}{2k(2k+\bar{k})}\, .
\end{align}
We plot the above four loadings in Figure \ref{fig:trade-loadings-separate}. The first three loadings are qualitatively similar to the situation where each agent's subjective valuation is governed by the same process, but the intuition behind understanding these loadings is more effectively shown by considering various values of some of the relevant parameters.

\begin{figure}
	\begin{center}
		{\includegraphics[trim=140 240 140 240, scale=0.4]{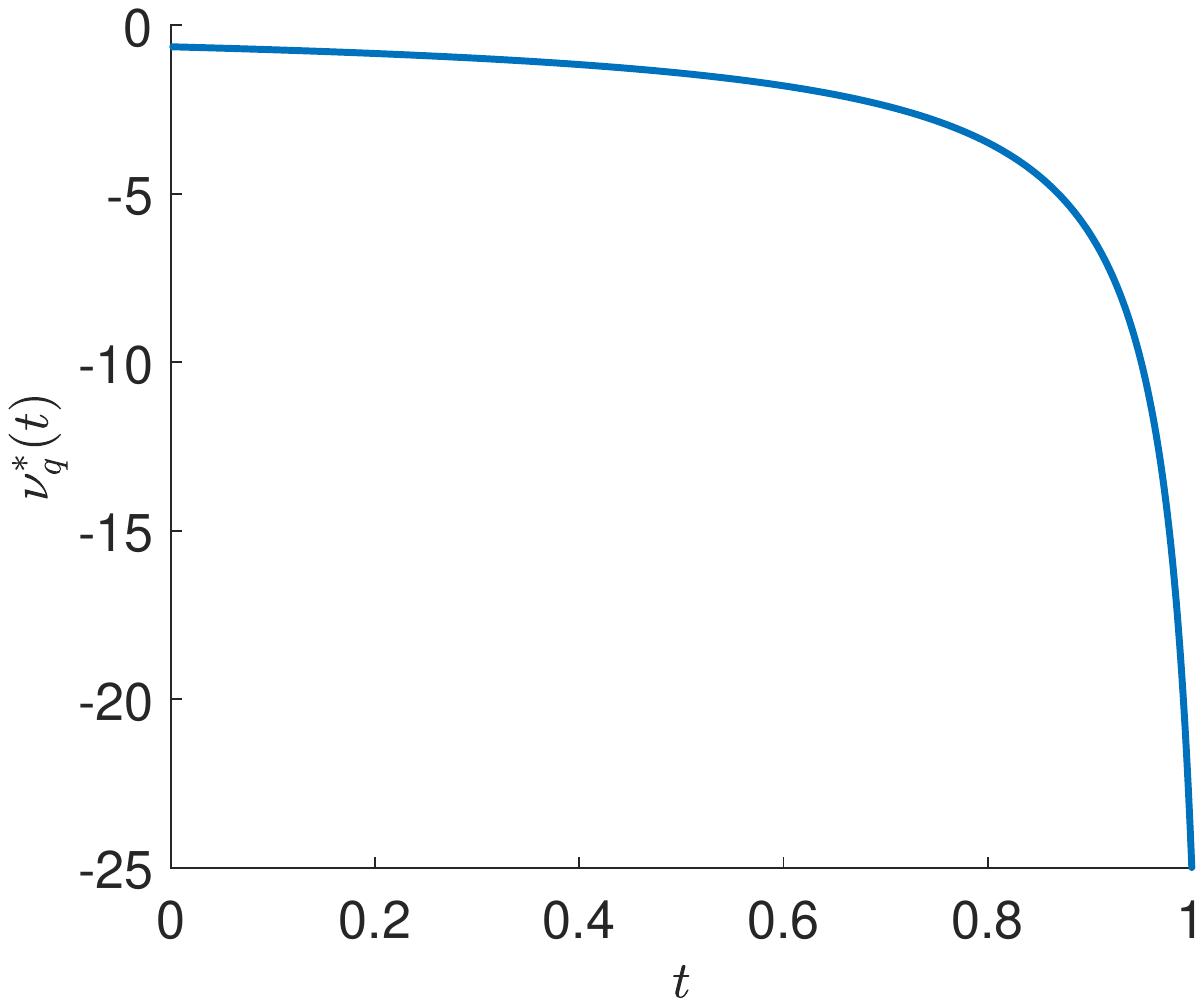}}\hspace{10mm}
		{\includegraphics[trim=140 240 140 240, scale=0.4]{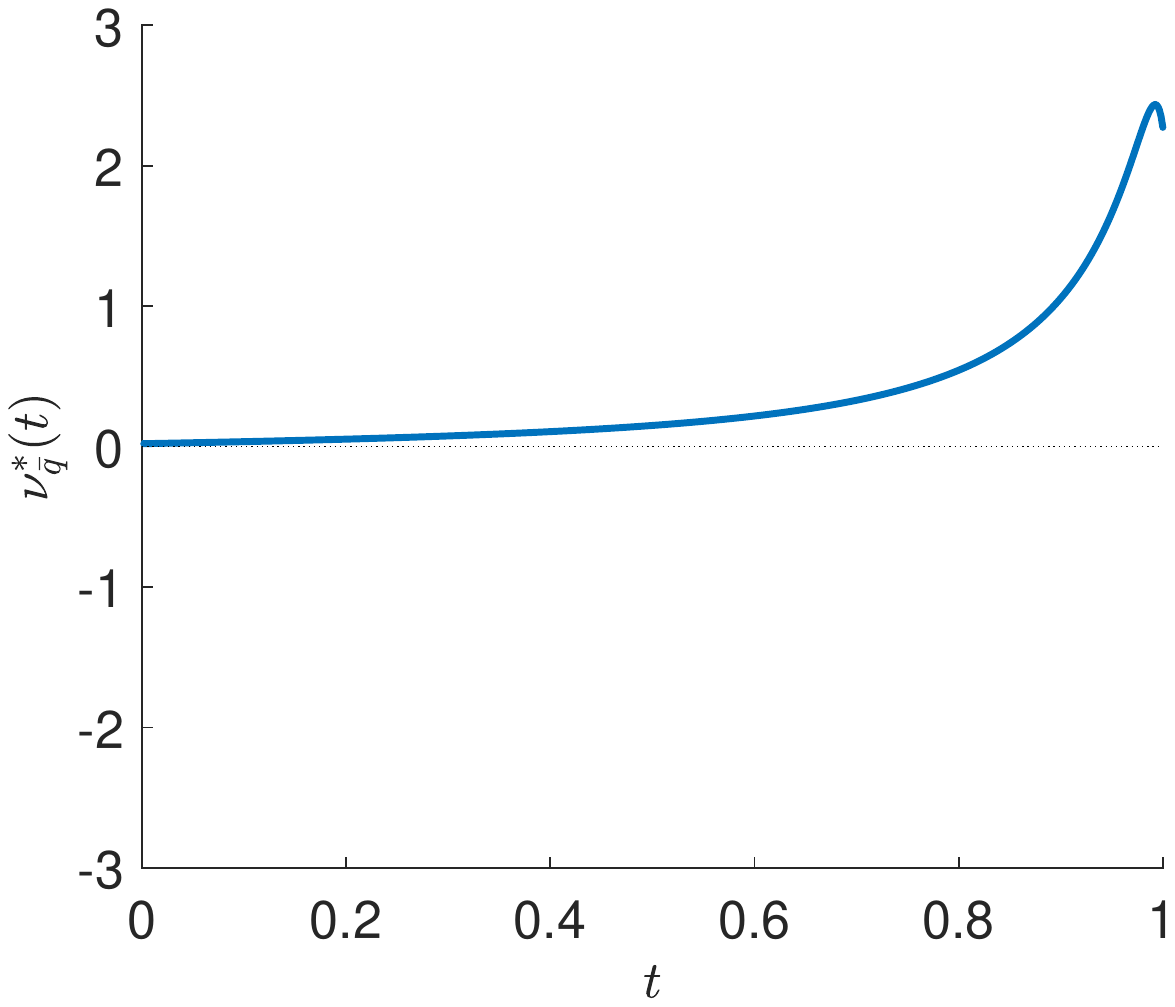}}\\
		{\includegraphics[trim=140 240 140 240, scale=0.4]{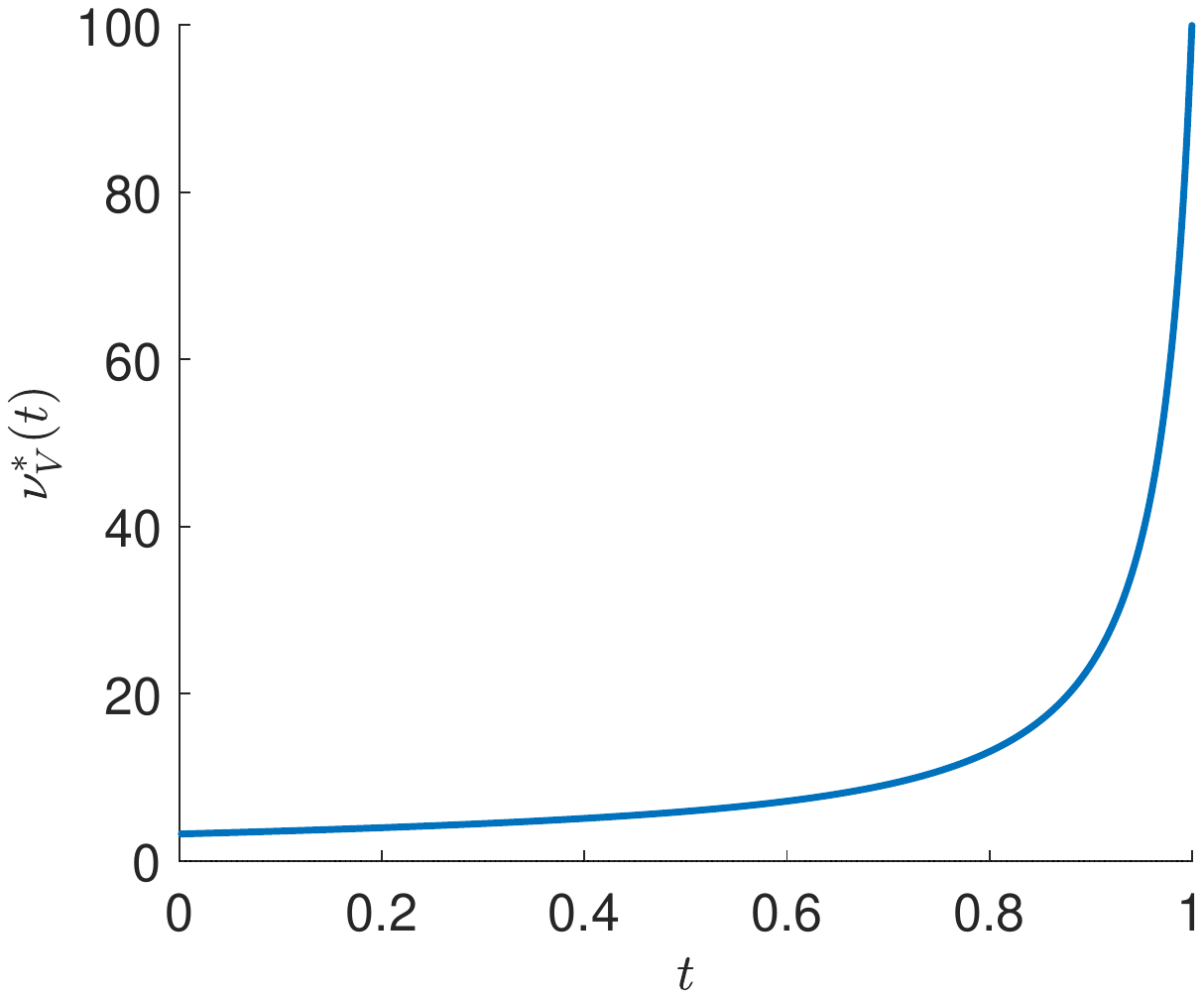}}\hspace{10mm}
		{\includegraphics[trim=140 240 140 240, scale=0.4]{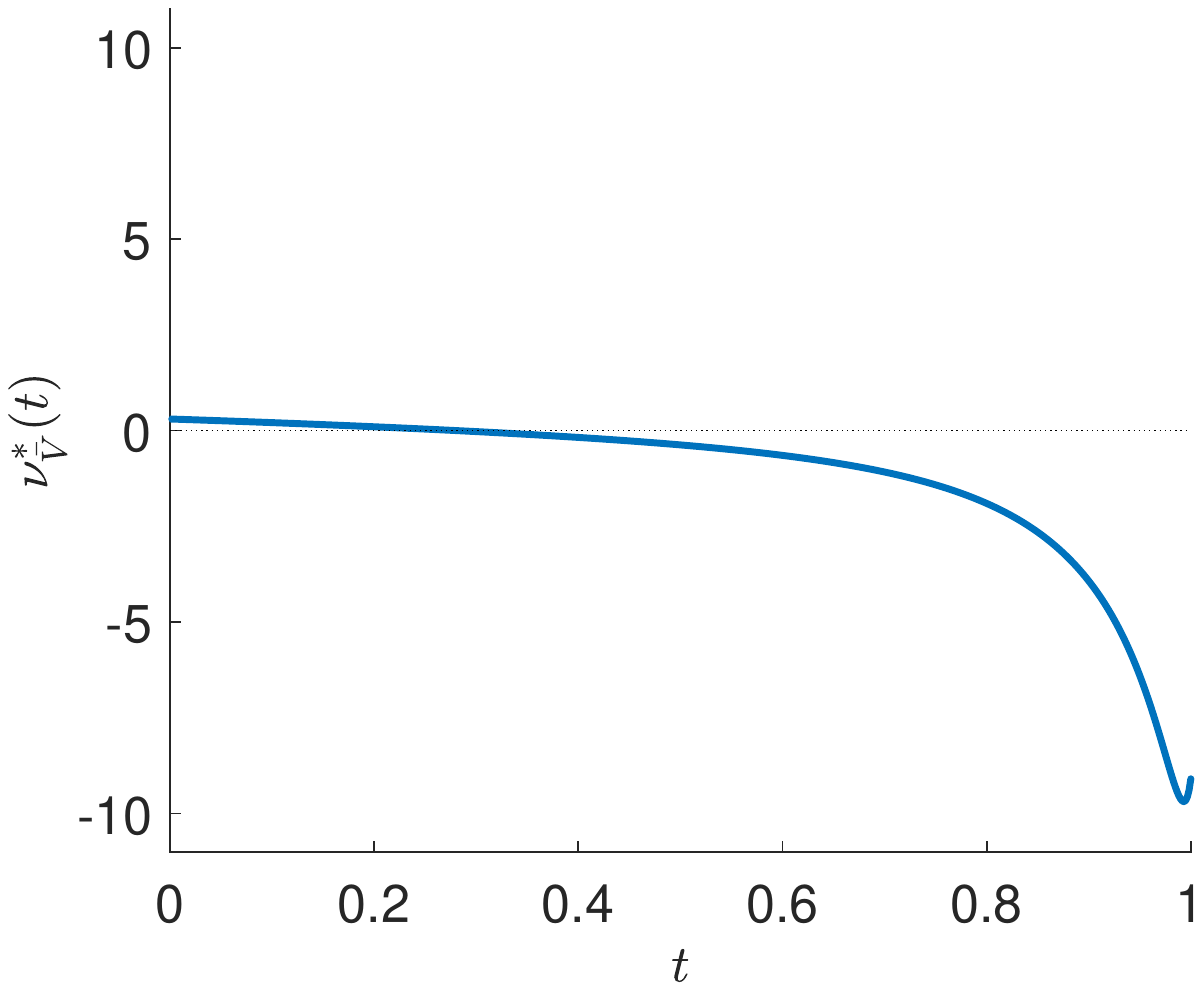}}
	\end{center}
	\vspace{-1em}
	\caption{Optimal loadings on $Q_t$, $\bar{Q}_t$, $V_t$, and $\bar{V}_t$. Parameters used are $\mu = 0$, $\sigma = 1$, $\eta = 0.5$, $\beta = 1$, $\gamma = 0.05$, $\bar{\gamma} = 0.1$, $\rho = 0.3$, $b = 10^{-2}$, $k = 5\cdot10^{-3}$, $\bar{k} = 10^{-3}$, $\alpha = 0.1$, and $T = 1$. \label{fig:trade-loadings-separate}}
\end{figure}

In Figure \ref{fig:nu_separate_qbarVbar_gammab} we show the loadings on $\bar{Q}_t$ and $\bar{V}_t$ as the three impact parameters $b$, $\gamma$, and $\bar{\gamma}$ are varied. Many of the features shown in this figure can be explained with similar reasoning to the discussion around Figure \ref{fig:nu_shared_qbar_gammab}. New features which deserve discussion are the qualitative shape of the loading $\nu^*_{\bar{V}}$ and the ordering of the curves based on the changing parameter.

Typically the loading $\nu_{\bar{V}}^*$ has a minimum value, usually negative, shortly before the end of the trading period. If the average signal viewed by agents is positive shortly before time $T$, then this will tend to increase the average order flow and the agent can expect their own trade signal to decrease, thus giving them reason to sell the asset. However, there is a counteracting effect which is the impact that the average order flow has on the asset price. When the average order flow is positive the asset price will tend to increase, giving incentive for the agent to buy shares shortly before time $T$. This explains why larger values of permanent price impact, $b$, result in higher loading $\nu_{\bar{V}}^*$ (bottom left panel) and why larger values of market impact on trade signal, $\bar{\gamma}$, result in lower loading $\nu_{\bar{V}}^*$ (bottom right panel). 

The permanence of price impact and the transience of trade signal impact also explain the sharp humps seen in this figure. Since any impact on the trade signal will decay over time due to mean reversion, the considerations of market wide order flow on trade signals become more significant shortly before $T$. The effects of market wide order flow on the price are long lasting, so the agent takes into account this effect over the entire trading period.


\begin{figure}
	\begin{center}
		{\includegraphics[trim=140 240 140 240, scale=0.4]{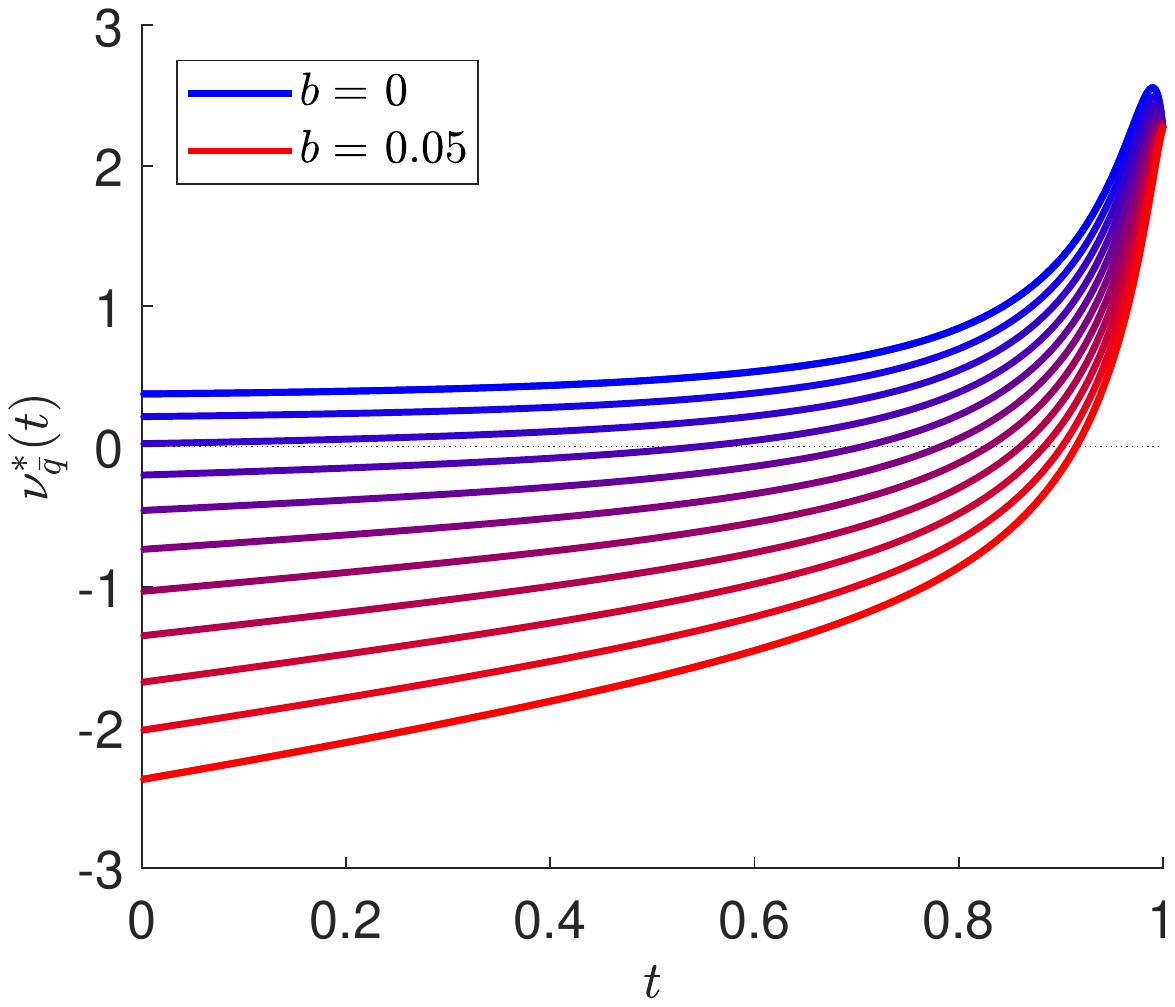}}\hspace{10mm}
		{\includegraphics[trim=140 240 140 240, scale=0.4]{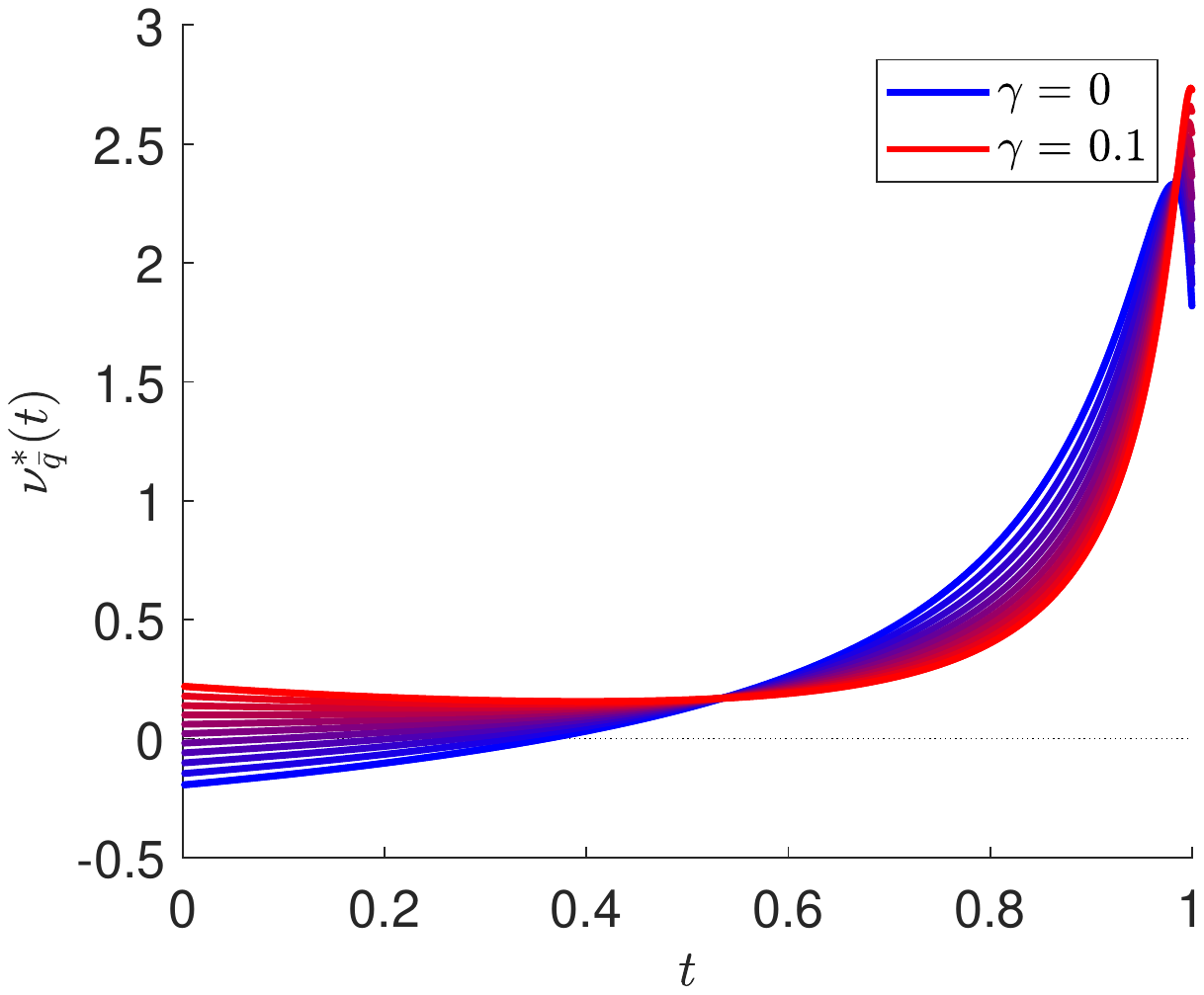}}\hspace{10mm}
		{\includegraphics[trim=140 240 140 240, scale=0.4]{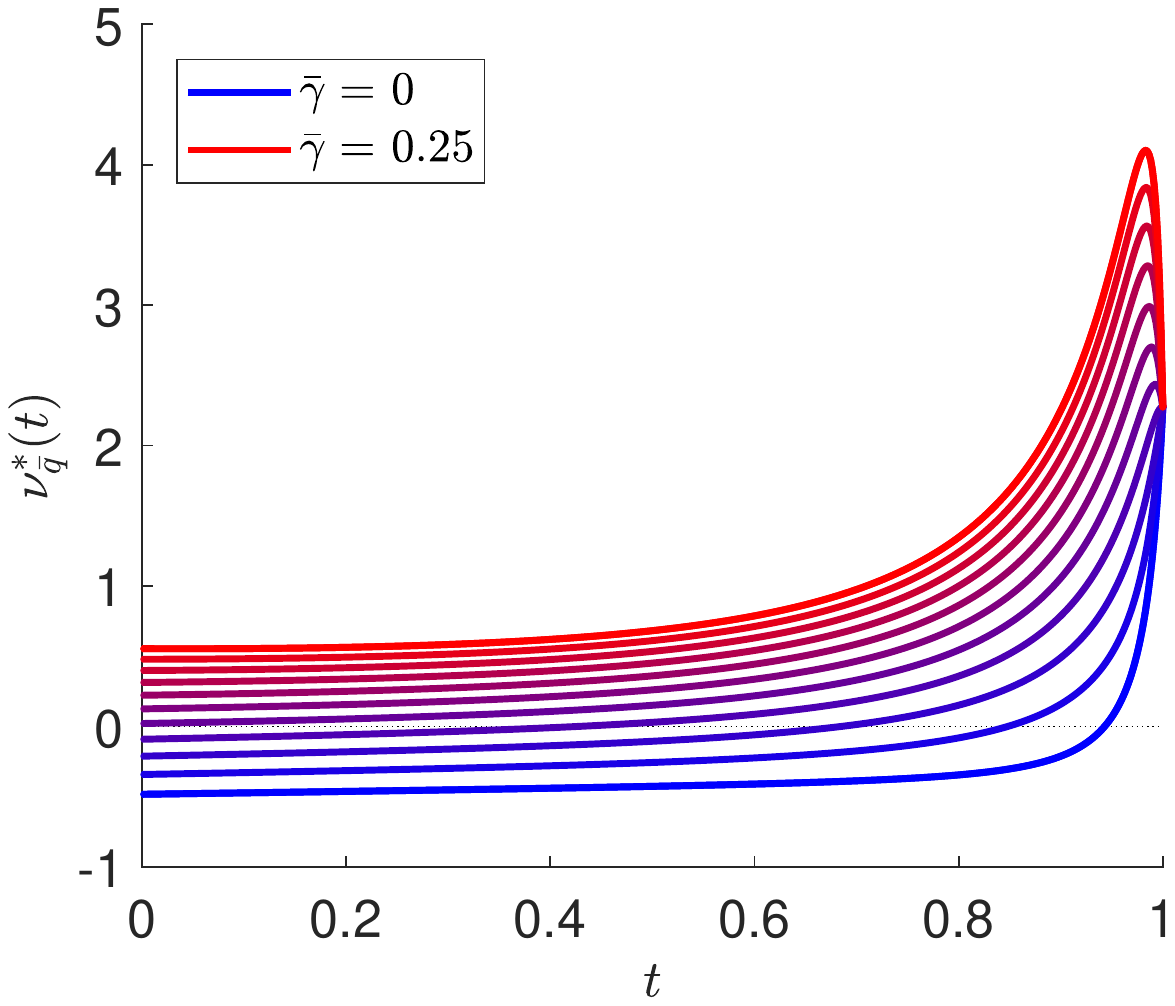}}\hspace{10mm}
		{\includegraphics[trim=140 240 140 240, scale=0.4]{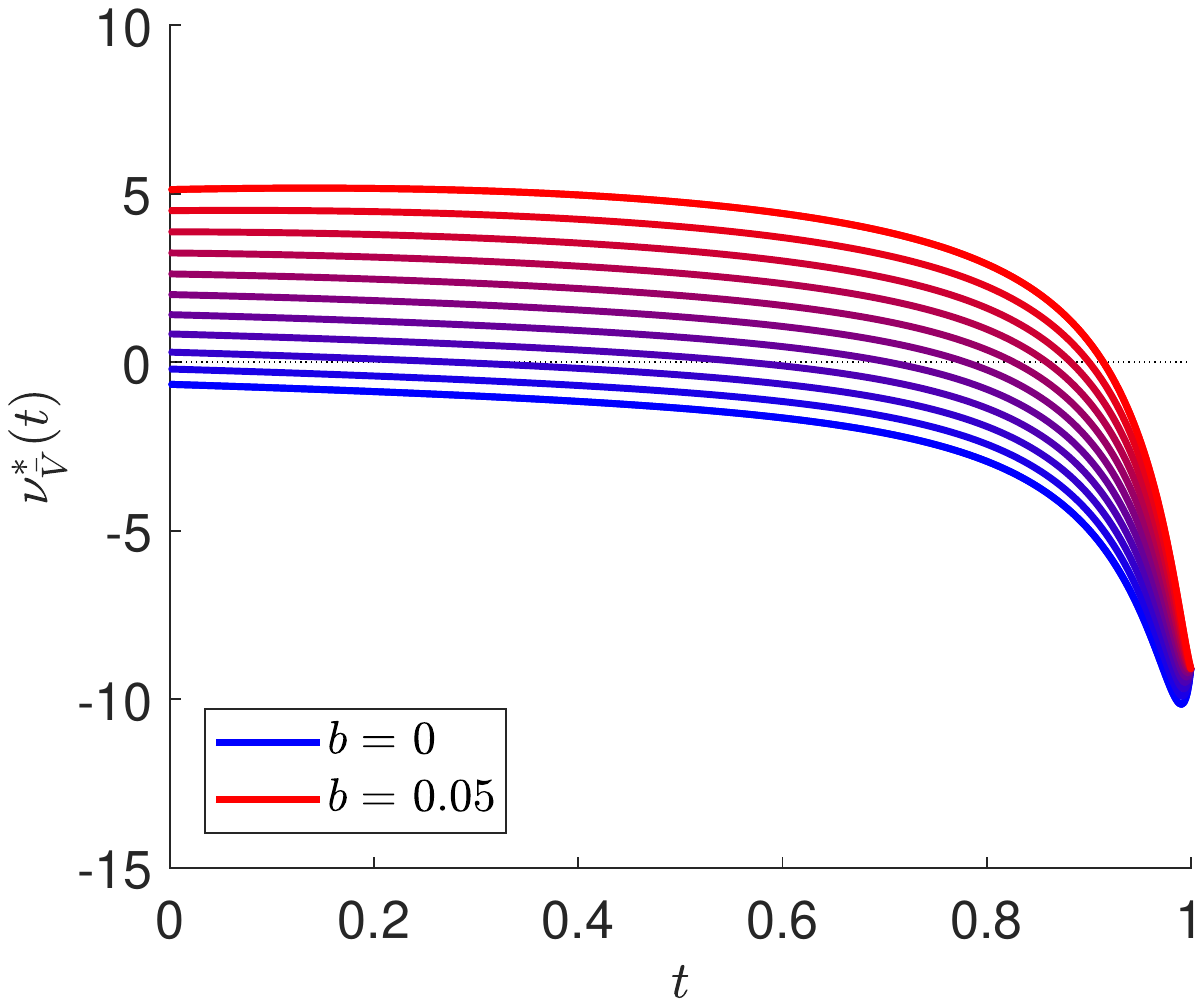}}\hspace{10mm}
		{\includegraphics[trim=140 240 140 240, scale=0.4]{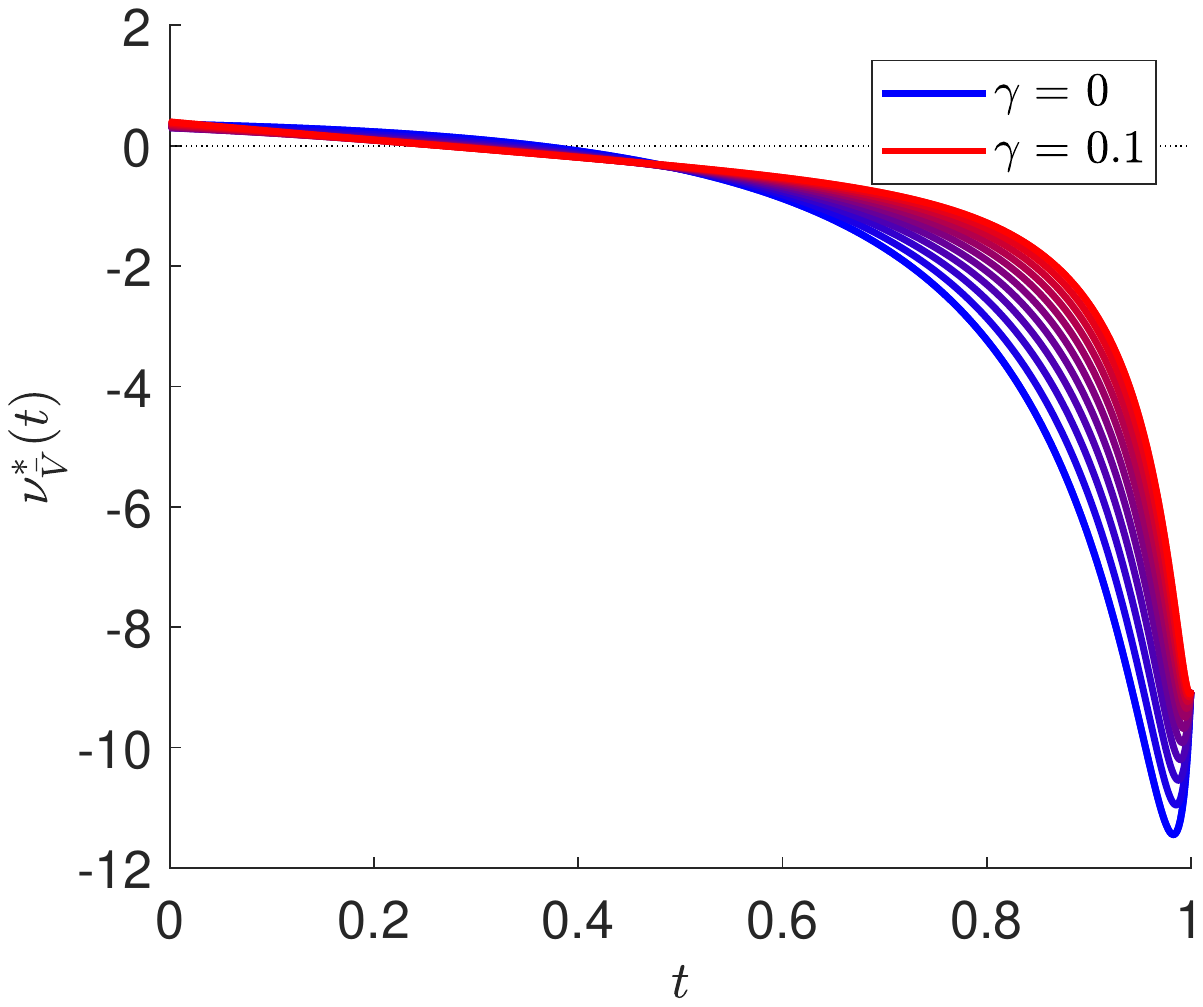}}\hspace{10mm}
		{\includegraphics[trim=140 240 140 240, scale=0.4]{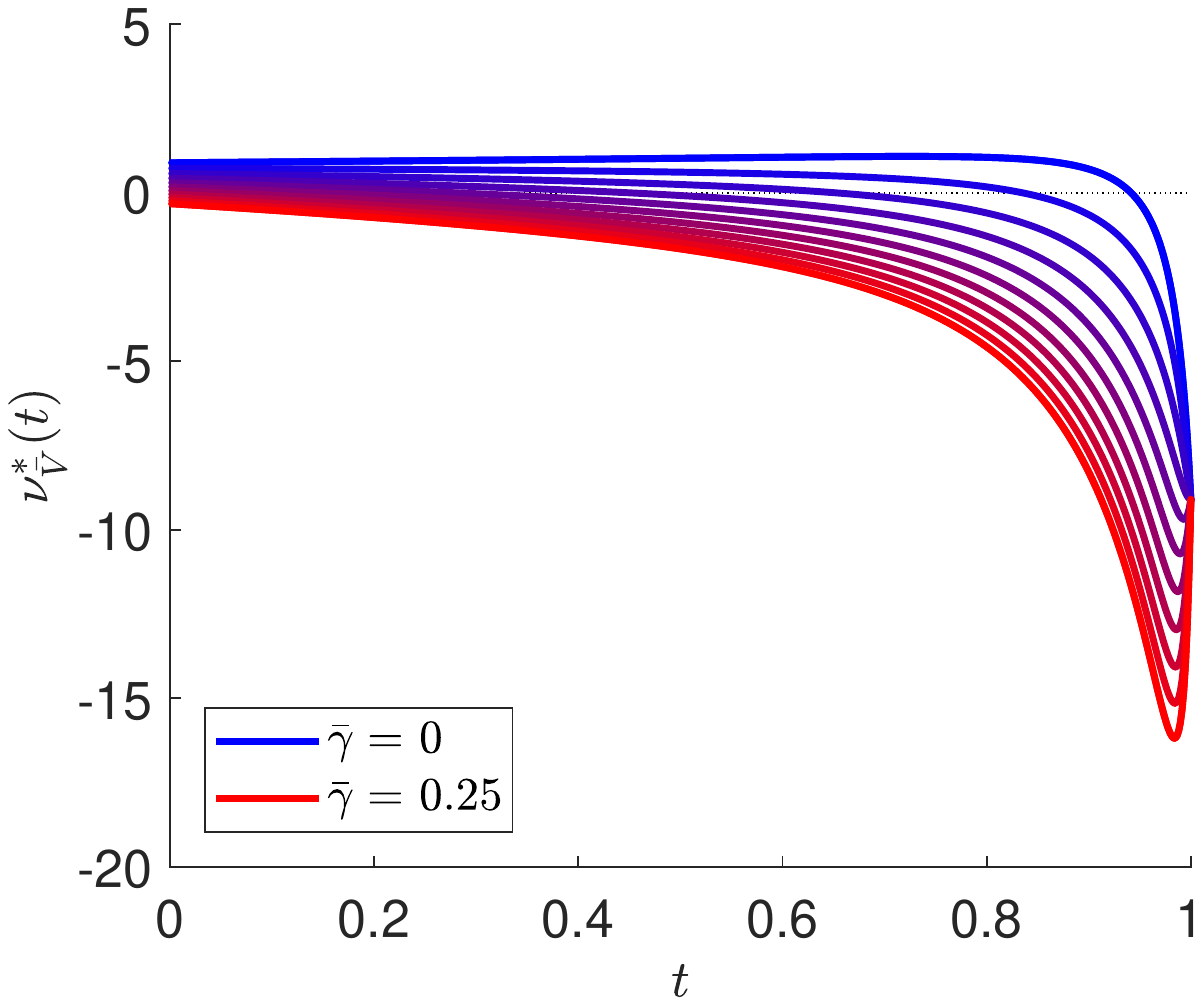}}\hspace{10mm}
	\end{center}
	\vspace{-1em}
	\caption{Top row shows optimal loading on $\bar{Q}_t$ for various parameters, bottom row shows optimal loading on $\bar{V}_t$. Each figure considers a change in only one parameter, indicated in the legend, from a minimum value (blue curve) to a maximum value (red curve). Otherwise the fixed parameters are $\mu = 0$, $\sigma = 1$, $\eta = 0.5$, $\beta = 1$, $\gamma = 0.05$, $\bar{\gamma} = 0.1$, $\rho = 0.3$, $b = 10^{-2}$, $k = 5\cdot10^{-3}$, $\bar{k} = 10^{-3}$, $\alpha = 0.1$, and $T = 1$.  \label{fig:nu_separate_qbarVbar_gammab}}
\end{figure}

In Figure \ref{fig:MC-separate} we show a simulation of relevant processes when each agent adopts the mean-field optimal strategy depicted in Figure \ref{fig:trade-loadings-separate}. The main qualitative difference between this simulation and that shown in Figure \ref{fig:MC-shared} is that the distribution of terminal inventories $(Q_T^n)_{n\leq N}$ does not become concentrated around a particular value based on the average trade signal. In fact, in this particular simulation the terminal inventories have sample variance $1.32$ which is significantly greater than the initial sample variance of $0.24$ (the initial inventories are drawn from a distribution with variance $0.5^2=0.25$).

\begin{figure}
	\begin{center}
		{\includegraphics[trim=140 240 140 240, scale=0.4]{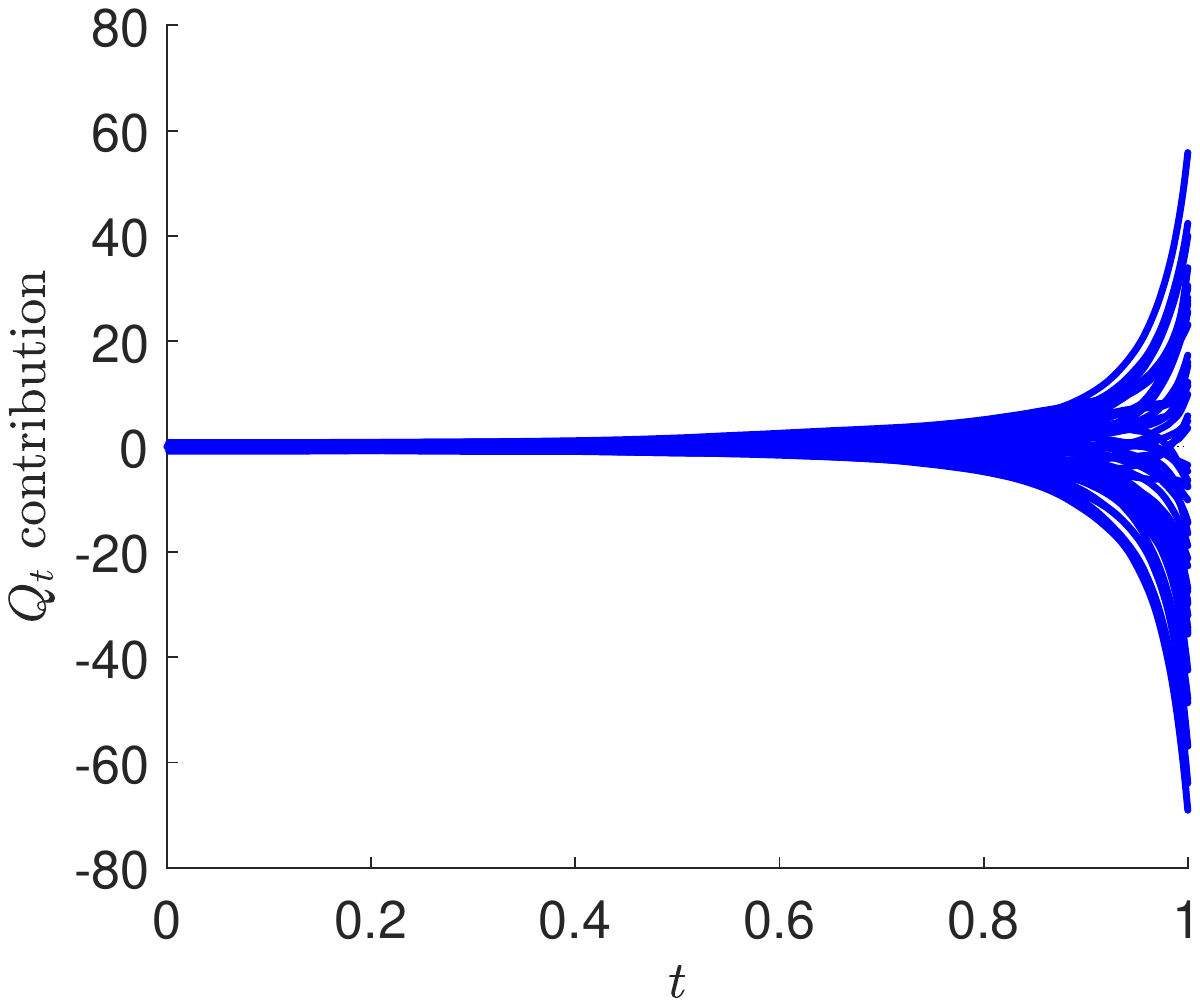}}\hspace{10mm}
		{\includegraphics[trim=140 240 140 240, scale=0.4]{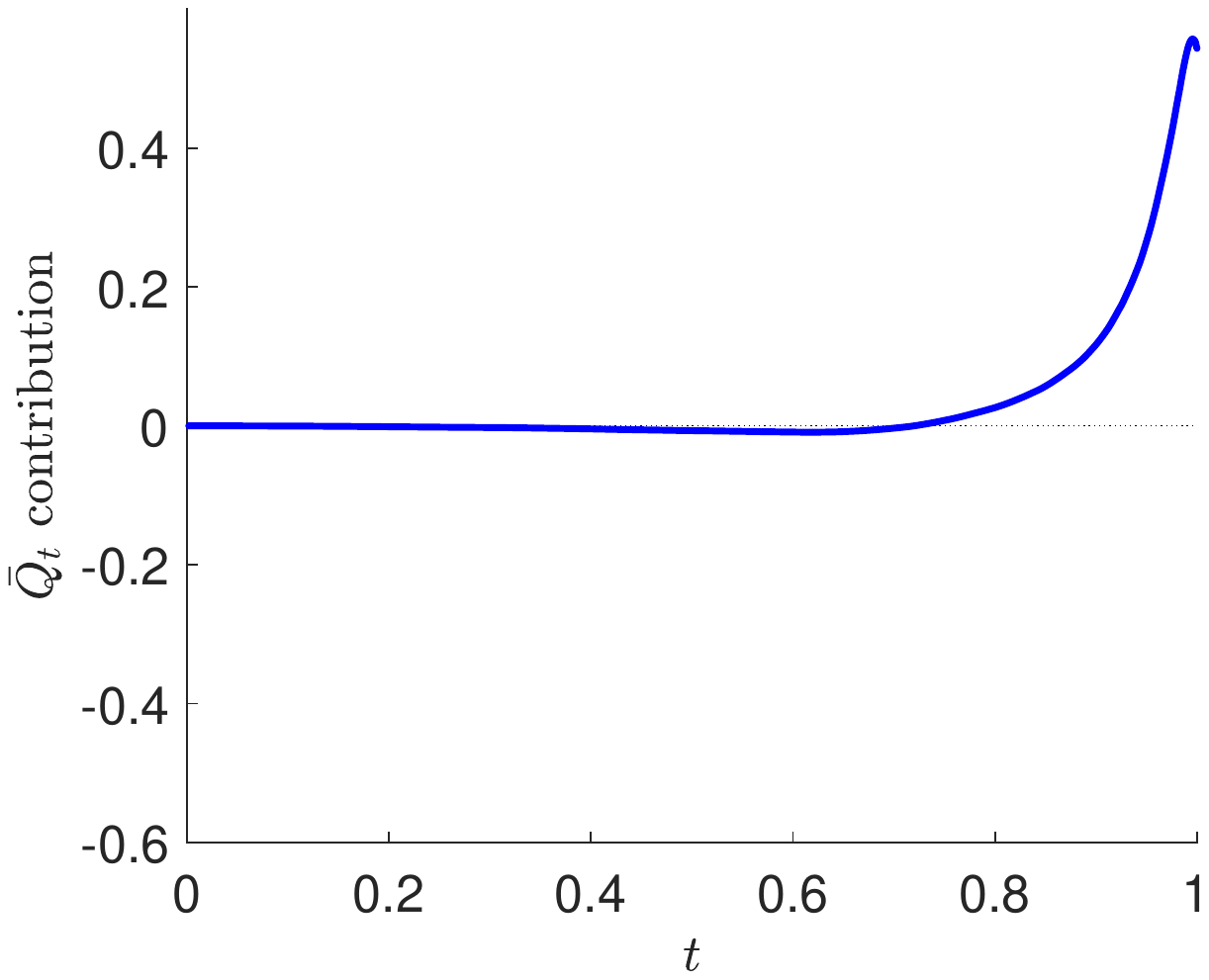}}\hspace{10mm}
		{\includegraphics[trim=140 240 140 240, scale=0.4]{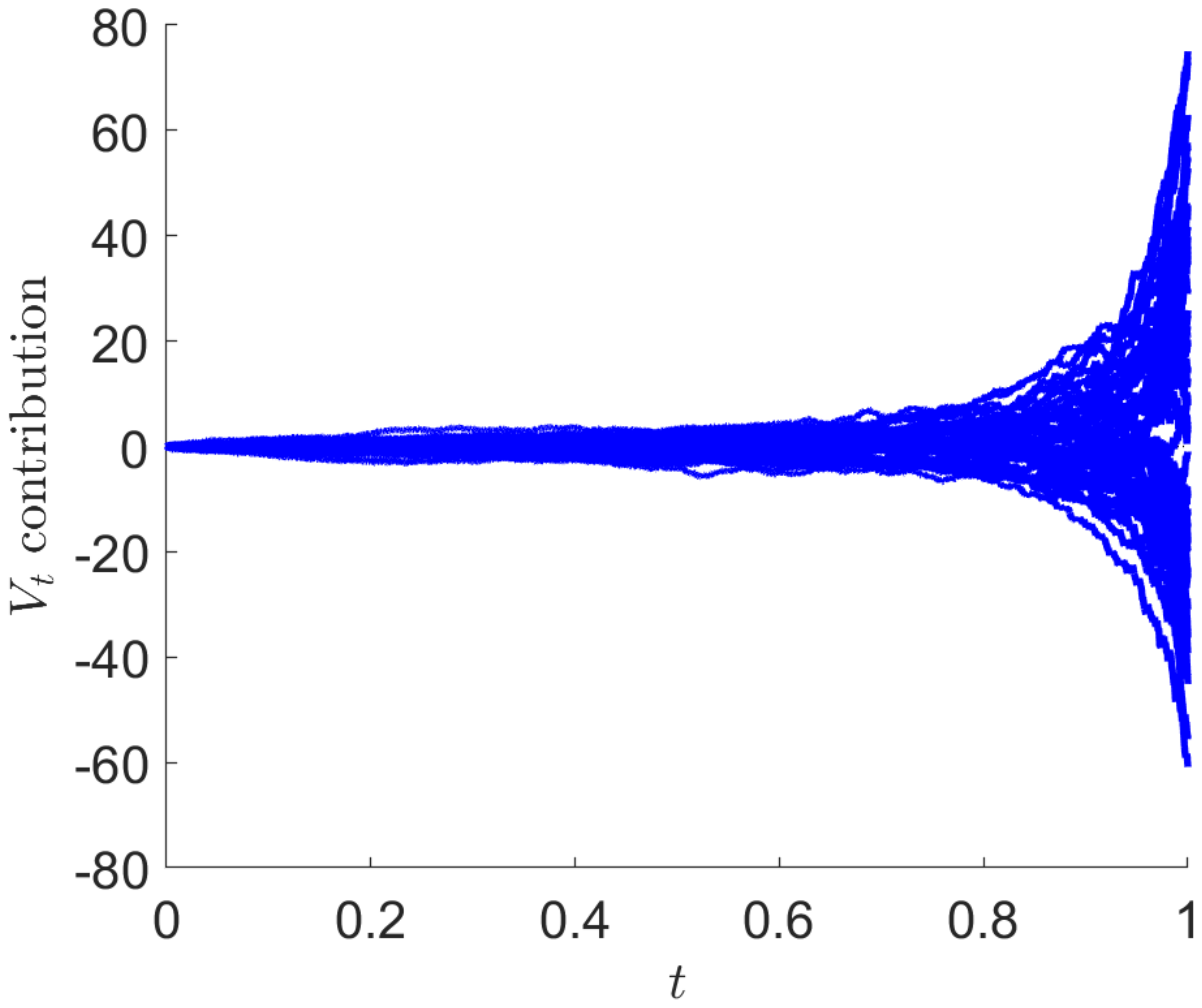}}\hspace{10mm}
		{\includegraphics[trim=140 240 140 240, scale=0.4]{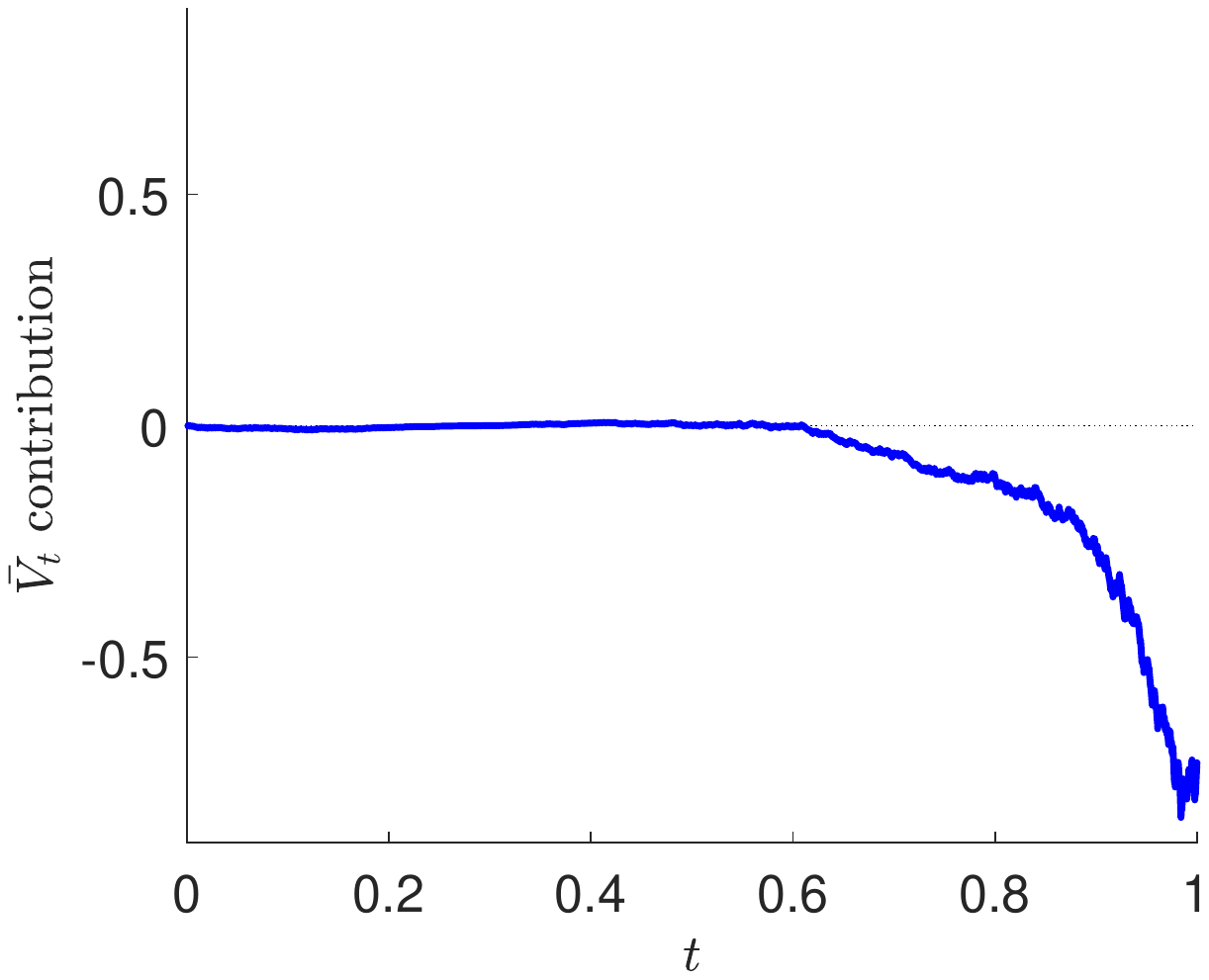}}\hspace{10mm}
		{\includegraphics[trim=140 240 140 240, scale=0.4]{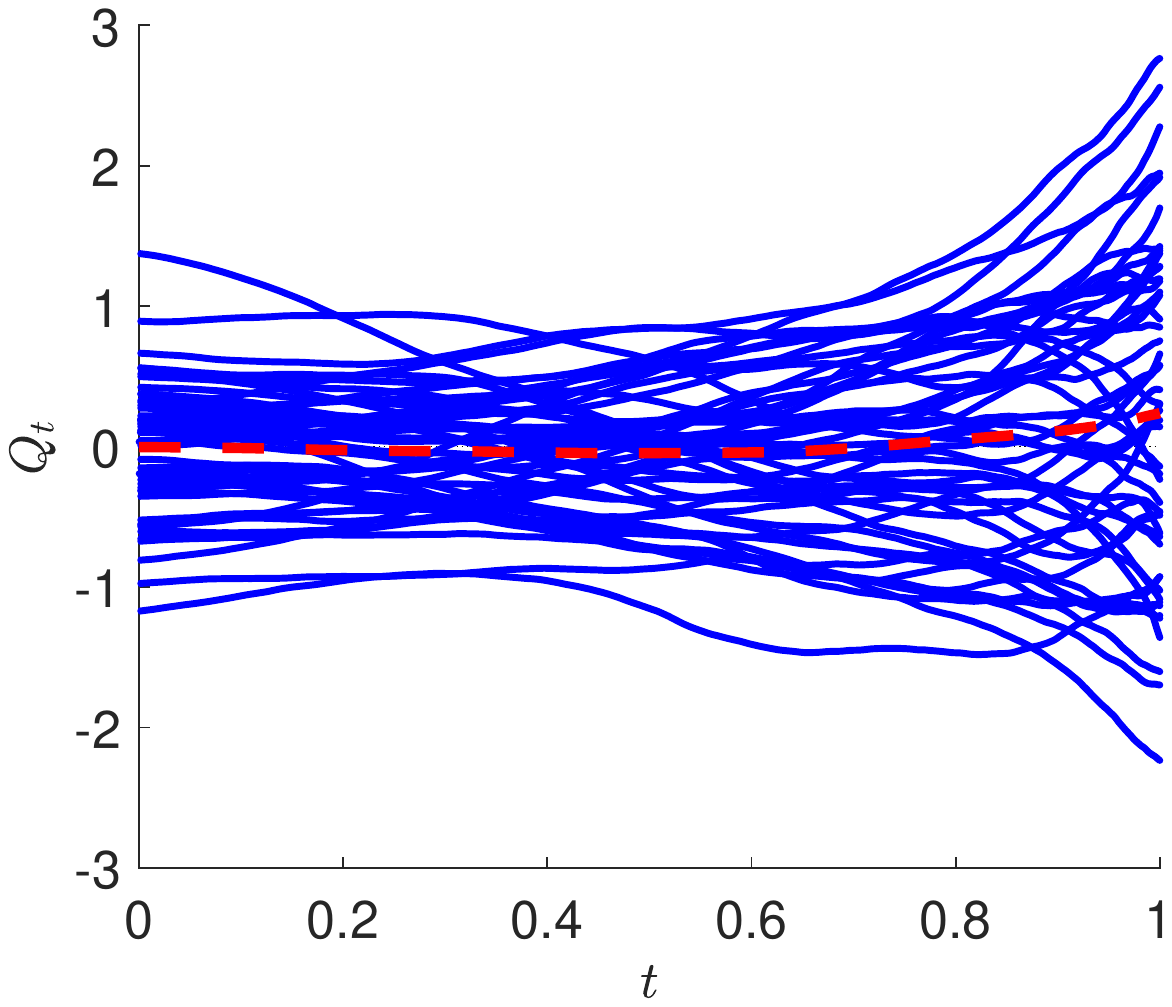}}\hspace{10mm}
		{\includegraphics[trim=140 240 140 240, scale=0.4]{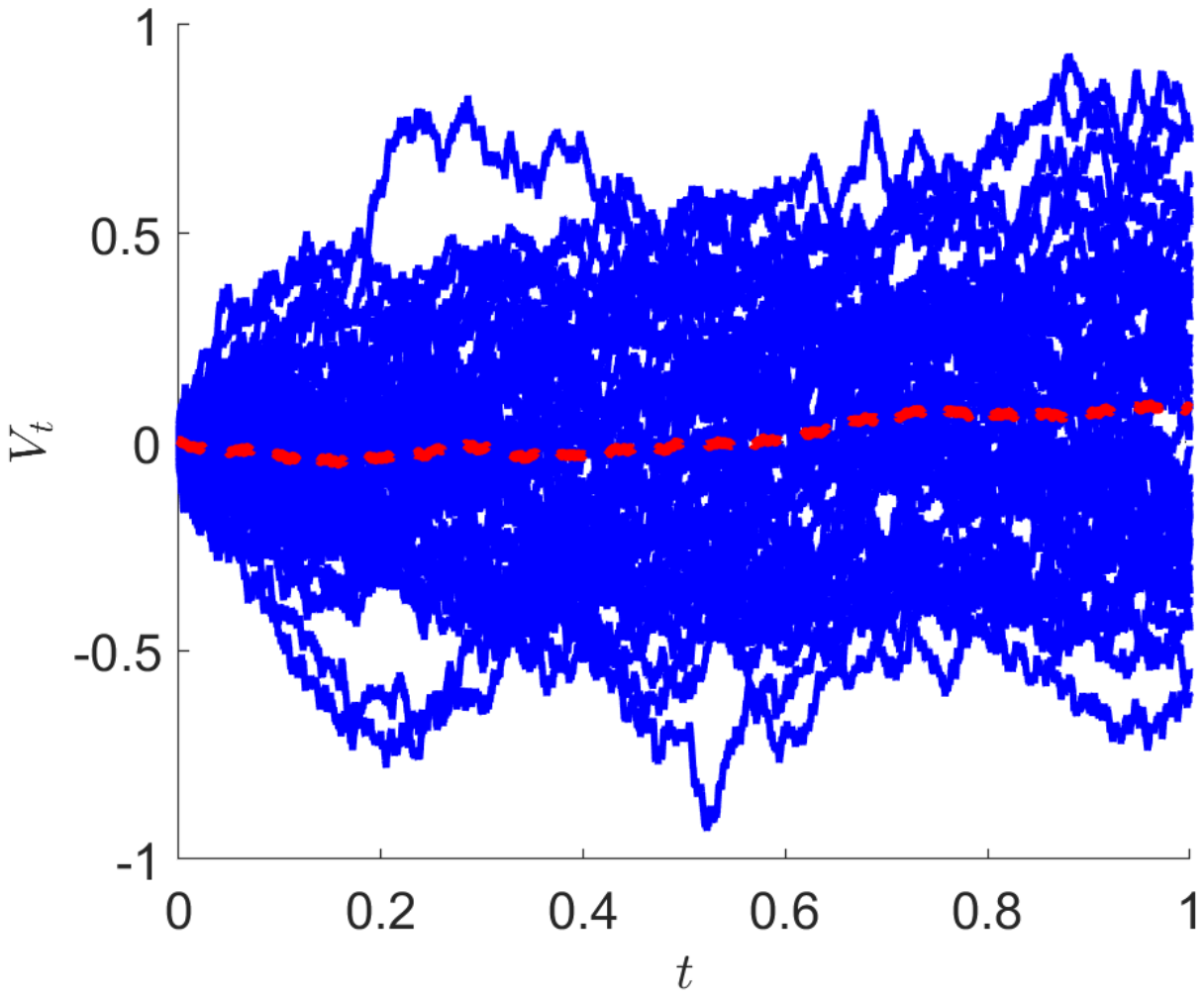}}\hspace{10mm}
		{\includegraphics[trim=140 240 140 240, scale=0.4]{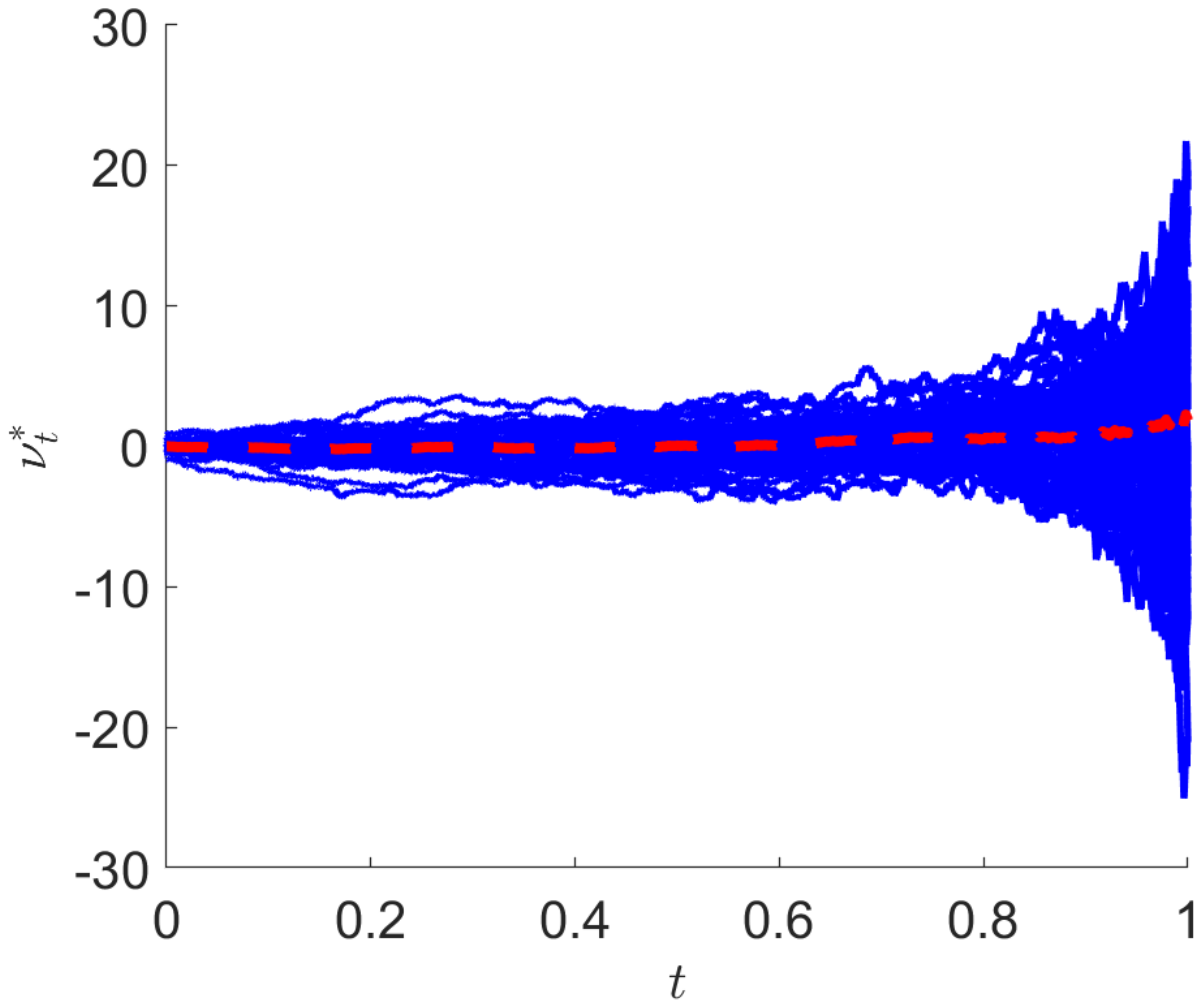}}
	\end{center}
	\vspace{-1em}
	\caption{The top row shows the contributions to trading speed from $Q_t$, $\bar{Q}_t$, and $V_t$. The left panel of the second row shows the contributions from $\bar{V}_t$. The middle panel of the second row shows each agent's inventory path $Q_t^{n,\nu^n}$ (blue curves) as well as the average inventory of all agents $\bar{Q}^{\bar{\nu}}_t$ (red dotted curve). The right panel of the second row shows each agent's trade signal path $V_t^{n,\nu^n,\bar{\nu}}$ (blue curves) as well as the average trade signal of all agents $\bar{V}^{\bar{\nu}}_t$ (red dotted curve). The bottom panel shows each agent's optimal trading speed $\nu^n_t$ (blue curves) and the average trading speed $\bar{\nu}_t$ (red dotted curve). Parameters used are $\mu = 0$, $\sigma = 1$, $\eta = 0.5$, $\beta = 1$, $\gamma = 0.05$, $\bar{\gamma}=0.1$, $\rho = 0.3$, $b = 10^{-2}$, $k = 5\cdot10^{-3}$, $\bar{k} = 10^{-3}$, $\alpha = 0.1$, $T = 1$, $S_0 = 100$, $V^n_0 \sim\mathcal{N}(0,0.02^2)$, $Q_0^n\sim\mathcal{N}(0,0.5^2)$, and $N = 50$. \label{fig:MC-separate}}
\end{figure}

\section{Cross-Sectional Distribution of Inventories and Signals}\label{sec:distribution}

In this section we compute the joint distribution of the agents' inventories and signals when all agents use the mean-field equilibrium strategy given by \eqref{eq:nustar-separate} (with \eqref{eqn:separate-f} enforced). We do not directly consider the case when all agents observe the same trade signal because those corresponding results can be obtained from those of the separate signal by setting $\rho=1$, $\gamma=0$, and each $V_0^n$ the same constant. In addition, as we are assuming all agents are using the mean-field equilibrium strategies, we drop the notational dependencies on $\nu^n$ and $\bar{\nu}$.

We begin by defining the processes $Y^n = (Y^n_t)_{0 \leq t \leq T}$ and $\bar{Y} = (\bar{Y}_t)_{0 \leq t \leq T}$ by
\begin{align*}
	Y_t^n &= \left[\begin{array}{c} Q_t^n \\ V_t^n\end{array}\right]\,, & \bar{Y}_t &= \left[\begin{array}{c} \bar{Q}_t \\ \bar{V}_t\end{array}\right]\,.
\end{align*}
We also introduce random measure processes on $\mathds{R}^2$, denoted $m^N = (m^N_t)_{0\leq t \leq T}$ and $m = (m_t)_{0\leq t \leq T}$, which are given by
\begin{align*}
	m^N_t &= \frac{1}{N}\sum_{n=1}^N \delta_{Y^n_t}\,, &
	m_t &= \lim_{N\rightarrow\infty} m^N_t\,.
\end{align*}
In the next proposition we provide expressions for the mean vector $\bar{Y}_t$ and covariance matrix $\bar{\Sigma}_t$ of the distribution induced by $m$. 

\begin{proposition}\label{prop:pop_distribution}
	Let $a_t$, $B_t$, and $C_t$ be given by
	\begin{align*}
		a_t &= \frac{c_2(t)-\gamma c_4(t)}{2k+\bar{k}}\left[\begin{array}{c} 1 \\ -(\gamma+\bar{\gamma})\end{array}\right]\,,\\
		B_t &= \left[\begin{array}{cc} \nu_q^*(t) & \nu_V^*(t) \\ -\gamma\nu_q^*(t) & -(\beta+\gamma\nu_V^*(t)) \end{array}\right]\,,\\
		C_t &= \left[\begin{array}{cc} \nu_q^*(t) + \nu_{\bar{q}}^*(t) & \nu_V^*(t)+\nu_{\bar{V}}^*(t) \\ -(\gamma+\bar{\gamma})(\nu_q^*(t) + \nu_{\bar{q}}^*(t)) & -\biggl(\beta + (\gamma+\bar{\gamma})(\nu_V^*(t) + \nu_{\bar{V}}^*(t))\biggr) \end{array}\right]\,,
	\end{align*}
	and let $\Phi_t$ and $\Psi_t$ be the solutions to the matrix differential equations
	\begin{align}
		\Phi_t' &= C_t\,\Phi_t\,, & \Phi_0 &= I^{2\times2}\,,\label{eqn:dPhi}\\
		\Psi_t' &= B_t\,\Psi_t\,, & \Psi_0 &= I^{2\times2}\,.\label{eqn:dPsi}
	\end{align}
	The mean vector and covariance matrix induced by $m_t$ are given by
	\begin{align}
		\bar{Y}_t &= \Phi_t\,\biggl(\bar{Y}_0 + \int_0^t \Phi_u^{-1}\,a_u\,\dd u+ \rho\,\int_0^t \Phi_u^{-1}\,\Theta\, \dd W_u\biggr)\,,\label{eqn:population_mean}\\
		\bar{\Sigma}_t &= \Psi_t\bar{\Sigma}_0\Psi_t^\top + (1-\rho^2)\Psi_t\int_0^t \Psi^{-1}_u\Theta\Theta^\top(\Psi_u^{-1})^\top \dd u\Psi_t^\top\,,\label{eqn:population_cov}
	\end{align}
	where
	\begin{align*}
		\Theta &= \left[\begin{array}{c} 0 \\ \eta\end{array}\right]\,.
	\end{align*}
	If the distribution induced by $m_0$ is Gaussian, then $m_t$ induces a Gaussian distribution for all $t\in[0,T]$. If $\mu=\alpha=\gamma=0$ then the covariance matrix in \eqref{eqn:population_cov} has individual elements
	\begin{align}
		\bar{\Sigma}_t^Q &= \bar{\Sigma}_0^Q + \frac{\ee^{-\beta T}}{k}\,t\,\bar{\Sigma}_0^{QV} + \frac{\ee^{-2\beta T}}{4\,k^2}\,t^2\,\bar{\Sigma}_0^V + (1-\rho^2)\,\frac{\eta^2\,\ee^{-2\beta T}}{16\,\beta^3\,k^2}\,(\ee^{2\beta t} - 1 - 2\,\beta\, t - 2\,\beta^2\, t^2)\,,\label{eqn:Sigma_q}\\
		\bar{\Sigma}_t^V &= 2^{-2\beta t}\,\bar{\Sigma}_0^V + (1-\rho^2)\,\frac{\eta^2}{2\,\beta}\,(1-\ee^{-2\beta t})\,,\label{eqn:Sigma_v}\\
		\bar{\Sigma}_t^{QV} &= \ee^{-\beta t}\,\bar{\Sigma}_0^{QV} + \frac{\ee^{-\beta T}}{2\,k}\,t\,\ee^{-\beta t}\,\bar{\Sigma}_0^V + (1-\rho^2)\,\frac{\eta^2\,\ee^{-\beta T}}{4\,\beta^2\,k}(\sinh(\beta t) - \beta\,t\,\ee^{-\beta t})\,.\label{eqn:Sigma_qv}
	\end{align}
\end{proposition}

\begin{proof}
	The dynamics of $Y^n$ and $\bar{Y}$ are given by
	\begin{align}
		\dd Y_t^n &= (a_t + B_t\,Y_t^n + (C_t-B_t)\,\bar{Y}_t)\, \dd t + \Theta \, \dd Z_t^n\,,\label{eqn:dY}\\
		\dd \bar{Y}_t &= (a_t + C_t\,\bar{Y}_t)\, \dd t + \rho\,\Theta \, \dd W_t\,.\label{eqn:dbarY}
	\end{align}
	The solution to \eqref{eqn:dbarY} is given by \eqref{eqn:population_mean} (see Section 5.6 of \cite{karatzas2012brownian}). By substituting this solution for $\bar{Y}_t$ into \eqref{eqn:dY} and performing some tedious computations we arrive at
	\begin{align}
		Y_t^n &= \Psi_t\,(Y_0^n-\bar{Y}_0) + \Phi_t\,\bar{Y}_0 + \Phi_t\int_0^t\Phi_u^{-1}\,a_u\,\dd u + \rho\,\Phi_t\int_0^t\Phi_u^{-1}\,\Theta\,\dd W_u + \sqrt{1-\rho^2}\,\Psi_t\int_0^t\Psi_u^{-1}\,\Theta\,\dd W^{n,\perp}_u\,.
	\end{align}
	Subtracting $\bar{Y}_t$ from this expression yields
	\begin{align}
		Y_t^n-\bar{Y}_t &= \Psi_t\,(Y_0^n-\bar{Y}_0) + \sqrt{1-\rho^2}\Psi_t\int_0^t\Psi_u^{-1}\Theta\, \dd W_u^{n,\perp}\,,
	\end{align}
	from which we also compute
	\begin{align*}
		(Y_t^n-\bar{Y}_t)(Y_t^n-\bar{Y}_t)^\top &= \Psi_t\,(Y_0^n-\bar{Y}_0)\,(Y_0^n-\bar{Y}_0)^\top\,\Psi_t^\top + 2\sqrt{1-\rho^2}\,\Psi_t\,(Y_0^n-\bar{Y}_0)\int_0^t\Psi_t^{-1}\Theta\,\dd W_u^{n,\perp}\\
			&\hspace{15mm} + (1-\rho^2)\Psi_t\biggl(\int_0^t\Psi_t^{-1}\Theta\,\dd W_u^{n,\perp}\biggr)\biggl(\int_0^t\Psi_t^{-1}\Theta\,\dd W_u^{n,\perp}\biggr)^\top\Psi_t^\top\,.
	\end{align*}
	We sum both sides over $1\leq n \leq N$ and divide by $N$. As $N\rightarrow\infty$ the left hand side converges to $\bar{\Sigma}_t$. The second term on the right converges to zero due to independence of $Y_0^n$ and $W^{n,\perp}$. Applying the law of large numbers and Ito's isometry to the third term yields \eqref{eqn:population_cov}. If the initial distribution of $m_0$ is Gaussian, then independence of $Y_0^n$ and $W^{n,\perp}$ and the fact that the stochastic integrand is deterministic result in $m_t$ being Gaussian.
	
	To obtain the expression in \eqref{eqn:Sigma_q}, \eqref{eqn:Sigma_v}, and \eqref{eqn:Sigma_qv} we first use Proposition \ref{prop:separate_closed} to write the matrix $B_t$ in closed form. Then \eqref{eqn:dPsi} can be solved in closed form, which yields
	\begin{align*}
		\Psi_t &= \left[\begin{array}{cc} 1 & \frac{\ee^{-\beta T}}{2\,k}\,t \\ 0 & \ee^{-\beta t}\end{array}\right]\,, & \Psi_t^{-1} &= \left[\begin{array}{cc} 1 & -\frac{\ee^{-\beta (T-t)}}{2\,k}\,t \\ 0 & \ee^{\beta t}\end{array}\right]\,.
	\end{align*}
	Substituting these expressions into \eqref{eqn:population_cov} and computing the integral gives the result.
\end{proof}

The covariance matrix in \eqref{eqn:population_cov} confirms an observation made in comparing the simulations of Figure \ref{fig:MC-shared} and Figure \ref{fig:MC-separate}: the sample variance of the terminal inventory of all agents is greater when they have separate signals compared to when they share the same signal. This is because of the lower correlation between signals implied by the separate signals and the term $1-\rho^2$ in \eqref{eqn:population_cov}. In fact the variance of inventory will be minimized when the correlation is $\rho=\pm1$. This has a clear intuitive reason being that if the agents have very similar signals then they will trade in a similar fashion, and any variance in their terminal inventory will be the result of variance of their initial inventory and the limited speed of trading due to market frictions such as temporary price impact.

In Figure \ref{fig:mf_moments} we show the variances and correlation across agents of inventories and signals in the mean-field limit. This gives a visual demonstration that the variance of inventories is lowest when $\rho^2$ is largest. In addition we also see that in the early parts of the trading period the variance does not depend much on the magnitude of shared information which is measured by $\rho^2$. This is due to the fact that for much of the trading interval the agents are controlling the size of their inventory by trading towards a common target of zero. When the end of the trading period is closer they begin to take advantage of the information in the trade signal, and their trading targets due to the trade signal may be different causing their inventories to diverge.

The behaviour of the trade signal variance is more expected. Since the initial distribution is relatively concentrated with a variance of $0.02^2$, the variance quickly increases, but at different rates depending on the magnitude of shared information measured by $\rho^2$. If $\rho^2$ is large then the agents share much of the same information, and so it is expected that the cross sectional variance of their trade signals is lower.

\begin{figure}
	\begin{center}
		{\includegraphics[trim=140 240 140 240, scale=0.4]{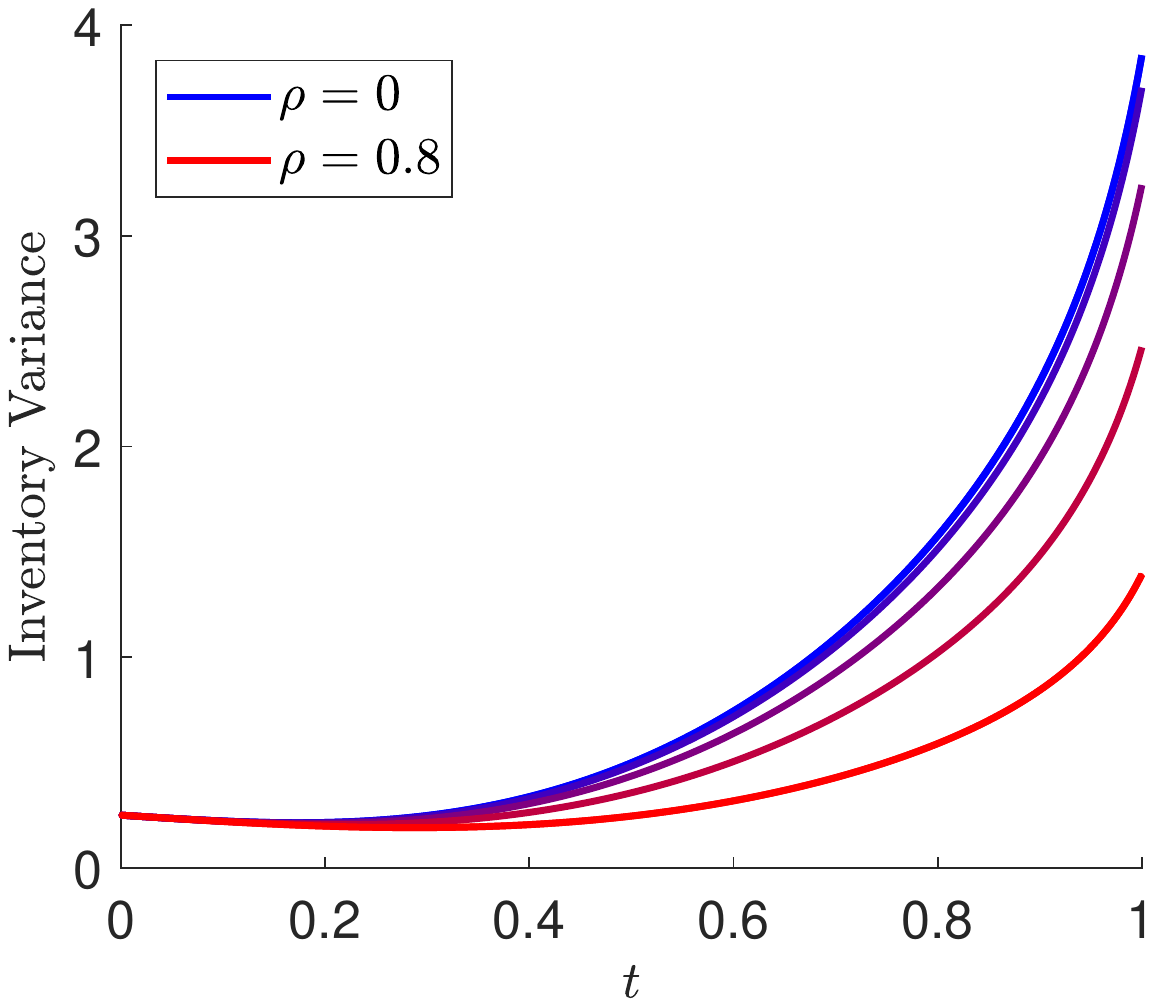}}\hspace{10mm}
		{\includegraphics[trim=140 240 140 240, scale=0.4]{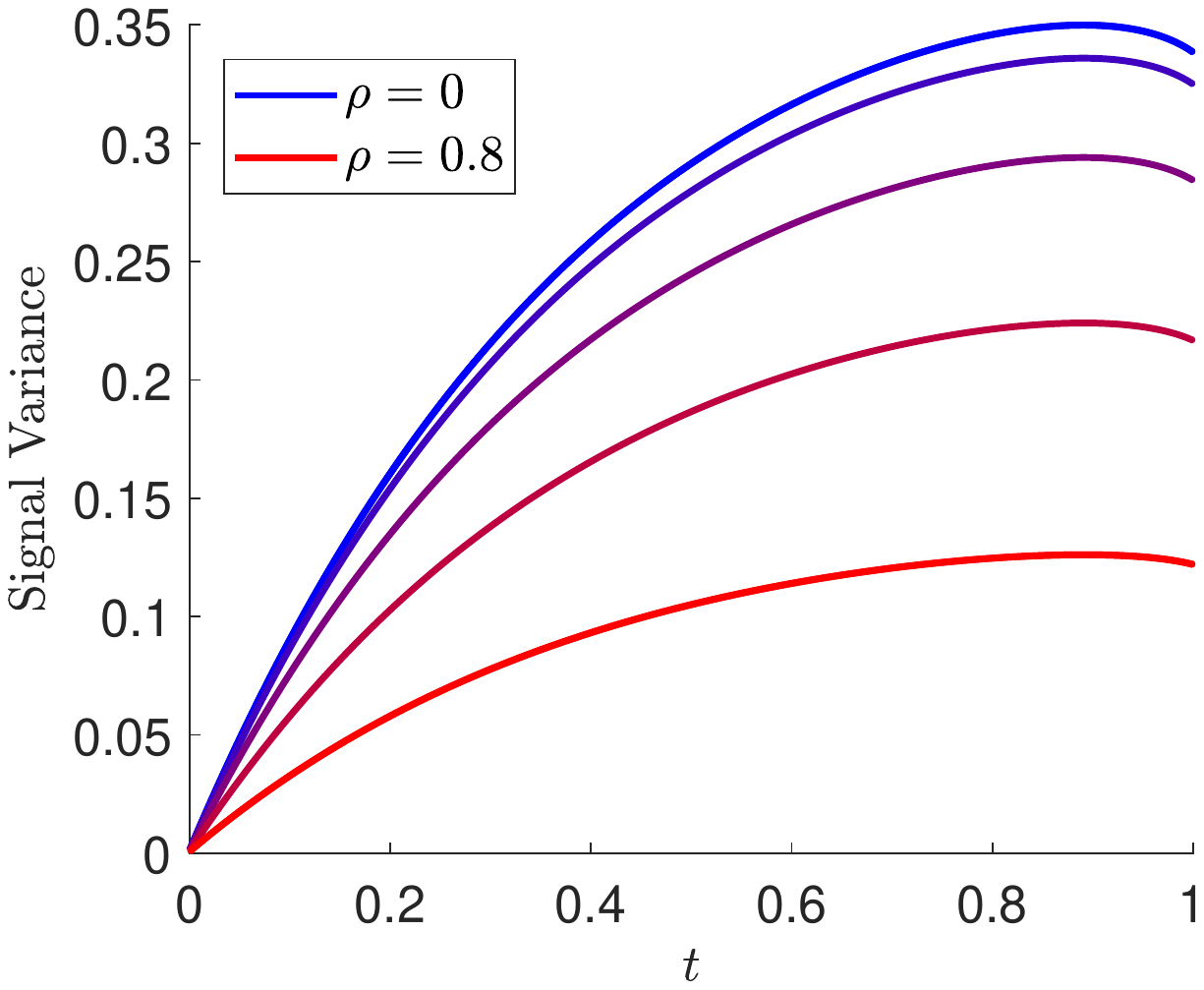}}\hspace{10mm}
		{\includegraphics[trim=140 240 140 240, scale=0.4]{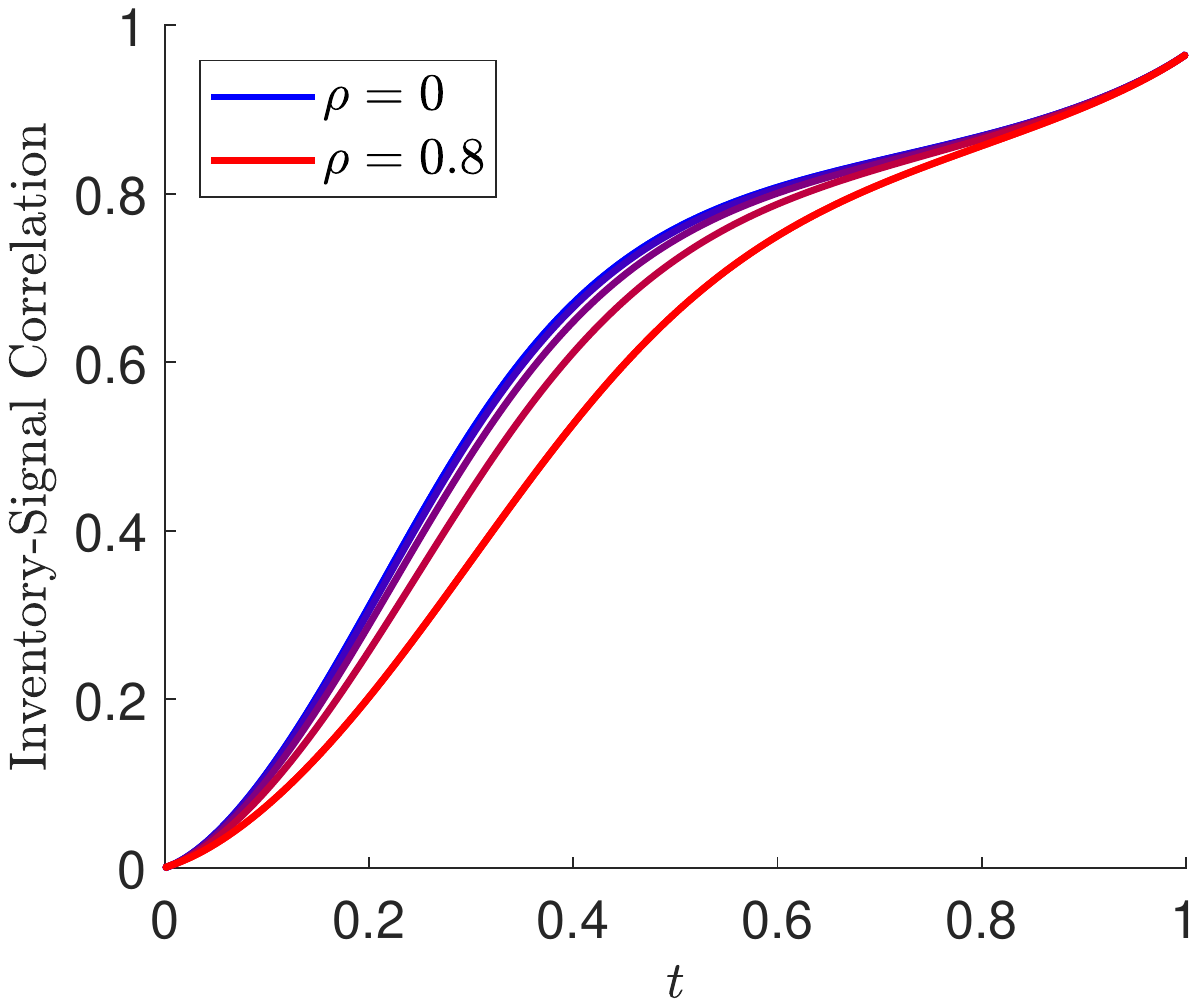}}
	\end{center}
	\vspace{-1em}
	\caption{Cross sectional variances and correlation of $Q_t^n$ and $V_t^n$ in mean-field limit. Parameters used are $\mu = 0$, $\sigma = 1$, $\eta = 1$, $\beta = 1$, $\gamma = 0.05$, $\bar{\gamma}=0.1$, $b = 5\cdot10^{-2}$, $k = 5\cdot10^{-3}$, $\bar{k} = 10^{-3}$, $\alpha = 0.1$, $T = 1$. The initial variances of inventories and signals are $0.5^2$ and $0.02^2$ respectively with an initial correlation of $0$.  \label{fig:mf_moments}}
\end{figure}

With the expression given in \eqref{eqn:population_mean} for the cross sectional mean of inventory and signal, we are able to demonstrate the effect of a shared trade signal on the variance of the asset price. This is done in the following proposition.

\begin{proposition} In mean-field equilibrium, the variance of the asset price is
	\begin{align}
		\mathds{E}[(S_t-\mathds{E}[S_t])^2] &= \int_0^t \biggl(\rho\,\eta\,b \left[\begin{array}{cc} 1 & 0 \end{array}\right]\,\Phi_t\,\Phi_u^{-1}\left[\begin{array}{c} 0 \\ 1 \end{array}\right] + \sigma\biggr)^2\,\dd u\,,\label{eqn:S_var}
	\end{align}
	where $\Phi_t$ is as in Proposition \ref{prop:pop_distribution}. If $\mu=\alpha=\gamma=0$ then this reduces to
	\begin{align}
		\mathds{E}[(S_t-\mathds{E}[S_t])^2] &= \int_0^t \biggl( \frac{2\,\rho\,\eta\, z\,(1-e^{-\frac{b}{\kappa}(t-u)})}{(1+2\,z)\,e^{\omega(T-u)}-1}    + \sigma\biggr)^2\,\dd u\,,\label{eqn:S_var_closed}
	\end{align}
	where
	\begin{align*}
		\omega &= \frac{\kappa\beta-b}{\kappa}\,, & z &= \frac{\kappa\beta - b}{2\bar{\gamma}}\,, & \kappa &= 2k+\bar{k}\,.
	\end{align*}
\end{proposition}

\begin{proof}
	With $\bar{\nu}_t$ in \eqref{eqn:S_shared} being set equal to the average trading speed in equilibrium we may write
	\begin{align*}
		\dd S_t &= \biggl(\mu + b\frac{c_2(t)-\gamma c_4(t)}{2k+\bar{k}} + N_t\,\bar{Y}_t\biggr)\,\dd t + \sigma\,\dd W_t\,,
	\end{align*}
	where
	\begin{align*}
		N_t &= b\left[\begin{array}{cc} \nu^*_q(t) + \nu^*_{\bar{q}}(t)\,\quad, & \nu^*_V(t) + \nu^*_{\bar{V}}(t) \end{array}\right]\,.
	\end{align*}
	With the expression for $\bar{Y}_t$ given in \eqref{eqn:population_mean} the solution to the SDE can be written as
	\begin{align*}
		S_t &= S_0 + \int_0^t \biggl(\mu + ba_u + N_u\,\Phi_u\,\bar{Y}_0 + N_u\,\Phi_u \int_0^u \Phi^{-1}_s\,a_s\,ds\biggr) \dd u \\ & \quad
					+ \int_0^t\biggl(\rho\,\eta\,b\left[\begin{array}{cc} 1 & 0 \end{array}\right]\, \Phi_t\,\Phi^{-1}_u\,\left[\begin{array}{c} 0 \\ 1 \end{array}\right] + \sigma\biggr)\, \dd W_u\,,
	\end{align*}
	where $a_t$ is defined as in Proposition \ref{prop:pop_distribution}. This allows us to write
	\begin{align*}
		S_t-\mathds{E}[S_t] &= \int_0^t\biggl(\rho\,\eta\,b\left[\begin{array}{cc} 1 & 0 \end{array}\right]\, \Phi_t\,\Phi^{-1}_u\,\left[\begin{array}{c} 0 \\ 1 \end{array}\right] + \sigma\biggr)\, \dd W_u\,,
	\end{align*}
	and the result in \eqref{eqn:S_var} follows from Ito's isometry. The expression in \eqref{eqn:S_var_closed} arises again from using Proposition \ref{prop:separate_closed} to solve \eqref{eqn:dPhi}, which yields
	\begin{align*}
		\Phi_t &= \left[\begin{array}{cc} 1 & \frac{2\,z\,(1-\ee^{-\frac{b}{\kappa}t})}{b((1+2\,z)\ee^{\omega T} - 1)} \\ 0 & \frac{(1+2\,z)\,\ee^{\omega(T-t)}-1}{(1+2\,z)\,\ee^{\omega T}-1}\ee^{-\frac{b}{\kappa}t}\end{array}\right]\,, & \Phi^{-1}_t &= \left[\begin{array}{cc} 1 & \frac{-2\,z\,(1-\ee^{-\frac{b}{\kappa}t})\,\ee^{\frac{b}{\kappa}t}}{b\,((1+2\,z)\,\ee^{\omega(T-t)}-1)} \\ 0 & \frac{(1+2\,z)\,\ee^{\omega T}-1}{(1+2\,z)\,\ee^{\omega (T-t)}-1}\ee^{\frac{b}{\kappa}t}\end{array}\right]\,.
	\end{align*}
	Substituting these expressions into \eqref{eqn:S_var} yields \eqref{eqn:S_var_closed}.
\end{proof}

In Figure \ref{fig:mf_S_var} we plot the variance of $S_t$ through time when agents trade according to the mean-field equilibrium strategy. If there were no price impact then this variance would be purely from the accumulated volatility over time. With price impact, the drift of the midprice has an element of randomness caused by the common noise component of the agents' trade signal. Here we see that the effect on price variance depends on more than just the information shared by agents, as measured by $\rho^2$, but the sign of $\rho$ also matters. When $\rho^2$ is large, agents share a lot of information and trade in a similar fashion. When this happens with positive $\rho$, their order flow is concentrated and tends to occur in the same direction as midprice changes, effectively increasing the size of midprice changes and therefore variance. When $\rho$ is negative, their order flow is concentrated but tends to occur in the opposite direction of midprice changes, lowering the variance. When $\rho$ is close to zero, they share little information and net order flow tends to be close to zero which adds no additional variance to the midprice.

\begin{figure}
	\begin{center}
		{\includegraphics[trim=140 240 140 240, scale=0.5]{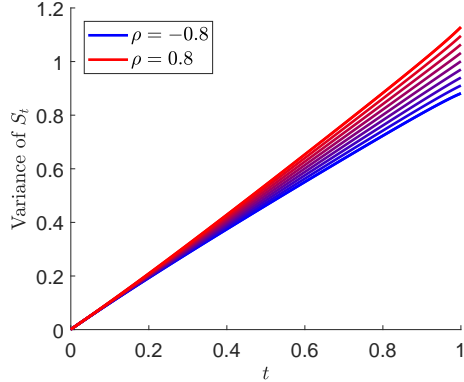}}
	\end{center}
	\vspace{-1em}
	\caption{Variances of midprice $S_t$ through time in mean-field limit for various $\rho$. Parameters used are $\mu = 0$, $\sigma = 1$, $\eta = 1$, $\beta = 1$, $\gamma = 0.05$, $\bar{\gamma}=0.1$, $b = 5\cdot10^{-2}$, $k = 5\cdot10^{-3}$, $\bar{k} = 10^{-3}$, $\alpha = 0.1$, $T = 1$. \label{fig:mf_S_var}}
\end{figure}

\section{Conclusion}\label{sec:conclusion}

In this paper we have presented a model for price dynamics and trading in which an agent attempts to extract profits from his own subjective valuation of an asset. When his subjective view of asset value is significantly different than the traded market price he wants to accumulate a large position, but friction effects and risk aversion prevent him from trading too quickly. Instead he manages a trade-off between the potential profits and costs. We continue our analysis when multiple agents are undertaking this task, either with a common trade signal shared between them or with individual signals correlated to each other. A mean-field game approach is taken to represent a setting with a large number of agents which keeps the problem tractable. This also allows us to study the cross sectional distribution of inventory as it depends on the correlation structure of the collection of signals. When correlation between signals is large, the inventory across all agents will have a tighter distribution because they are essentially trading off of the same information and therefore have similar behaviour. The correlation between signal and price innovations also modifies the asset price variance, as the random order flow will cause it to deviate from its accumulated volatility over time. Positive correlation between each signal and price innovations will increase the variance of the asset price at any fixed point in time.

\bibliographystyle{chicago}
\bibliography{References}

\begin{thebibliography}{}

\bibitem[\protect\citeauthoryear{Almgren and Chriss}{Almgren and
  Chriss}{2001}]{almgren2001optimal}
Almgren, R. and N.~Chriss (2001).
\newblock Optimal execution of portfolio transactions.
\newblock {\em Journal of Risk\/}~{\em 3}, 5--40.

\bibitem[\protect\citeauthoryear{Bigiotti and Navarra}{Bigiotti and
  Navarra}{2018}]{bigiotti2018optimizing}
Bigiotti, A. and A.~Navarra (2018).
\newblock Optimizing automated trading systems.
\newblock In {\em The 2018 International Conference on Digital Science}, pp.\
  254--261. Springer.

\bibitem[\protect\citeauthoryear{Cartea, Donnelly, and Jaimungal}{Cartea
  et~al.}{2020}]{cartea2018hedging}
Cartea, {\'A}., R.~Donnelly, and S.~Jaimungal (2020).
\newblock Hedging nontradable risks with transaction costs and price impact.
\newblock {\em Mathematical Finance\/}~{\em 30\/}(3), 833--868.

\bibitem[\protect\citeauthoryear{Cartea and Jaimungal}{Cartea and
  Jaimungal}{2016a}]{cartea2016closed}
Cartea, {\'A}. and S.~Jaimungal (2016a).
\newblock A closed-form execution strategy to target volume weighted average
  price.
\newblock {\em SIAM Journal on Financial Mathematics\/}~{\em 7\/}(1), 760--785.

\bibitem[\protect\citeauthoryear{Cartea and Jaimungal}{Cartea and
  Jaimungal}{2016b}]{cartea2016incorporating}
Cartea, {\'A}. and S.~Jaimungal (2016b).
\newblock Incorporating order-flow into optimal execution.
\newblock {\em Mathematics and Financial Economics\/}~{\em 10\/}(3), 339--364.

\bibitem[\protect\citeauthoryear{Casgrain and Jaimungal}{Casgrain and
  Jaimungal}{2018a}]{casgrain2018beliefs}
Casgrain, P. and S.~Jaimungal (2018a).
\newblock Mean-field games with differing beliefs for algorithmic trading.
\newblock {\em Mathematical Finance\/}.

\bibitem[\protect\citeauthoryear{Casgrain and Jaimungal}{Casgrain and
  Jaimungal}{2018b}]{casgrain2018mean}
Casgrain, P. and S.~Jaimungal (2018b).
\newblock Mean field games with partial information for algorithmic trading.
\newblock {\em arXiv preprint arXiv:1803.04094\/}.

\bibitem[\protect\citeauthoryear{Donnelly and Gan}{Donnelly and
  Gan}{2018}]{donnelly2018optimal}
Donnelly, R. and L.~Gan (2018).
\newblock Optimal decisions in a time priority queue.
\newblock {\em Applied Mathematical Finance\/}~{\em 25\/}(2), 107--147.

\bibitem[\protect\citeauthoryear{Huang, Jaimungal, and Nourian}{Huang
  et~al.}{2015}]{huang2015mean}
Huang, X., S.~Jaimungal, and M.~Nourian (2015).
\newblock Mean-field game strategies for optimal execution.
\newblock {\em Applied Mathematical Finance\/}.

\bibitem[\protect\citeauthoryear{Karatzas and Shreve}{Karatzas and
  Shreve}{2012}]{karatzas2012brownian}
Karatzas, I. and S.~Shreve (2012).
\newblock {\em Brownian motion and stochastic calculus}, Volume 113.
\newblock Springer Science \& Business Media.

\bibitem[\protect\citeauthoryear{Kaya, Schildbach, and Ag}{Kaya
  et~al.}{2016}]{kaya2016high}
Kaya, O., J.~Schildbach, and D.~B. Ag (2016).
\newblock High-frequency trading.
\newblock {\em Reaching the limits, Automated trader magazine\/}~{\em 41},
  23--27.

\bibitem[\protect\citeauthoryear{Kyle}{Kyle}{1985}]{kyle1985continuous}
Kyle, A.~S. (1985).
\newblock Continuous auctions and insider trading.
\newblock {\em Econometrica: Journal of the Econometric Society\/}, 1315--1335.

\bibitem[\protect\citeauthoryear{Lehalle and Mounjid}{Lehalle and
  Mounjid}{2017}]{lehalle2017limit}
Lehalle, C.-A. and O.~Mounjid (2017).
\newblock Limit order strategic placement with adverse selection risk and the
  role of latency.
\newblock {\em Market Microstructure and Liquidity\/}~{\em 3\/}(01), 1750009.

\end{thebibliography}

\end{document}